\documentclass[aps,pra,showpacs,amsmath,amssymb,superscriptaddress,reprint,10pt,longbibliography]{revtex4-1}
\usepackage{graphicx,mathrsfs,enumerate}
\usepackage[utf8]{inputenc}
\usepackage{mathtools}
\usepackage{amsfonts,amstext}
\usepackage{amsmath,amsthm,amssymb,mathbbol,tensor}
\usepackage{graphicx,xcolor,mathrsfs,enumerate}
\usepackage[hidelinks]{hyperref}\hypersetup{pdfpagemode=UseNone,pdfstartview=}
\usepackage[left=2cm,right=2cm]{geometry}
\usepackage{url}
\usepackage{times}%{amssymb,amsthm,amsfonts,amstext}
\usepackage{bbm}
\usepackage[english]{babel}
\usepackage[normalem]{ulem} %rabisgar comm o \sout
\allowdisplaybreaks

\newcommand{\bra}[1]{\langle #1|}
\newcommand{\bbra}[1]{\langle\!\langle #1|}
\newcommand{\ket}[1]{|#1\rangle}
\newcommand{\kett}[1]{|#1\rangle\!\rangle}

\newcommand{\ketbra}[2]{| #1 \rangle \langle #2 |}
\newcommand{\kettbbra}[2]{| #1 \rangle\!\rangle \langle\!\langle #2 |}

\DeclareMathOperator{\Tr}{Tr}

\newtheorem{thm}{Theorem}
\newtheorem{dfn}[thm]{Definition}

\newtheorem{cor}[thm]{Corollary}
\newtheorem{lem}[thm]{Lemma}

\newtheorem*{thm*}{Theorem}

\begin{document}
\title{Quantum superpositions of causal orders as an operational resource}
\author{M\'arcio M. Taddei}
\affiliation{Instituto de F\'isica, Universidade Federal do Rio de Janeiro, P. O. Box 68528, Rio de Janeiro, RJ 21941-972, Brazil}
\author{Ranieri V. Nery}
\affiliation{Instituto de F\'isica, Universidade Federal do Rio de Janeiro, P. O. Box 68528, Rio de Janeiro, RJ 21941-972, Brazil}
\author{Leandro Aolita}
\affiliation{Instituto de F\'isica, Universidade Federal do Rio de Janeiro, P. O. Box 68528, Rio de Janeiro, RJ 21941-972, Brazil}
%\affiliation{ICTP South American Institute for Fundamental Research,
%Instituto de F\'isica Te\'orica, UNESP-Universidade Estadual Paulista, R. Dr. Bento T. Ferraz 271, Bl. II, S\~ao Paulo, SP 01140-070, Brazil}

\begin{abstract}
Causal nonseparability refers to processes where events take place in a coherent superposition of different causal orders. These may be the key resource for experimental violations of causal inequalities and have been recently identified as resources for concrete information-theoretic tasks. Here, we take a step forward by deriving a complete operational framework for causal nonseparability as a resource. Our first contribution is a formal definition of quantum control of causal orders, a stronger form of causal nonseparability (with the celebrated quantum switch as best-known example) where the causal orders of events for a target system are coherently controlled by a control system. We then build a resource theory -- for both generic causal nonseparability and quantum control of causal orders -- with a physically-motivated class of free operations, based on process-matrix concatenations. We present the framework explicitly in the mindset with a control register. However, our machinery is versatile, being applicable also to scenarios with a target register alone. Moreover, an important subclass of our operations not only is free with respect to causal nonseparability and quantum control of causal orders but also preserves the very causal structure of causal processes. Hence, our treatment contains, as a built-in feature, the basis of a resource theory of quantum causal networks too. As applications, first, we establish a sufficient condition for pure-process free convertibility. This imposes a hierarchy of quantum control of causal orders with the quantum switch at the top. Second, we prove that causal-nonseparability distillation exists, i.e. we show how to convert multiple copies of a process with arbitrarily little causal nonseparability into fewer copies of a quantum switch. Our findings reveal conceptually new, unexpected phenomena, with both fundamental and practical implications.
\end{abstract}
%%%%%%%%%%%%%%%%%%%%%%%%%%%%%%%%%%%%%%%%%%%%%%%%%%%%%%%%%%%%%%%%%%%%%%%%%%%%%%%%%
%%%%%%%%%%%%%%%%%%%%%%%%%%%%%%%%%%%%%%%%%%%%%%%%%%%%%%%%%%%%%%%%%%%%%%%%%%%%%%%%%

%\pacs{05.70.Ln, 05.20.-y, 05.30.-d, 05.50.+q} 

\maketitle

%%%%%%%%%%%%%%%%%%%%%%%%%%%%%%%%%%%%%%%%%%
%%%%%%%%%%%%%%%%%%%%%%%%%%%%%%%%%%%%%%%%%%
%\emph{Introduction}---
\section{Introduction}
The study of physical processes with events without a predefined, fixed causal order is ultimately motivated by general relativity, whereby the dynamical distribution of energy has a bearing on whether events are time- or space-like separated. In fact, it has been conjectured \cite{Hardy2005,Hardy2007,Hardy2009} that quantum gravity may require a theory where a dynamical causal order between events plays an important role. In this context, quantum-mechanical effects on causal orders cannot be disregarded. For instance, this is particularly relevant when one considers the spacetime warping caused by spatial quantum superpositions of a massive body \cite{Zych2017}.

On a more down-to-earth plane,  processes with events in indefinite causal orders have sparked a great deal of interest in quantum information and foundations \cite{Chiribella2012,Chiribella2013}. From a fundamental point of view, they constitute an exotic class of  quantum operations, adding to the extensive list of counterintuitive properties of quantum theory. This class does not fit the usual quantum-computing paradigm of quantum circuits with fixed gates, and more general frameworks have been developed to encompass it \cite{Chiribella2008,Chiribella2009,Oreshkov2012,Araujo2015,Oreshkov2016,MacLean2017}, such as, e.g., the so-called \emph{process matrices} \cite{Oreshkov2012,Araujo2015,Oreshkov2016}.
In general, a process is called \emph{causally nonseparable} if it cannot be decomposed as a classical (i.e. probabilistic) mixture of \emph{causal processes} \cite{Oreshkov2012,Araujo2015,Oreshkov2016} (i.e. processes with a fixed causal order). 
These processes are fundamentally important since they are suspected to be the key resource for potential experimental violations of causal inequalities \cite{Oreshkov2012,Branciard2016}. A notable subclass of causally nonseparable processes is the one displaying \emph{quantum control of causal orders}, where a quantum system (the control) coherently controls the causal order with which events for another system (the target) take place. 
The best known example thereof is the celebrated \textit{quantum switch} \cite{Chiribella2012,Chiribella2013,Araujo2014,Procopio2014,Guerin2016}. 
The latter is special because it represents the only form of causal nonseparability so far known to be physical \cite{Araujo2015,Araujo2017}.
In turn, from an applied viewpoint, it has been recently shown to be a useful resource for a number of interesting information-processing tasks \cite{Chiribella2012,Araujo2014,Guerin2016,Ebler2018}. Moreover, it has already been subject of experimental investigations \cite{Procopio2014,Rubino2016,Goswami2018,Wei2018}. Curiously, even though quantum control of causal orders is the rule-of-thumb terminology evoked to discuss the quantum switch, a precise formal definition of this notion is -- to our knowledge -- still missing.
  
Here, we study quantum superpositions of causal orders from a resource-theoretic perspective. Resource theories provide powerful frameworks for the formal treatment of a physical property as an operational resource, adequate for its characterization, quantification, and manipulation \cite{Brandao2015, Coecke2016}. Their central component is a set of transformations -- called the \emph{free operations} of the theory -- that are unable to create the resource in question. We build a physically-meaningful class of free operations of both causal nonseparability and quantum control of causal orders. This requires a satisfactory rigorous definition of the latter notion, which we provide on the way.  
The proposed free operations are reminiscent in spirit to the free wirings of other types of quantum resources \cite{Gallego2012,Gallego2015,Gallego2017,Amaral2018}.
More precisely, they are given by concatenations of the input process with causally separable processes of two elementary kinds. Processes are mathematically represented by process matrices \cite{Oreshkov2012,Araujo2015,Oreshkov2016} and process concatenations by the so-called \emph{link product} \cite{Chiribella2009}.
Both elementary types of process concatenations are remarkably simple and, yet, they give rise to highly non-trivial effects.
First, they establish an ordering for a conceptually-interesting and experimentally-relevant subset of processes to which we refer as \emph{generalized quantum switches}. The ordering is mathematically captured  by a simple majorization condition sufficient for a pure process to be freely obtained from another.
As a corollary of the latter, it follows that any generalized quantum switch can be freely obtained from the quantum switch. 
This yields a hierarchy of quantum control of causal orders where the quantum switch sits at the top, thus giving it
 the status of \emph{basic unit} of this exclusive form of causal nonseparability. Second, we prove that, remarkably, it is possible to concentrate the causal nonseparability spread among multiple copies of non-maximally causally nonseparable processes (even those arbitrarily close to the causally separable ones) into a quantum switch. Hence, \emph{distillation of quantum control of causal orders} exists.
Our proof is constructive, with an explicit distillation protocol, so that a lower bound to the optimal concentration rate is obtained. 
Finally, we emphasize that our machinery is both highly versatile and notably unifying. On the one hand, it is explicitly formulated in the mindset with a control register but is also readily applicable to scenarios with a target system alone. On the other hand, one of the two elementary types of free operations mentioned leaves invariant not only both the sets of processes without causal nonseparability or quantum control of causal orders but also that of causal processes, for all underlying causal structure. Thus, our framework  also includes, as a built-in feature, the basis of an eventual resource theory of quantum causal networks.

The paper is organized as follows. In Sec. \ref{sec:prelim}, we introduce  preliminary concepts and notation. In Sec. \ref{sec:def_qcco}, we propose a formal definition of quantum control of causal orders. In Sec. \ref{sec:free_ops}, we introduce our operational framework with the free operations. In Sec. \ref{sec:conversion}, we study single-shot conversions, a hierarchy, and units of quantum control of causal orders. In Sec. \ref{sec:distillation}, we show that distillation is possible. Finally, Sec. \ref{sec:conclusions} is devoted to our conclusions.

%%%%%%%%%%%%%%%%%%%%%%%%%%%%%%%%%%%%%%%%%%
%%%%%%%%%%%%%%%%%%%%%%%%%%%%%%%%%%%%%%%%%%
\section{Preliminaries}
\label{sec:prelim}
We consider physical processes in the scenario outlined in Fig.~\ref{fig:basic}a.
A convenient tool to describe such processes is the process-matrix formalism \cite{Oreshkov2012,Araujo2015,Oreshkov2016}, which extends the quantum combs formalism \cite{Chiribella2009}, both in turn based on the Choi-Jamio\l kowski (CJ) isomorphism \cite{Choi1975,Bengtsson2006}. 
For any Hilbert space $\mathsf H$, we denote by $\mathsf B(\mathsf H)$ the space of bounded-trace, linear operators on $\mathsf H$. The CJ isomorphism allows one \cite{Chiribella2009,Araujo2015} to represent any completely-positive trace-preserving linear map $\mathcal{E}:\mathsf B(\mathsf H_I)\rightarrow\mathsf B(\mathsf H_O)$ from arbitrary input to output spaces $\mathsf B(\mathsf H_I)$ and $\mathsf B(\mathsf H_O)$, respectively,
as the \emph{CJ state}
\begin{equation}
E=  \left(\mathcal I \otimes\mathcal{E}\right) \left(\kett{\mathbb1}\bbra{\mathbb1}\right) \ .
\label{eq:defCJ}
\end{equation}
%
%%%%%%%%%%%%%%%%%%%%%%%%%%%%%%%%%%%%%%%%%
%%%%%%%%%%%%% Fig 1 %%%%%%%%%%%%%%%%%%%%%
\begin{figure}[t!]
\centering
\raisebox{-0.5\height}{\includegraphics[height=.17\textheight]{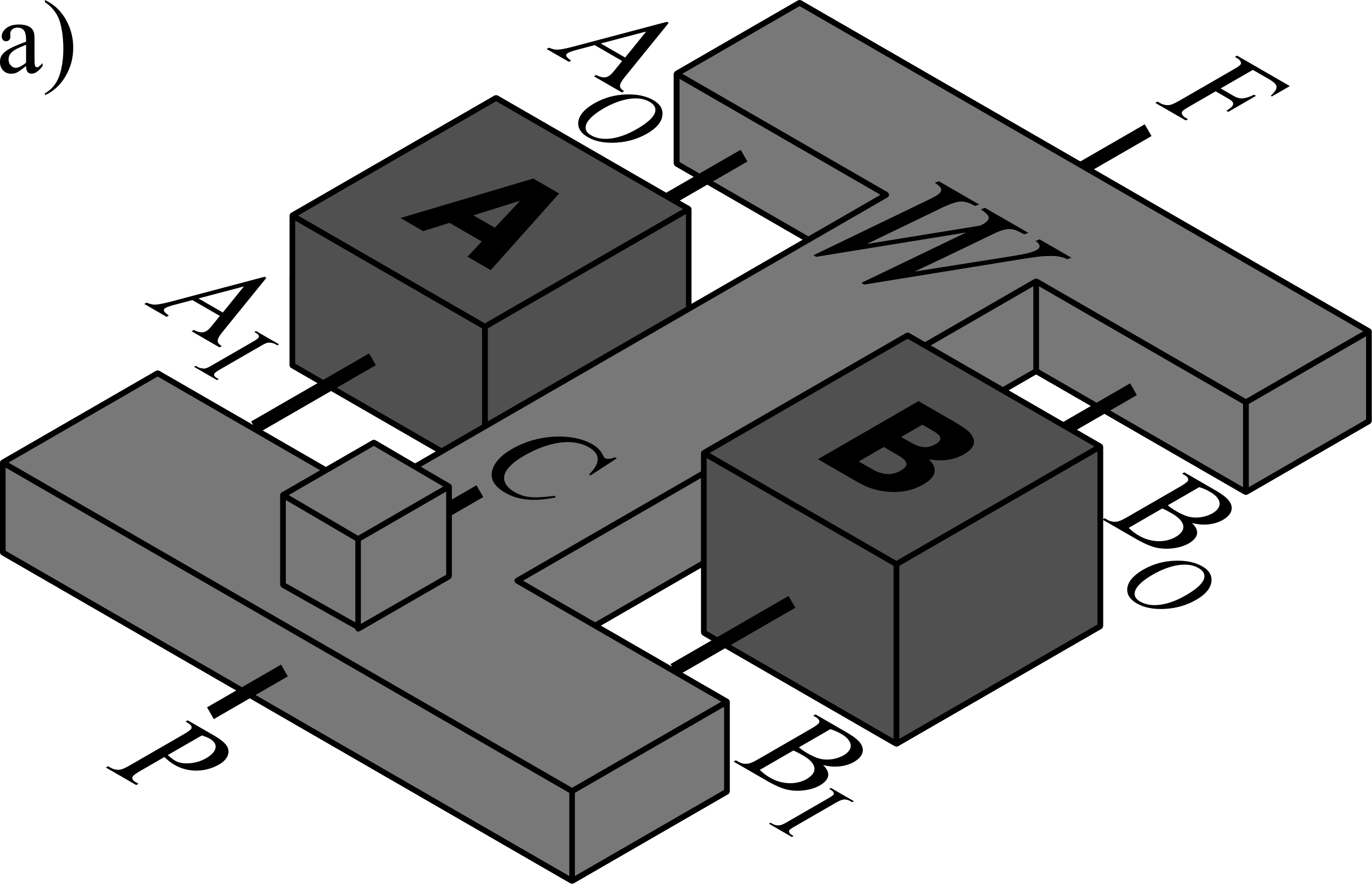}} \hspace{.3cm}
\raisebox{-0.5\height}{\includegraphics[height=.2\textheight]{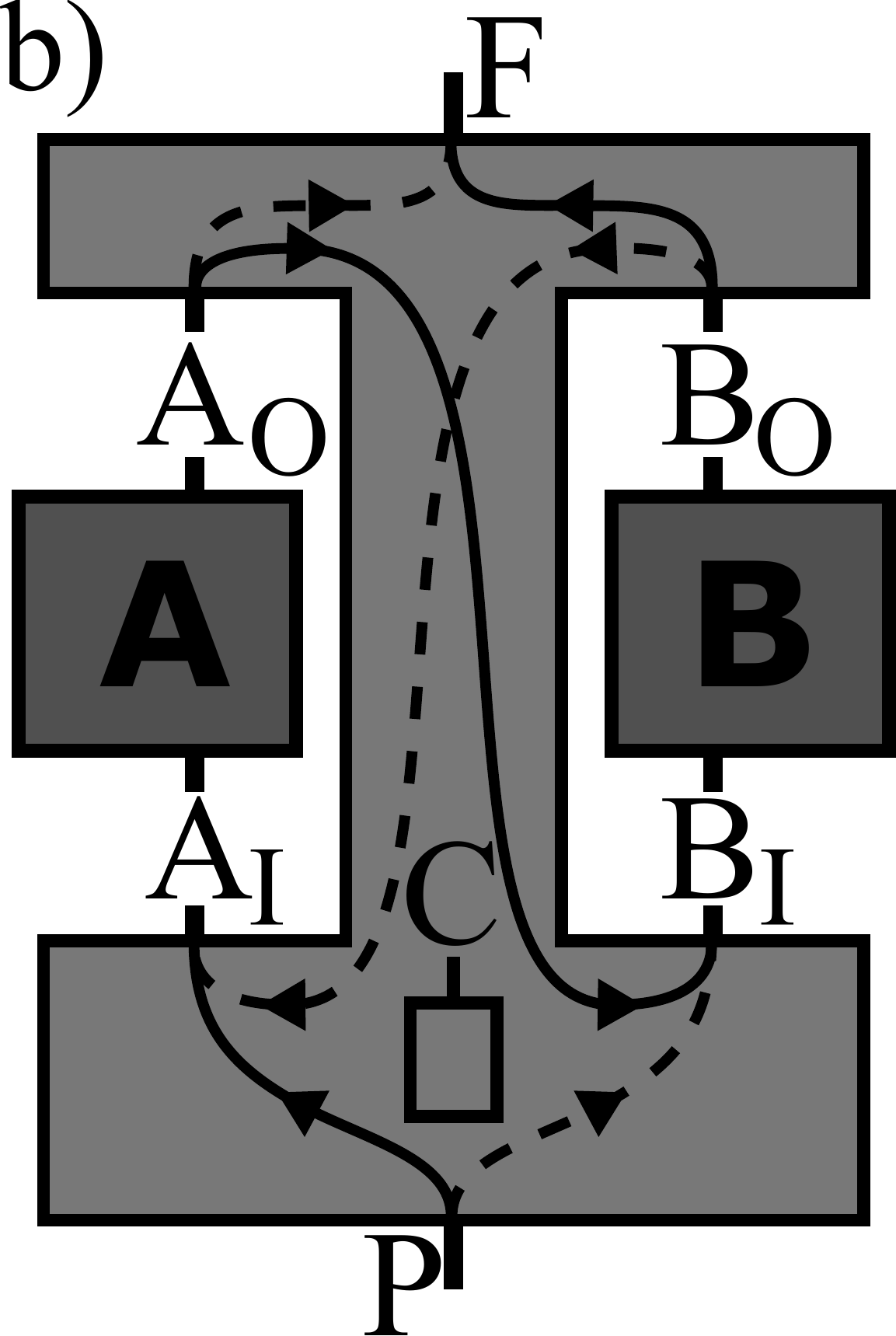}}
\caption{ Schematics of a process. a) Two users, Alice and Bob, in labs $A$ and $B$, respectively, receive a \emph{target qudit} as local input and subsequently send out an equivalent system as output after some operation. We label Alice's input (output) $A_I$ ($A_O$), and Bob's $B_I$ ($B_O$).
The events inside each lab occur in a definite order: $A_I$ ($B_I$) happens before $A_O$ ($B_O$). However, no causal order between $A$ and $B$ is assumed. For practical purposes, but without loss of generality, we do nevertheless assume a past and a future to both $A$ and $B$. These are realized by labs $P$ and $F$, whose only function is to input into and output from the process the initial and final  states of the target qudit, respectively. Finally, there is yet another lab, $C$, operated by a third user, Charlie, who gets as local input a \emph{control qubit}. This can control the causal order between  $A$ and $B$. 
The process, denoted by $W$, then specifies the entire system history -- preparations, evolutions, measurements, etc. -- outside the labs (light gray).
The operations inside each lab, in turn, are described by \emph{instruments} (dark-gray).
 b) Pictorial representation of processes with definite causal orders $A\to B$ (solid) and $B\to A$ (dashed). Coherent superpositions of the latter give \emph{causally nonseparable} processes. If, in addition, such superpositions involve entanglement with the control qubit, the process can feature \emph{quantum control of causal orders} (see Fig. \ref{fig:geometry} for precise definition).}
\label{fig:basic}
\end{figure}
%%%%%%%%%%%%% End of Fig 1 %%%%%%%%%%%%%%%
%%%%%%%%%%%%%%%%%%%%%%%%%%%%%%%%%%%%%%%%%%
%
Here $\kett{\mathbb1}:=\sum_{j}\ket j\otimes\ket j$ is a (non-normalized) maximally entangled state on $\mathsf H_I\otimes\mathsf H_{I'}$, where $\mathsf H_{I'}$ is a space isomorphic to (i.e. a copy of) $\mathsf H_I$, with $\{\ket j\}_j$ an orthonormal basis of $\mathsf H_I$. 
In turn, $\mathcal I:\mathsf B(\mathsf H_{I'})\rightarrow\mathsf B(\mathsf H_{I'})$ is the identity map on the copy space.
Complete-positivity of $\mathcal{E}$ implies, by virtue of Choi's theorem \cite{Choi1975}, that $E$ is positive semi-definite. 
Thus, $E$ is technically equivalent to a (non-normalized) state on the extended space $\mathsf B(\mathsf H_{I'}\otimes\mathsf H_O)$. 
Whenever there is no risk of ambiguity we omit (in a slight abuse of notation) the apostrophe that distinguishes copy from system spaces. For instance, we sometimes write $E\in\mathsf B(\mathsf H_{I}\otimes\mathsf H_O)$ instead of $E\in\mathsf B(\mathsf H_{I'}\otimes\mathsf H_O)$. In addition, for (non-normalized) maximally entangled states between a system and a copy of it we directly omit the copy subindex. That is, we use the short-hand notation $\kett{\mathbb1}_{I}$ to denote $\kett{\mathbb1}_{I I'}$.

In turn, the composition $\mathcal{D}\circ\mathcal{E}$ of $\mathcal{E}$ with another map $\mathcal{D}:\mathsf B(\mathsf H_{O})\rightarrow\mathsf B(\mathsf{H}_{\tilde{O}})$
 is given in the CJ representation by the \emph{link product}  \cite{Chiribella2009}, denoted here by ``$*$". More precisely, if $D\in\mathsf B(\mathsf H_{O}\otimes\mathsf{H}_{\tilde{O}})$ is the CJ state of $\mathcal{D}$, then the CJ state of $\mathcal{D}\circ\mathcal{E}$ is $D*E\in\mathsf B(\mathsf H_{I}\otimes\mathsf{H}_{\tilde{O}})$, defined by 
\begin{equation}
\label{eq:def_link_prod}
D*E \coloneqq \Tr_{O}\! \left[\left( 				E							\!\otimes\!	\mathcal I_{\tilde{O}}	\right)
															\left( \mathcal I_{I}		\!\otimes\!			D^{T_{O}}				\right) \right] \ .
\end{equation}
Here, $\mathcal I_{\tilde{O}}$ and $\mathcal I_{I}$ are respectively the identity maps on $\mathsf{H}_{\tilde{O}}$ and  $\mathsf H_{I}$, $\Tr_{O}$ denotes the partial trace over $\mathsf H_{O}$, and $T_{O}$ the partial transpose over $\mathsf H_{O}$ (in the chosen basis $\{\ket j\}_j$).

Process matrices generalize the notion of CJ states to encapsulate state preparations, operations, and measurements all in a unified description. 
In our setting, Fig. \ref{fig:basic}a), they can be defined by all CJ states that, upon composition with any arbitrary instruments at $A$ and $B$ (including instruments exploiting entangled ancillas between the labs), yield a CJ state on the remaining labs that describes a valid completely-positive (CP) trace-preserving (TP) channel from $\mathsf B(\mathsf H_{P})$ to $\mathsf B(\mathsf H_{F}\otimes\mathsf H_{C})$ \cite{Oreshkov2012,Araujo2015,Oreshkov2016}. This corresponds to CJ states $W\in\mathsf B(\mathsf H_{P}\otimes\mathsf H_{A_O}\otimes\mathsf H_{B_O}\otimes\mathsf H_{F}\otimes\mathsf H_{A_I}\otimes\mathsf H_{B_I}\otimes\mathsf H_{C})$ that must be positive semi-definite and satisfy a few normalization constraints (given in App. \ref{app:norm_cond_W}).
If, in addition, a process $W$ has rank 1, it decomposes as $W= \kett{\mathbb w}\bbra{\mathbb w}$, with $\kett{\mathbb w}\in\mathsf H_{P}\otimes\mathsf H_{A_O}\otimes\mathsf H_{B_O}\otimes\mathsf H_{F}\otimes\mathsf H_{A_I}\otimes\mathsf H_{B_I}\otimes\mathsf H_{C}$ the corresponding \emph{pure CJ state vector}. In that case we refer to $W$ as a \emph{pure process} \cite{Araujo2015} and denote it simply by $\kett{\mathbb w}$. We denote the set of generic process matrices  for the scenario in question $\mathsf P\subset\mathsf B(\mathsf H_{P}\otimes\mathsf H_{A_O}\otimes\mathsf H_{B_O}\otimes\mathsf H_{F}\otimes\mathsf H_{A_I}\otimes\mathsf H_{B_I}\otimes\mathsf H_{C})$.

Two basic examples are shown in Fig. \ref{fig:basic} b). The first one (solid line) represents processes of the type
\begin{equation}
\ket0_C\kett{\mathbb 1_0}\coloneqq\ket0_C\kett{\mathbb1}_{PA_I}\kett{\mathbb1}_{A_OB_I}\kett{\mathbb1}_{B_OF} \ ,
\label{eq:def_id0}
\end{equation}
where the subindices in the right-hand side indicate the space supporting each ket. 
The target-system process $\kett{\mathbb 1_0}$ defines a quantum causal model \cite{Costa2016a,Allen2017} with causal structure $P\to A\to B\to F$.
More precisely, the composite-system process in Eq. \eqref{eq:def_id0} describes the situation where Charlie receives the control qubit state $\ket{0}_C$ and the target qudit is directed from $P$ to Alice, who (after applying her instrument) in turn sends it to Bob, who (after his intervention) finally forwards it towards the final target-system output at $F$. 
The second process (dashed line) defines a quantum causal model with causal structure $P\to B\to A\to F$ for the target system and gives lab $C$ a different local input: 
\begin{equation}
\ket1_C\kett{\mathbb 1_1}\coloneqq\ket1_C\kett{\mathbb1}_{PB_I}\kett{\mathbb1}_{B_OA_I}\kett{\mathbb1}_{A_OF} \ .
\label{eq:def_id1}
\end{equation}
That is, Charlie now receives the orthogonal state $\ket{1}_C$ while the target now goes from $P$ to $B$, then to $A$, and finally to $F$. 

Clearly, $\kett{\mathbb 1_0}$ and $\kett{\mathbb 1_1}$ display fixed causal orders between $A$ and $B$: $A\to B$ and $B\to A$, respectively. They are thus particular instances of \emph{causal processes}. The causal relations between the different labs are captured by the signaling constraints of the process \cite{Araujo2015}. Namely, a process  $W_{A\to B}$ is compatible with a causal order $A\to B$ iff it is \emph{nonsignaling} from $B$ to $A$, i.e. if it cannot be used to send information from Bob's output $B_{O}$ to Alice's input $A_{I}$ (see App. \ref{app:norm_cond_W} for the explicit  definition); and analogously for  $W_{B\to A}$.
In addition, we demand that processes are compatible with the orders $P\to (A,B)$ and $(A,B)\to F$. 
 That is, $P$ and $F$ are respectively taken as the global past and future of the target system (see App. \ref{app:norm_cond_W}).
In turn, a process is said to be \emph{causally separable} \cite{Oreshkov2012,Araujo2015,Oreshkov2016} if it can be decomposed as a probabilistic mixture of causal processes
\begin{equation}
W_{\rm cs} = p\,W_{A\to B}+(1-p)\,W_{B\to A}, 
\label{eq:caussep}
\end{equation}
with $0\leq p\leq 1$. 
We denote by $\mathsf{CS}\subset\mathsf P$ the set of all causally separable process for our scenario. Any $W\in\mathsf P\setminus\mathsf{CS}$ is called \emph{causally nonseparable}.

Causal nonseparability is known to appear in processes that can violate causal inequalities \cite{Oreshkov2012, FAB16,Branciard2016}. These processes involve coherent superpositions of causal orders on the target system alone, i.e. with $C$ playing no role in the causal nonseparability. However, it is not clear whether such processes admit a physical realization \cite{Araujo2017}.
A conceptually different form of causal nonseparability, called \emph{quantum control of causal orders}, takes place when the superposition involves entanglement with $C$. The \emph{quantum switch} \cite{Chiribella2012,Chiribella2013}
\begin{equation}
\kett{\mathbb{w}_{\rm qs}}\coloneqq\frac{\ket0_C\kett{\mathbb 1_0}+\ket1_C\kett{\mathbb 1_1}}{\sqrt{2}}\
\label{eq:def_qswitch}
\end{equation}
is the paradigmatic example thereof. There, $C$ coherently controls the causal order in which the target qudit passes through $A$ and $B$. 
Quantum control of causal orders constitutes a stronger form of causal nonseparbility in the sense of requiring not only coherence but also entanglement. Interestingly, in addition, it admits clear physical interpretations in terms of interferometers  \cite{Procopio2014,Rubino2016,Goswami2018,Wei2018}. Somewhat surprisingly though, even though the terminology quantum control of causal orders appears quite frequently in the literature, a precise formal definition of this notion has -- to our knowledge -- not been provided yet.
We propose one next.
 
%%%%%%%%%%%%%%%%%%%%%
%%%%%%%%%%%%%%%%%%%%%%%%%%%%%%%%%%%%%%%%%%
%%%%%%%%%%%%%%%%%%%%%%%%%%%%%%%%%%%%%%%%%%

\section{Definition of quantum control of causal orders}
\label{sec:def_qcco}
While generic causal nonseparability is a rigorously defined concept, the specific notion of quantum control of causal orders has so far been -- surprisingly -- only colloquially introduced. 
Here, we need a precise mathematical definition of this notion. 
We begin by formalizing the notion of entanglement for processes. This is done in the obvious way, in analogy to entanglement for states \cite{Horodecki2009}. First, for a tripartite process $W_{ABC}\in\mathsf B(\mathsf H_C\otimes\mathsf H_{A_O}\otimes\mathsf H_{B_O}\otimes\mathsf H_{A_I}\otimes\mathsf H_{B_I})$ (without past and future labs), we define $W_{ABC}$ to be \emph{separable between control and target} if it belongs to the convex hull of \emph{product processes} in that bipartition, i.e. if 
\begin{equation}
W_{ABC}= \sum_{\mu} q_{j}\, \varrho^{(j)}_C\otimes W^{(j)}_{AB}, 
\label{eq:def_sep}
\end{equation}
with $\{q_{j}\}_{j}$ an arbitrary probability distribution over $\mu$, $\varrho^{(j)}_C\in\mathsf B(\mathsf H_C)$ an arbitrary state of the control, and $W^{(j)}_{AB}\in\mathsf B(\mathsf H_{A_O}\otimes\mathsf H_{B_O}\otimes\mathsf H_{A_I}\otimes\mathsf H_{B_I})$ an arbitrary process (causally separable or not) for Alice and Bob's labs alone. 
Then, we define a five-partite process $W\in\mathsf P$ (with past and future labs) to be \emph{separable between the control and the indefinite labs} if its reduced process over $A$, $B$, and $C$, given by its partial trace $\Tr_{PF}[W]$ over $P$ and $F$, is separable between control and target. We denote by $\mathsf{S}\subset\mathsf P$ the set of all processes separable between the control and the indefinite labs. In turn, any $W\in\mathsf P\setminus\mathsf{S}$ is \emph{entangled between the control and the indefinite labs}. What is more, here we refer for short to  separability or entanglement between the control and the indefinite labs simply as \emph{separability or entanglement}, respectively. 

The reason why our definition of entanglement focuses on the reductions over $A$, $B$ and $C$ is to isolate  the entanglement between the control and exclusively the target labs that can admit indefinite causal orders. Recall that the past and future labs have a fixed causal order. In fact, there exist processes in $\mathsf P$ that are entangled over $B(\mathsf H_C\otimes\mathsf H_{P}\otimes\mathsf H_{F})$ but separable over $\mathsf B(\mathsf H_C\otimes\mathsf H_{A_O}\otimes\mathsf H_{B_O}\otimes\mathsf H_{A_I}\otimes\mathsf H_{B_I})$. Such processes clearly cannot contain quantum control of causal orders. Hence, we exclude them as entangled, for if we did not Def. \ref{def:QCCO} below would assign them quantum control of causal orders. Moreover, it is often the case that the target is initialized in a fixed state and subject to a fixed instrument (e.g., traced out) at the end, being therefore readily given by tripartite processes on $\mathsf B(\mathsf H_C\otimes\mathsf H_{A_O}\otimes\mathsf H_{B_O}\otimes\mathsf H_{A_I}\otimes\mathsf H_{B_I})$ \cite{Araujo2014,Araujo2015,Procopio2014,Guerin2016,Rubino2016,Goswami2018,Wei2018}. Our definition of entanglement directly applies there too (because there are no target labs other than the indefinite ones). Still, entanglement turns out to be necessary but not sufficient for quantum control of causal orders. 

Consider for instance the process $\kett{\mathbb{w}_{\rm ent}}=\left(\ket0_C\kett{\mathbb 1_0}+\ket1_C\kett{\mathbb u_{AB}}\right)/\sqrt{2}$, where $\kett{\mathbb u_{AB}}\coloneqq\kett{\mathbb 1}_{PA_I}\kett{\mathbb u_{AB}}_{A_OB_I}\kett{\mathbb 1}_{B_OF}$ is a causal process analogous to $\kett{\mathbb 1_0}$ but with an arbitrary unitary gate $\mathbb u_{AB}\neq\mathbb1$ from $A$ to $B$. 
This can be physically implemented by a quantum circuit with definite causal order $A\to B$ and controlled unitary gates.
Process $\kett{\mathbb{w}_{\rm ent}}$ is pure and entangled, thus featuring quantum control of unitary gates between $A$ and $B$. Nevertheless, since both $\kett{\mathbb 1_0}$ and $\kett{{\mathbb u_{AB}}}$ have causal order $A\to B$, no control of causal orders takes place. In fact, $\kett{\mathbb{w}_{\rm ent}}$ is itself a causal process. The following is a satisfactory definition that rules out such cases.
\begin{dfn}[Quantum control of causal orders]
A process $W\in\mathsf P$ has quantum control of causal orders, or, equivalently, is \emph{quantum-control causally ordered}, if it is outside the convex hull $\mathrm{Conv}\left(\mathsf{CS}\cup\mathsf{S}\right)$ of the sets $\mathsf{CS}$ and $\mathsf{S}$ of causally-separable and separable processes, respectively.
\label{def:QCCO}
\end{dfn}
\noindent In turn, any $W\in\mathrm{Conv}\left(\mathsf{CS}\cup\mathsf{S}\right)$ is not  quantum-control causally ordered. We denote by $\mathsf{NQC}\coloneqq\mathrm{Conv}\left(\mathsf{CS}\cup\mathsf{S}\right)$ the set with no quantum control of causal orders. See Fig. \ref{fig:geometry}.
%%%%%%%%%%%%%%%%%%%
\begin{figure}[t]
\includegraphics[width=.8\linewidth]{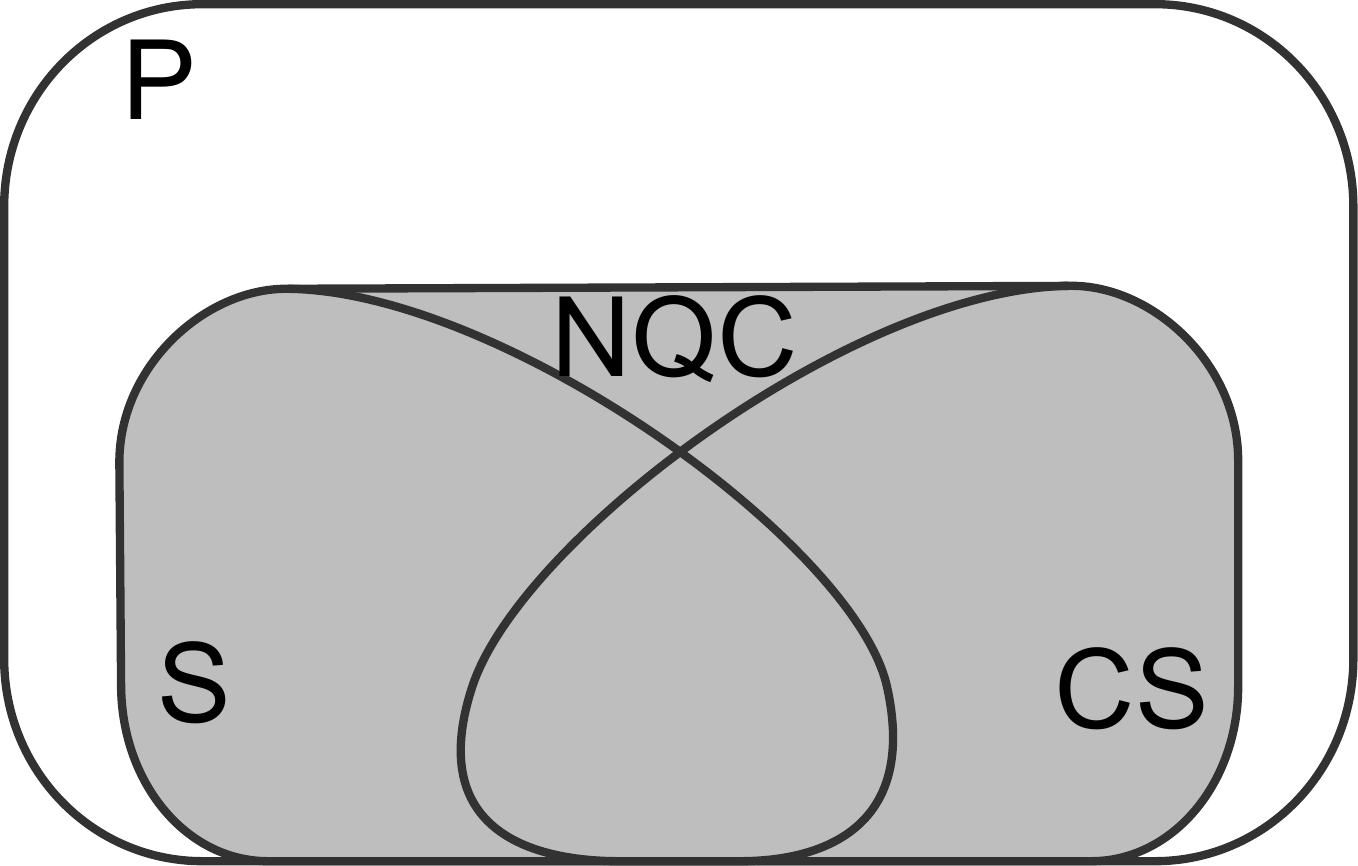} 
\caption{Pictorial representation of the internal geometry of the set $\mathsf{P}$ of all processes. $\mathsf{CS}$ and $\mathsf{S}$ are the subsets of causally-separable and separable processes, respectively. The processes outside $\mathsf{CS}$ are causally nonseparable, whereas those outside $\mathsf{S}$ are entangled between control and target.
The convex hull of $\mathsf{CS}$ and $\mathsf{S}$ gives the set $\mathsf{NQC}$ (gray) with no quantum control of causal orders. This includes processes (given by convex combinations of elements in  $\mathsf{CS}$ and $\mathsf{S}$) that are  causally nonseparable and entangled. We define the processes outside $\mathsf{NQC}$ to have quantum control of causal orders.
}
\label{fig:geometry}%
\end{figure}
%%%%%%%%%%%%%%%%%%%%%%%%%%%%%%%%%%%%%%%%%%%

Def. \ref{def:QCCO} excludes the convex hull of $\mathsf{CS}$ and $\mathsf{S}$ (instead of just their union) because convex mixing describes a purely classical operation. In other words, any process that can be operationally generated by probabilistically choosing one out of two resourceless processes must also be resourceless. (Otherwise, probabilistically choosing would not be a free operation of quantum control of causal orders.) This is reminiscent of the definition of genuinely multipartite entanglement, where multi-partite states entangled in each and all of the system bipartitions but within the convex hull of the bi-separable states are also excluded as genuinely multipartite entangled (see, e.g., Refs. \cite{Jungnitsch2011,Aolita2015}). 

Notable examples of $\mathsf{P}\setminus\mathsf{NQC}$ are all pure processes
\begin{align}          
\kett{\mathbb w} =
 \sqrt{p_0}\,\ket{\Phi_0}_C\kett{\mathbb u_0}+ \sqrt{p_1}\,\ket{\Phi_1}_C\kett{\mathbb u_1}\ ,
\label{eq:generalizedQS}
\end{align}
with $\{\ket{\Phi_0},\ket{\Phi_1}\}$ any orthonormal basis of $\mathsf H_{C}$, $\boldsymbol p\coloneqq\{p_0,p_1\}$ any binary probability distribution, and 
$\kett{\mathbb u_0}\coloneqq\kett{\mathbb u_{PA}}_{PA_I}\kett{\mathbb u_{AB}}_{A_OB_I}\kett{\mathbb u_{BF}}_{B_OF}$ and $\kett{\mathbb u_1}\coloneqq\kett{\mathbb u_{PB}}_{PB_I}\kett{\mathbb u_{BA}}_{B_OA_I}\kett{\mathbb u_{AF}}_{A_OF}$ causal processes analogous to $\kett{\mathbb 1_0}$ and $\kett{\mathbb 1_1}$, respectively, but with arbitrary unitary gates $\mathbb u_{PA}$, $\mathbb u_{AB}$, $\mathbb u_{BF}$, $\mathbb u_{PB}$, $\mathbb u_{BA}$, and $\mathbb u_{AF}$ instead of $\mathbb1$. That is, $\kett{\mathbb u_0}$ and $\kett{\mathbb u_1}$ have opposite definite causal orders, similarly to $\kett{\mathbb 1_0}$ and $\kett{\mathbb 1_1}$, but with channels other than the identity. These processes capture the most pristine form of causal nonseparability. In fact, for $p_0=\frac12=p_1$, they can be physically realized by applying on $\kett{\mathbb{w}_{\rm qs}}$ local (i.e. single-lab) unitary transformations on $C$ and non-local (i.e. multi-lab) controlled unitary gates on the target controlled by $C$. A particular interesting subset of the processes in Eq. \eqref{eq:generalizedQS} is that where the six unitary channels are not arbitrary but satisfy
the following constraints
\begin{equation}
\mathbb u_{PA}^\dagger\mathbb u_{BA}\mathbb u_{BF}^\dagger=\mathbb1=\mathbb u_{PB}^\dagger\mathbb u_{AB}\mathbb u_{AF}^\dagger.
\label{eq:constraints}
\end{equation}
Remarkably, this condition turns out to  characterize, for $p_0=\frac12=p_1$, the subset of processes 
 that are local-unitary equivalent to $\kett{\mathbb{w}_{\rm qs}}$  (see App. \ref{app:proof_conv} for details). We refer to all processes (for arbitary $\boldsymbol p$) satisfying both Eqs. \eqref{eq:generalizedQS} and \eqref{eq:constraints} as \emph{generalized quantum switches}. These are experimentally friendlier than the general processes in Eq. \eqref{eq:generalizedQS} with unconstrained unitaries and will be crucial in Secs. \ref{sec:conversion} and \ref{sec:distillation}.

%%%%%%%%%%%%%%%%%%
\begin{figure*}
\includegraphics[height=.29\textheight]{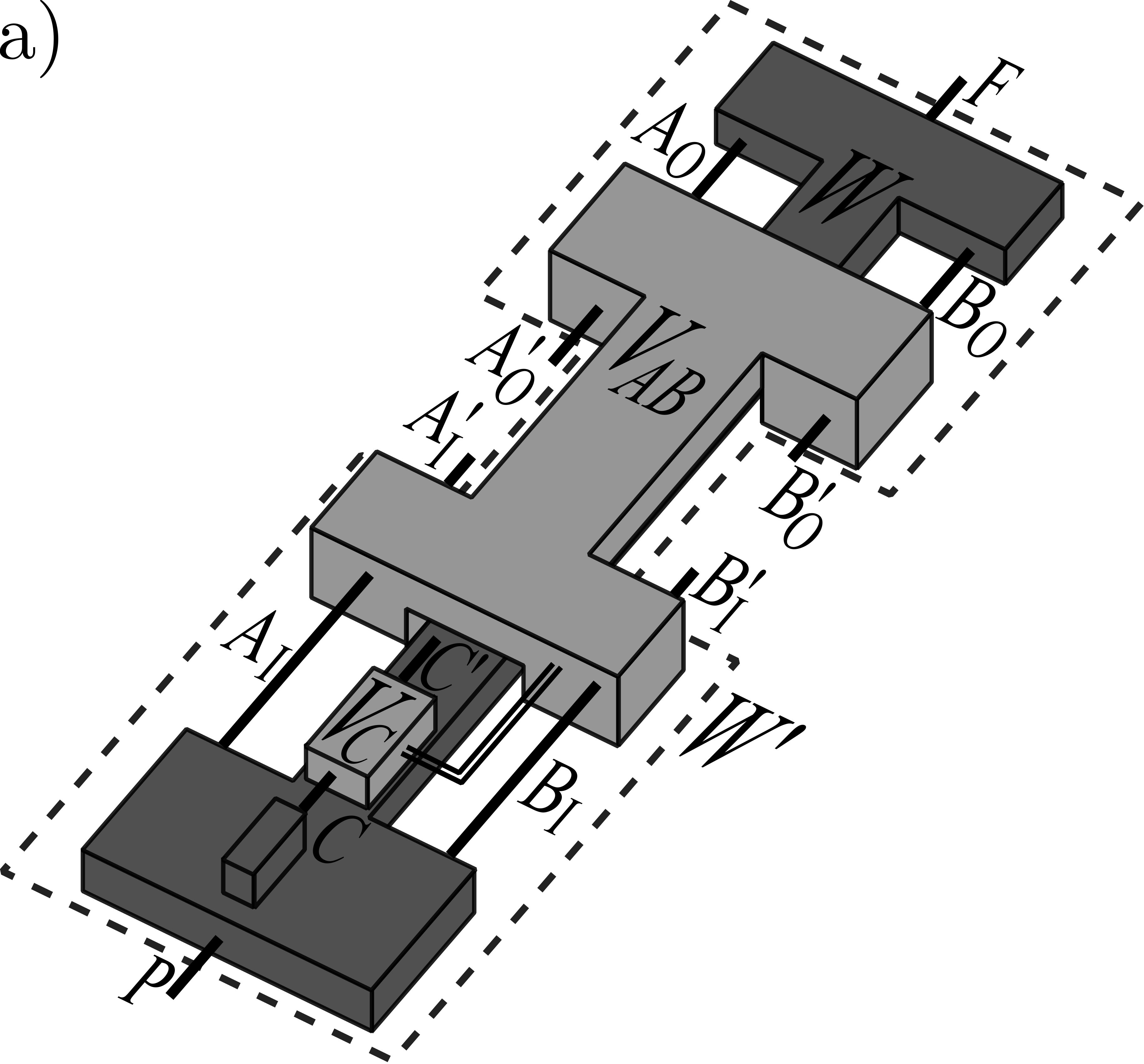} \hspace{5mm} %
\includegraphics[height=.29\textheight]{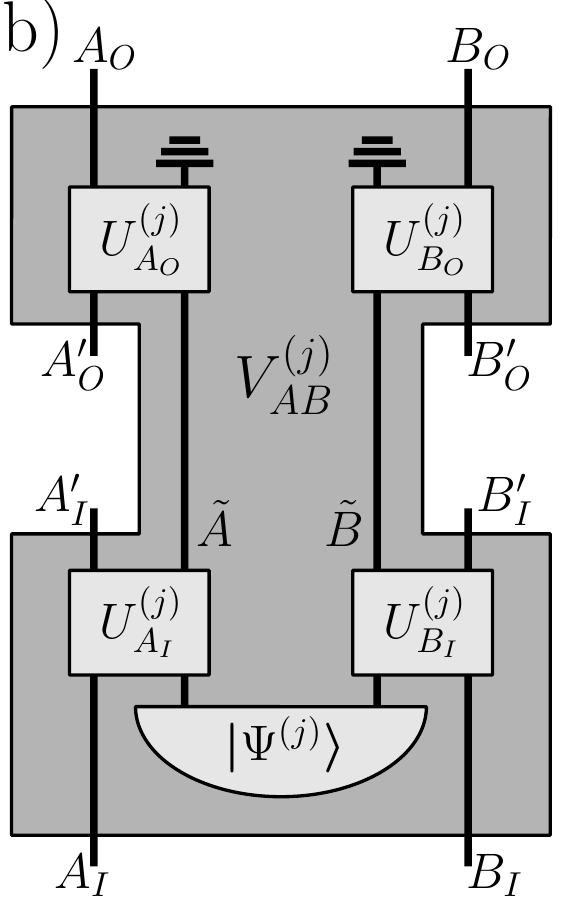} \hspace{3mm}
\includegraphics[height=.29\textheight]{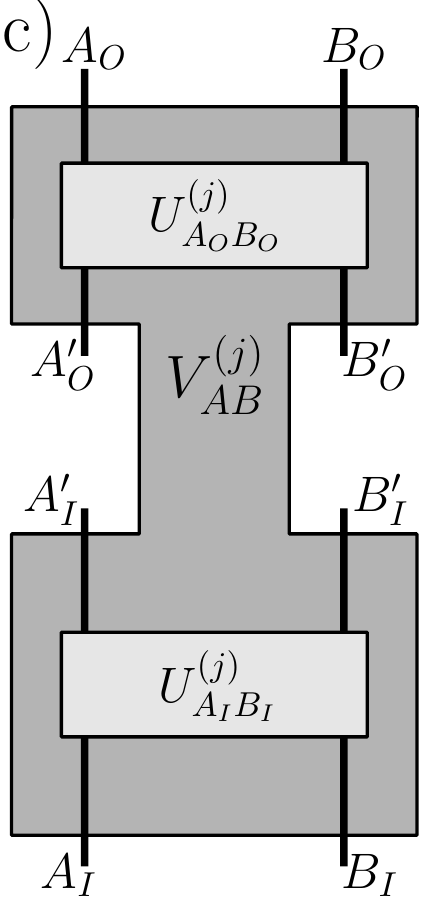}
\caption{Schematics of the operational framework of causal nonseparability and quantum control of causal orders. a) An initial process $W$ (dark gray) is concatenated -- i.e. link-product multiplied -- with the average process $V$ (light gray) of an instrument at Alice, Bob and Charlie's labs, whereas the past and future labs are untouched. The final process is thus $\mathcal{V}(W)= V*W$ (dashed outline). This has the same configuration of labs as the initial process: $P$, $A'$, $B'$, $C'$, and $F$, where $A'$, $B'$, and $C'$ have the same structure of input and output systems as the initial process' labs $A$, $B$, and $C$. 
The instrument's inputs are $A_I$, $A'_O$, $B_I$, $B'_O$, and $C$; and its outputs are $A'_I$, $A_O$, $B'_I$, $B_O$, and $C'$.
All the correlations in the control-versus-target bipartition that $V$ produces are due to 1-way classical communication (double line) of the local instrument outcomes from the control lab of Charlie to the target labs of Alice and Bob. That is, $V$ is separable in the bipartition and, therefore, creates no entanglement between control and target systems. Furthermore, for each outcome $j$ at $C$, the outcome process $V^{(j)}_{AB}$ on the target is designed to contain no causal nonseparability either. 
 For this reason, whenever the initial process is causally separable, so is the final one. 
 All this, together with linearity of the link product, implies also that whenever the initial process is not quantum-control causally ordered, neither is the final one.
 b) and c) Structure of the instruments on the target register. Each process $V^{(j)}_{AB}$ has six labs. Four of them correspond to the inputs $A_I$ and $B_I$ and outputs $A_O$ and $B_O$ of the initial process. While the other two correspond to Alice and Bob's final labs, $A'$ and $B'$, both equipped with input and output systems ($A'_I$ and $B'_I$ and $A'_O$ and $B'_O$, respectively). We consider two elementary types of instruments. 
b) The first one gives rise to the class of local operations and ancillary entanglement (LOAE). There, each $V^{(j)}_{AB}$ describes a local (in the $A|B$ bipartition) unitary evolution  on the instrument inputs $A_I$, $A'_O$, $B_I$, and $B'_O$ together with (arbitrary-dimensional) ancillary registers $\tilde A$ and $\tilde B$ in a (possibly entangled) state $\ket{\Psi^{(j)}}_{\tilde A \tilde B}$, followed by the disposal of the ancilas. 
c) In the second one, called probabilistic lab swaps (PLS), 
each $V^{(j)}_{AB}$ is a pure process where the same unitary operator (either the swap $S$ or the identity $\mathbb1$ gate) is applied to $A_I$ together with $B_I$ and to $A'_I$ together with $B'_I$. That is, conditioned on Charlie's outcome, Alice and Bob either exchange their systems (through swap gates on their inputs and outputs) or leave them untouched. This probabilistically exchanges the causal orders $A\to B$ and $B\to A$, but it never introduces causal nonseparability or quantum control of causal orders.
 }
\label{fig:concatenation}%
\end{figure*}
%%%%%%%%%%%%%%%%%%%%%%%%%%%%%%%%%%%%%%%%%%
%%%%%%%%%%%%%%%%%%%%%%%%%%%%%%%%%%%%%%%%%%

\section{The operational framework}
\label{sec:free_ops}
The fundamental property of the free operations of a resource theory is that of mapping the subset of resourceless objects of the theory onto itself. Here we consider linear transformations $\mathcal{V}:\mathsf P\to\mathsf P$ such that
\begin{equation}
\mathcal{V}(W)
\begin{cases}
\in\mathsf{CS}\ \ \ \ \ \text{if}\ \, W\in\mathsf{CS}, \\
\in\mathsf{NQC}\, \ \text{if}\ \, W\in\mathsf{NQC}.
\end{cases}
\end{equation}
In other words, we demand that the operations are free with respect to both causal nonseparability and quantum control of causal orders. This may in general be too restrictive if one is only interested in a resource theory of quantum control of causal orders alone. 
In the end of the section, we mention some subtleties towards such a theory though. 
In any case, here we are interested in 
a unified resource theory for both types of resources.

In concrete terms, we propose the following general parametrization for the elementary free operations:
\begin{equation}
\mathcal{V}(W)= V*W\ .
\label{eq:def_V}
\end{equation}
where $V$ is a (well-normalized) process matrix in $\mathsf B(\mathsf H_{A_O}\otimes\mathsf H_{A'_O}\otimes\mathsf H_{B_O}\otimes\mathsf H_{B'_O}\otimes\mathsf H_{A_I}\otimes\mathsf H_{A'_I}
\otimes\mathsf H_{B_I}\otimes\mathsf H_{B'_I}\otimes\mathsf H_{C}\otimes\mathsf H_{C'})$. 
Here, we again explicitly distinguish isomorphic spaces with an apostrophe because the link product in Eq. \eqref{eq:def_V} requires careful space matching.
In fact, using Eq. \eqref{eq:def_link_prod}, note that $\mathcal{V}$ effectively represents a CP TP map from $\mathsf B(\mathsf H_{A_I}\otimes\mathsf H_{A'_O}\otimes\mathsf H_{B_I}\otimes\mathsf H_{B'_O}\otimes\mathsf H_{C})$ to $\mathsf B(\mathsf H_{A'_I}\otimes\mathsf H_{A_O}\otimes\mathsf H_{B'_I}\otimes\mathsf H_{B_O}\otimes\mathsf H_{C'})$, acting trivially on $P$ and $F$ [see Fig. \ref{fig:concatenation} a)]. 
The explicit form of $V$ is taken as 
\begin{equation}
V = \sum_j V^{(j)}_C \otimes V^{(j)}_{AB} \ ,
\label{eq:Vdecomposed}
\end{equation}
with $V^{(j)}_C \in \mathsf B(\mathsf H_{C}\otimes\mathsf H_{C'})$ sub-normalized process matrices (each one representing a CP non-TP map) that sum up to the normalized process matrix $\sum_j V^{(j)}_C$ (representing a CP TP map) and  $V^{(j)}_{AB} \in \mathsf B(\mathsf H_{A_O}\otimes\mathsf H_{A'_O}\otimes\mathsf H_{B_O}\otimes\mathsf H_{B'_O}\otimes\mathsf H_{A_I}\otimes\mathsf H_{A'_I}\otimes\mathsf H_{B_I}\otimes\mathsf H_{B'_I})$ normalized process matrices (each one representing a CP TP map). 

More technically, each term in Eq. \eqref{eq:Vdecomposed} represents the $j$-th outcome of an \emph{instrument} at Charlie's lab coordinated with a different process at Alice and Bob's labs. 
Instruments generalize the notion of positive operator-valued measures (POVMs) from measurements to state transformations \cite{Davies1970}. They reduce to POVMs when the output space has dimension 1.
The above-mentioned coordination is achieved through classical communication of Charlie's outcome $j$ to Alice and Bob. Thus, all the free operations arising from Eq. \eqref{eq:Vdecomposed} belong to the generic class of local operations and one-way classical communication from the control to the target and are therefore separable in the control-versus-target bipartition. Charlie's instrument can be arbitrary. 
However, we demand that all instruments at $A$ and $B$ satisfy the following basic constraints to avoid introducing causal loops: labs $A_I$ and $B_I$ are jointly in the causal past of $A'$ and $B'$, and all latter four are jointly in the causal past of $A_O$ and $B_O$. This is mathematically captured by the essential requirement that $A'$ and $B'$ both cannot signal from their local outputs $A'_O$ or $B'_O$ to neither of their inputs $A'_I$ or $B'_I$.  
This, together with the fact that $V$ is separable automatically implies that $\mathcal{V}$ preserves the set $\mathsf S$ of separable processes  in the control-versus-target bipartition. Next, we impose more fine-tuned conditions on the instruments at the target labs so that $\mathcal{V}$ preserves also the set $\mathsf{CS}$ of causally separable processes. Because of linearity, this will automatically imply preservation of $\mathsf{NQC}$ too. 

Specifically, we consider two broad families of elementary instruments. The first one arises from restricting all $V^{(j)}_{AB}$ to local quantum operations in the $A|B$ bipartition assisted by pre-shared entanglement. More precisely, we take each $V^{(j)}_{AB}$ in Eq. \eqref{eq:Vdecomposed} as a process resulting from local unitary dynamics of the instrument inputs $A_I$, $A'_O$, $B_I$, and $B'_O$ together with arbitrary-dimensional ancillary registers $\tilde A$ and $\tilde B$, which are subsequently discarded [see Fig. \ref{fig:concatenation} b)] \footnote{One has in principle three different isomorphic ancilla spaces on Alice's side, one before the application of $U_{A_I}^{(j)}$, one between the application of $U_{A_I}^{(j)}$ and $U_{A_O}^{(j)}$, and one after $U_{A_O}^{(j)}$. Unlike the other variables, however, this ancilla is not matched to any Hilbert space outside $V_{AB}^{(j)}$, so we can use a single variable $\tilde A$ for all of them. Same for $\tilde B$.}. The ancillas are initialized in an arbitrary pure (normalized) state $\ket{\Psi^{(j)}}_{\tilde A\tilde B}\in\mathsf H_{\tilde A}\otimes\mathsf H_{\tilde B}$. The evolution is in turn given by local unitary operators $U^{(j)}_{A_I}$ and $U^{(j)}_{B_I}$, from $\mathsf H_{A_I}\otimes\mathsf H_{\tilde A}$ to $\mathsf H_{A'_I}\otimes\mathsf H_{\tilde A}$ and from $\mathsf H_{B_I}\otimes\mathsf H_{\tilde B}$ to $\mathsf H_{B'_I}\otimes\mathsf H_{\tilde B}$, respectively, and $U^{(j)}_{A_O}$ and $U^{(j)}_{B_O}$, from $\mathsf H_{A'_O}\otimes\mathsf H_{\tilde A}$ to $\mathsf H_{A_O}\otimes\mathsf H_{\tilde A}$ and from $\mathsf H_{B'_O}\otimes\mathsf H_{\tilde B}$ to $\mathsf H_{B_O}\otimes\mathsf H_{\tilde B}$, respectively. Finally, after the unitary evolution, both ancillary registers are traced out.
We refer to the resulting class as \emph{local operations and ancillary entanglement} (LOAE); and denote it by $\mathsf{LOAE}$:
\begin{dfn}[Local operations and ancillary entanglement]
A process transformation $\mathcal{V}$ is in the class $\mathsf{LOAE}$ if it can be parametrized by Eqs. \eqref{eq:def_V} and \eqref{eq:Vdecomposed} with a process $V$ such that, for all $j$, 
\label{def:LOAE}
\begin{equation}
\label{eq:def_LOAE}
V^{(j)}_{AB}\coloneqq \Tr_{\tilde A\tilde B}\!\left[U^{(j)}\ \ketbra{\Psi^{(j)}}{\Psi^{(j)}}_{\tilde A\tilde B}\otimes\kettbbra{\mathbb1}{\mathbb1}_{\rm in} \ {U^{(j)}}^{\dagger}\right],
\end{equation}
with the short-hand notations $U^{(j)}\coloneqq \big(U^{(j)}_{A_O}\otimes\mathbb1_{A_I'}\big)\big(\mathbb1_{A_O'}\otimes U^{(j)}_{A_I}\big)\otimes \left(U^{(j)}_{B_O}\otimes\mathbb1_{B_I'}\right)\left(\mathbb1_{B_O'}\otimes U^{(j)}_{B_I}\right)\otimes\mathbb1_{\rm copy}$ and $\kett{\mathbb1}_{\rm in}\coloneqq\kett{\mathbb1}_{A_I}\otimes\kett{\mathbb1}_{B_I}\otimes\kett{\mathbb1}_{A'_O}\otimes\kett{\mathbb1}_{B'_O}\in\mathsf H_{A_I}^{\otimes 2}\otimes\mathsf H_{B_I}^{\otimes 2}\otimes\mathsf H_{A'_O}^{\otimes 2}\otimes\mathsf H_{B'_O}^{\otimes 2}$, 
where  $\mathbb1_{\rm copy}$ is the identity operator on the inputs' copy's Hilbert space $\mathsf H_{A_I}\otimes\mathsf H_{B_I}\otimes\mathsf H_{A'_O}\otimes\mathsf H_{B'_O}$. 
\end{dfn}

We emphasize that only pre-shared quantum correlations between Alice and Bob are allowed in $\mathsf{LOAE}$, with no communication of any sort between them. 
Thus, clearly, each $V^{(j)}_{AB}$ (and therefore also $V$) is a nonsignaling process with respect to the $A|B$ bipartition, i.e. nonsignaling both from $A$ to $B$ and vice versa (see Lemma \ref{th:LOAEinNSO} in App. \ref{app:cond_free_V} for an explicit proof). 
Explicitly, no information can flow from $A_I$ to $B'_I$ or $B_O$, from $B_I$ to $A'_I$ or $A_O$, from $A'_O$ to $B_O$, and from $B'_O$ to $A_O$.
In particular, this excludes the possibility of teleporting the incoming state of any of the instrument's inputs towards the other side of the bipartition. Finally, note that Def. \ref{def:LOAE} does not impose any restriction on the dimension or structure of the ancillary spaces $\mathsf H_{\tilde A}$ and $\mathsf H_{\tilde B}$. Therefore, by virtue of Stinespring's dilation theorem \cite{Stinespring1955}, Eq. \eqref{eq:def_LOAE} effectively parametrizes a quantum process describing an arbitrary, fully-generic CP TP map without signaling between Alice and Bob and subject to the above-mentioned essential local-causality requirement that $A'_O$ and $B'_O$ cannot signal to $A'_I$ and $B'_I$, respectively. In fact, arbitrary-dimensional ancilas are actually not required for the latter to hold, just $\mathrm{dim}(\mathsf H_{\tilde A})=2\, \mathrm{dim}(\mathsf H_{A_I}\otimes \mathsf H_{A'_O})$ and $\mathrm{dim}(\mathsf H_{\tilde B})=2\, \mathrm{dim}(\mathsf H_{B_I}\otimes \mathsf H_{B'_O})$.

The second family of elementary instruments we consider is called \emph{probabilistic lab swaps} (PLS), denoted by $\mathsf{PLS}$. It is simpler than the class $\mathsf{LOAE}$ in that each $V^{(j)}_{AB}$ in Eq. \eqref{eq:Vdecomposed} is a pure process, describing the same unitary transformation applied from $\mathsf H_{A_I}\otimes\mathsf H_{B_I}$ to $\mathsf H_{A'_I}\otimes\mathsf H_{B'_I}$ and from $\mathsf H_{A'_O}\otimes\mathsf H_{B'_O}$ to $\mathsf H_{A_O}\otimes\mathsf H_{B_O}$. No ancillary registers are used here. Moreover, we allow for only two such unitary operations: the swap gate $S$ and the identity gate $\mathbb1$.
That is, each process $V^{(j)}_{AB}$ describes either the joint swap of both inputs and outputs, which effectively exchanges Alice and Bob's labs, or the trivial identity map:
\begin{dfn}[Probabilistic lab swaps]
A process transformation $\mathcal{V}$ is in the class $\mathsf{PLS}$ if it can be parametrized by Eqs. \eqref{eq:def_V} and \eqref{eq:Vdecomposed} with a process $V$ such that, for all $j$, $V^{(j)}_{AB}$ is a rank-1 process given by either the identity $\kettbbra{\mathbb 1_{AB}}{\mathbb 1_{AB}}$ or the lab-swap $\kettbbra{\mathbb s_{AB}}{\mathbb s_{AB}}$ processes, defined as
\label{def:ILS}
\begin{subequations}
\label{eq:def_ILS}
\begin{equation}
\label{eq:def_id}
\kett{\mathbb 1_{AB}}\coloneqq\kett{\mathbb1}_{A_IA'_I}\otimes\kett{\mathbb1}_{B_IB'_I}\otimes\kett{\mathbb1}_{A'_OA_O}\otimes\kett{\mathbb1}_{B'_OB_O}
\end{equation}
and
\begin{equation}
\label{eq:def_swap}
\kett{\mathbb s_{AB}}\coloneqq\kett{\mathbb1}_{A_IB'_I}\otimes\kett{\mathbb1}_{B_IA'_I}\otimes\kett{\mathbb1}_{A'_OB_O}\otimes\kett{\mathbb1}_{B'_OA_O}.
\end{equation}
\end{subequations}
\end{dfn}

Importantly, due to the swap gates, $\mathsf{PLS}$ is not only nonlocal but also even signaling in the $A|B$ bipartition, in contrast to $\mathsf{LOAE}$. In fact, for a causal initial process $W$, i.e. with a fixed causal order $A\to B$ or $B\to A$, the causal orders are probabilistically exchanged. However, this never creates causal nonseparability because such exchanges are incoherent. Coherence between the 
different $j$-th terms in Eq. \eqref{eq:Vdecomposed} would be required so that Eqs. \eqref{eq:def_ILS} can lead to a non-free operation able to create causal nonseparability. 
Finally, a comment on the experimental feasibility of $\mathsf{PLS}$ is in place. Even though mathematically formulated in terms of joint swaps of both inputs and outputs, process transformations in $\mathsf{PLS}$ can in many cases  be simulated without any swap gate. More precisely, in experiments, processes are often detected through local instruments at Alice and Bob's lab  \cite{Procopio2014,Rubino2016,Goswami2018,Wei2018}. Thus, in those cases, instead of actually applying the joint swap gates to their initial process $W$ and detecting their final process $\mathcal{V}(W)$ with local instruments in some given settings, Alice and Bob can simply do nothing to $W$ and swap the settings of their local instruments. That is, the lab swap on the process can be absorbed into the instruments' settings on the final process. This considerably alleviates physical implementations of PLS processes.

The validity of the elementary classes $\mathsf{LOAE}$ and $\mathsf{PLS}$ as free operations of causal nonseparability is formalized by the following theorem. We take advantage of the theorem also to formally introduce the complete class of free operations we propose:  local operations and one-way classical communication from the control to the target given by arbitrary sequential concatenations of transformations in $\mathsf{LOAE}$ or $\mathsf{PLS}$.
\begin{thm}[Free operations of causal nonseparability and quantum control of causal orders]
Any process transformation $\mathcal{V}$ in $\mathsf{LOAE}\cup\mathsf{PLS}$ is an automorphism of the sets $\mathsf{P}$ of generic processes, $\mathsf{CS}$ of causally separable ones, and $\mathsf{NQC}$ of non quantum-control causally ordered ones. Therefore, so is any sequence of such elementary transformations.
\label{th:free_ops}
\end{thm}

The theorem is proven in App. \ref{app:cond_free_V}. In fact, there we actually prove a stronger result, where $\mathsf{LOAE}$ is replaced by the more general class $\mathsf{NSO}$ of \emph{nonsignaling operations}. In the latter, instead of pre-shared entangled ancillas, Alice and Bob may be assisted by generic (potentially supra-quantum) nonsignaling resources. What is more, our proof strategy to show that $\mathsf{CS}$ is closed under $\mathsf{NSO}$ is to show that even its subsets of causal processes with definite causal orders are preserved as well by $\mathsf{NSO}$. Recall that $\mathsf{CS}$ is the convex hull of the latter subsets, so that, by linearity, the implication automatically follows. In other words, we show that any process transformation in $\mathsf{NSO}$ (and, by inclusion, also  in $\mathsf{LOAE}$) maps an arbitrary causal process, with order either $A\to B$ or $B\to A$, into a causal process with the same order. That is, it preserves the underlying causal structure of every quantum causal model. Although this is explicitly proven here for quantum causal models that are effectively bipartite (involving Alice and Bob's labs), it can be straightforwardly generalized to arbitrary causal networks with more labs. 
As such, $\mathsf{LOAE}$ provides the basis of a yet-to-be resource theory of quantum causal networks, where the resourceful set consists of all quantum causal models incompatible with a given multipartite causal structure under scrutiny. 
This is a promising and exciting prospect, but it is beyond the scope of this work. Still, the unifying power of the elementary class $\mathsf{LOAE}$ could not be left unmentioned here. 
In the next two sections, we exploit simple examples of our two elementary classes of free operations to implement highly nontrivial information-theoretic manipulations of causal nonseparability. 

Finally, we briefly comment on the possibility of a resource theory of just quantum control of causal orders (and not causal nonseparability). In principle, the condition that $\mathsf{CS}$ is closed under the transformations is an unnecessary restriction to that end. However, physically-meaningful relaxations of Defs. \ref{def:LOAE} and \ref{def:ILS} so that $\mathsf{P}$ and $\mathsf{NQC}$ are invariant but not $\mathsf{CS}$ have been elusive to us. For instance, one could relax the constraint that each $V^{(j)}_{AB}$ in Eq. \eqref{eq:Vdecomposed} does not create causal nonseparability, so that -- say -- $\kett{\mathbb u_0}$ goes to $(\kett{\mathbb u_0}+\kett{\mathbb u_1})/\sqrt2$ for some $\mathbb u_{PA}$, $\mathbb u_{AB}$, and $\mathbb u_{BF}$. The latter corresponds to a coherent lab swap without a control system \footnote{{Strictly speaking, it does not define a valid process matrix, as it does not satisfy the necessary normalization conditions. However, it does define a  physical process that can be implemented  (probabilistically, with a probability depending on the instruments on the target) by post-selecting $\kett{\mathbb{w}_{\rm qs}}$ on a local measurement on $C$.}}.
However, the same transformation would then map pure entangled processes as $(\ket{0}_C\,\kett{\mathbb 1_0}+\ket{1}_C\,\kett{\mathbb u_0})/\sqrt2\in\mathsf{CS}$
 out of $\mathsf{NQC}$. 
Alternatively, one could even relax the constraint that the instruments are separable in the control-versus-target bipartition [i.e. the tensor-product decomposition of Eq. \eqref{eq:Vdecomposed}], so that -- say -- pure causally-separable processes in $\mathsf{S}$ are mapped into $\mathsf{NQC}\setminus\mathsf{S}$ (the set with quantum control of different processes without quantum control of causal orders). The instruments on the target would then be applied coherently with that on the control, instead of conditioned on its classical outcomes. However, similarly to the example above, one can then find pure causally-nonseparable processes in $\mathsf{S}$ that would be taken out of $\mathsf{NQC}$ by the same transformations. 
We leave the questions of resource theories of quantum control of causal orders that do not preserve $\mathsf{CS}$ or $\mathsf{S}$ open.

%%%%%%%%%%%%%%%%%%%%%%%%%%%%%%%%%%%%%%%%%%
%%%%%%%%%%%%%%%%%%%%%%%%%%%%%%%%%%%%%%%%%%
\section{Single-copy conversions and a hierarchy of quantum control of causal orders}
\label{sec:conversion}

Here we study free interconversions between processes in the regime where a single copy of the system is available. (In the next section we study transformations in the multi-copy regime.) 
More precisely, we consider deterministic conversions between generalized quantum switches, i.e. between any $\kett{\mathbb w}$ and $\kett{\mathbb w'}$ obeying Eq. \eqref{eq:generalizedQS}.
We characterize the allowed conversions in terms of a majorization relation between the corresponding  distributions of $\kett{\mathbb w}$ and $\kett{\mathbb w'}$, respectively  denoted by $\boldsymbol p$ and $\boldsymbol p'$. For binary distributions, majorization is defined in a particularly simple way:  $\boldsymbol p$ is majorized by $\boldsymbol p'$ (denoted $\boldsymbol p\preccurlyeq\boldsymbol p'$) if $\max_i p_i\leq \max_i p'_i$. In other words, $\boldsymbol p\preccurlyeq\boldsymbol p'$ if $\boldsymbol p$ is more flat than $\boldsymbol p'$. The characterization is formalized as follows. 
\begin{thm}[Single-copy pure-process conversion]
\label{th:conversion} 
Let $\kett{\mathbb w}$ and $\kett{\mathbb w'}$ be generalized quantum switches such that $\boldsymbol p \preccurlyeq \boldsymbol p'$. Then there is a free operation that converts  $\kett{\mathbb w}$ into $\kett{\mathbb w'}$ with unit probability.
\end{thm}
\noindent The proof is given App. \ref{app:proof_conv}, where we construct an explicit protocol that does the claimed transformation. 

Theo. \ref{th:conversion} plays a role for quantum control of causal orders similar to the one played in entanglement theory by Nielsen's seminal theorem \cite{Nielsen1999} (see also \cite{Du2015}) for pure-state conversions under entanglement-free operations. It induces a hierarchy -- more precisely, a so-called total preorder -- on the set of pure processes obeying Eq. \eqref{eq:generalizedQS}. It is called a preorder because there are cases where $\kett{\mathbb w}\preccurlyeq\kett{\mathbb w'}$ and $\kett{\mathbb w'}\preccurlyeq\kett{\mathbb w}$ with $\kett{\mathbb w'}\neq\kett{\mathbb w}$, so that both processes can be reversibly interconverted. This is for instance the case when $\kett{\mathbb w}$ and $\kett{\mathbb w'}$ are local-unitary equivalent. In turn, such  preorder is called total because, for binary probability distributions, there exists no pair of distributions such that none majorizes the other. That is, Theo. \ref{th:conversion} leaves no pair of generalized quantum switches unconnected. 

Hierarchies of this kind are important because they substantiate with a clear operational interpretation the notions of  ``more'' and ``less'' quantum control of causal orders: If a process can be deterministically transformed freely into another then the former is not less quantum-control causally ordered than the latter. The theorem thus lays the basis of formal quantifiers through causal nonseparability monotones, in the same spirit as entanglement monotones \cite{Horodecki2009}. Interestingly, at the top of the hierarchy lies the quantum switch.
\begin{cor}[Partial unit of quantum control of causal orders]\label{th:bit}
Let $\kett{\mathbb w}$ be an arbitrary process given by Eq. \eqref{eq:generalizedQS}. Then there is a free operation that converts  $\kett{\mathbb{w}_{\rm qs}}$ into $\kett{\mathbb w}$ with unit probability.
\end{cor}

The corollary follows from the fact that $\boldsymbol p_{\rm qs}\coloneqq\{\frac12,\frac12\}$ is majorized by all distributions. The basic unit of a resource is important because it renders the notion of maximal amount of resource operationally meaningful, independently of the particular choice of quantifier. For instance, a process would be \emph{the unit of quantum control of causal orders} if all processes could be freely obtained deterministically from it. This would be the counterpart of Bell states in entanglement theory, which are used as entanglement bits \cite{Dur2000,Horodecki2009}. Here, we use the terminology \emph{partial unit of quantum control of causal orders} to stress that the quantum switch is the basic unit only within the subset of generalized quantum switches. An interesting possibility would be that $\kett{\mathbb w_{\rm qs}}$ can be freely converted probabilistically into all processes (be it exactly or approximately, up to arbitrarily small error). This would render $\kett{\mathbb w_{\rm qs}}$ a full unit of causal nonseparability in an operational sense. In fact, invoking again entanglement theory, GHZ states are considered more entangled than W ones precisely in that sense  \cite{Vrana2015,Walter2016}. Alternatively, it may as well be the case that there are intrinsically-inequivalent classes of causal nonseparability, even under free operations beyond the ones proposed here. These are fascinating open questions that our framework offers for future explorations.

%%%%%%%%%%%%%%%%%%%%%%%%%%%%%%%%%%%%%%%%%%
%%%%%%%%%%%%%%%%%%%%%%%%%%%%%%%%%%%%%%%%%%
\section{Distillation of quantum control of causal orders}
\label{sec:distillation}
We now study the concentration of the quantum control of causal orders contained in multiple copies of a process (with non-maximal resource) into partial units of the resource, i.e. into (fewer) copies of the quantum switch. This is similar in spirit to the notion of entanglement distillation \cite{Bennett1995,Bennett1996,Bennett1996a,Horodecki2009}. Before we proceed, however, a brief digression on the composition of independent copies of a process is useful. 

In general, the tensor product of two (or more) valid process matrices on a given system is known not to yield a valid process matrix on the system copies \cite{Jia2018,Guerin2018}. The conceptual reason behind this is that, in the generic situation where Alice and Bob can apply arbitrary instruments globally on the copies of their subsystems, the tensor product of two processes that do not have the same definite causal order renders causal loops possible \cite{Jia2018}. This is an expected and reasonable impossibility if a process is used to describe space-time structures \cite{Feix2017,Zych2017,ZychPhD}, as it is difficult to conceive that Alice and Bob could share two copies of spacetime. However, for processes describing, e.g., interferometric experiments \cite{Procopio2014,Rubino2016}, it is perfectly admissible to describe two independent setups with the tensor product of two process matrices, so long as one restricts the type of instruments on the system copies \cite{Jia2018,Guerin2018}. 
In fact, this is the most natural description to adopt for experiments. Because, since each lab corresponds to a local space-time region, certain configurations of global instruments turn out to be unphysical. For instance, for the above-mentioned processes without the same definite causal order, implementing the instruments that would induce the causal loops requires signaling from one subsystem copy into another towards the past within the same lab \cite{Guerin2018}. 

Since our focus is operational, we adopt this description here. Indeed, any transformation given by Eq. \eqref{eq:def_V} can be thought of as ``adding elements to an experimental setup", as Fig. \ref{fig:concatenation} suggests. Hence, we represent copies of a process with tensor products of it and restrict to independent instruments on each system copy, described in turn by tensor products of single-system instruments. This rules out inter-copy signaling, which guarantees non-negative and well-normalized instrument-outcome probabilities. That is, the description is self-consistent and appropriate for operational frameworks.

More technically, we consider the distillation of quantum switches from generic processes $\kett{\mathbb w}$ parametrized by Eq. \eqref{eq:generalizedQS}, even those arbitrarily close to being causally separable.
One says one can distill quantum control of causal orders from $\kett{\mathbb w}$, with $p_0\neq p_1$, if there exists a free operation $\mathcal{V}$ that attains the transformation 
\begin{equation}
\kett{\mathbb w}^{\otimes N} \ \ \substack{{\rm free\ op}\\\longrightarrow} \ \ \kett{\mathbb w_{\rm qs}}^{\otimes rN} \ 
\label{eq:distillrate}
\end{equation}
with unit probability in the limit $N\to\infty$, for some rate
 $0\leq r\leq 1$. That latter is in turn called the distillation rate of $\kett{\mathbb w}$ relative to $\mathcal{V}$. Since we restrict ourselves to independent single-copy instruments, the deterministic multi-copy transformation is possible only if it is possible probabilistically on each copy. 
That is, Eq. \eqref{eq:distillrate} is achieved in the asymptotic limit iff $\kett{\mathbb w}$ is freely converted into $\kett{\mathbb w_{\rm qs}}$ with  probability $p_{\rm success}\coloneqq r$. Such conversion is shown in what follows.

\begin{lem}[Probabilistic single-copy pure-process conversion]
Let $\kett{\mathbb w}$ and $\kett{\mathbb w'}$ be generalized quantum switches such that $\boldsymbol p \succcurlyeq \boldsymbol p'$. Then there is a free operation that converts $\kett{\mathbb w}$ into $\kett{\mathbb w'}$ with probability $p_{\rm success}=\frac{\min\{p_0,p_1\}}{\min\{p_0',p_1'\}}$.
\label{th:prob_conversion}
\end{lem}

\noindent The proof is simple, consisting of a local filtering operation on Charlie's qubit followed by the protocol of Theo. \ref{th:conversion}. 
It is given explicitly in App. \ref{app:distill_proof}. 

Lemma \ref{th:prob_conversion} is important because it shows that single-copy process conversions where the final process is majorized by the initial one are also possible (albeit probabilistically). It thus complements Theo. \ref{th:conversion} for the deterministic case, possible only when the final process majorizes the initial one. In a sense, it is reminiscent of Vidal's theorem \cite{Vidal1999} for probabilistic single-copy entanglement conversions between arbitrary pure states. 
As anticipated, Eq. \eqref{eq:distillrate} follows as a corollary of Lem. \ref{th:prob_conversion}. Applying the lemma independently to each copy in $\kett{\mathbb w}^{\otimes N}$, taking the limit $N\to\infty$, and using the fact that $\boldsymbol p_{\rm qs}\coloneqq\{\frac12,\frac12\}$ proves the main result of this section:
\begin{cor}[Distillation of quantum control of causal orders]
\label{th:distillation}
Distillation of quantum control of causal orders exists. In fact, a perfect quantum switch can be distilled from any $\kett{\mathbb w}$ given by Eq. \eqref{eq:generalizedQS} at a rate $r=2\min\{p_0,p_1\}$.
\end{cor}

The specialized reader may note that the rate in Cor. \ref{th:distillation} is in general lower than the corresponding optimal entanglement- \cite{Bennett1995,Bennett1996,Bennett1996a,Horodecki2009} and coherence-distillation \cite{Yuan2015,Winter2016} rates for states analogous to the processes in Eq. \eqref{eq:generalizedQS}. Yet, in both entanglement and coherence distillation, global operations on each subsystem's copies are exploited, whereas here only  independent single-copy instruments are used. Interestingly, it is also possible to distill quantum switches using certain global multi-copy instruments: In the App. \ref{sec:multicopy_dist}, we briefly describe a protocol based on single-copy instruments on the target system conditioned on multi-qubit measurements on the copies of Charlie's control. These more general instruments still yield licit free operations -- also compatible with tensor products of processes -- because they do not involve any inter-copy signaling for the labs with indefinite causal orders (Alice and Bob's). As instrument for Charlie, this protocol uses the well-known global measurement from optimal entanglement \cite{Bennett1996} and coherence \cite{Yuan2015,Winter2016} distillation protocols. However, after the measurement, these protocols require operations whose equivalent here are not free operations. So, the restriction on the target's instruments renders the resulting distillation rate lower than that of Cor. \ref{th:distillation}. 
Furthermore, one could even explore protocols that exploit inter-copy signaling, as certain restricted signaling arrangements still give rise to licit free operations. However, such exploration is outside the scope of the current work, and we leave the question of whether those more powerful free operations do actually yield rates higher than that of Cor. \ref{th:distillation} as an open problem.
In any case, it is remarkable that distilling causal nonseparability is possible at all.

%%%%%%%%%%%%%%%%%%%%%%%%%%%%%%%%%%%%%%%%%%
%%%%%%%%%%%%%%%%%%%%%%%%%%%%%%%%%%%%%%%%%%
\section{Final discussion}
\label{sec:conclusions}
We studied  processes displaying quantum coherence between opposite causal orders as an operational resource. In particular, we derived a unified resource theory of both causal nonseparability and quantum control of causal orders. This required a rigorous definition of the latter notion, which, curiously, was still missing. We provided one. Our operational framework is based on resource-free operations consisting of sequential concatenations of the input process with physically-meaningful causally separable processes. As applications, first, we established a sufficient condition for pure-process convertibility, mathematically captured by a simple majorization relationship. 
This orders a  broad, important subclass of 
processes into a hierarchy of quantum control of causal orders with the quantum switch at the top, thus giving the latter
 the status of basic unit of this exclusive form of causal nonseparability. Second, we proved that distillation of quantum control of causal orders exists, and provided an explicit simple protocol for it. Our machinery is versatile in that it applies to both the mindsets with and without a control register.

As further direct potential applications, we may for instance mention causal-nonseparability measures and a resource theory of quantum causal networks. As for measures, here we have focused on process conversions, but the framework also directly paves the way for quantifiers. From an axiomatic point of view, causal nonseparability monotones can now be defined by any function that is non-increasing under the  free operations proposed. Examples thereof could for instance be the relative entropy and robustness of causal nonseparability or the causally nonseparable weight, which could be defined analogously to in other resource theories \cite{Gallego2012,Gallego2015,Gallego2017,Amaral2018}. 
From a more operational viewpoint, in turn, Cor. \ref{th:distillation} gives a lower bound to the distillable causal nonseparability of generalized quantum switches. As for causal networks, notably, our machinery not only treats causal nonseparability and quantum control of causal orders in a unified way but it also contains the basis of an eventual resource theory of quantum causal networks. More precisely, the subclass $\mathsf{LOAE}$ of local operations assisted by ancillary entanglement preserves the causal structure (either $A\to B$ or $B\to A$) of any quantum causal model for the simplest non-trivial case of two nodes (Alice and Bob's). Straightforward generalizations of it to more nodes will automatically define free operations for quantum causal networks, where the resourceful objects consist of quantum causal models incompatible with some given multi-node causal structure.

Besides, there are several exciting open questions that arise from this work. First, it is not clear whether there exist physical pure processes with quantum control of causal orders (or, more generally, causal nonseparability) apart from those of Eq. \eqref{eq:generalizedQS}. 
Second, if there is no single total unit of causal nonseparability, what other inequivalent classes of causal nonseparability are there? Third, regarding conversions, an important question is whether causal nonseparability dilution -- the converse of distillation -- is possible or not. Finally, all these questions clearly depend on the class of free operations adopted. Hence, a fourth vast unknown territory  is free operations beyond the ones proposed here. In particular, two interesting problems are whether there are free operations of quantum control of causal orders that are not free with respect to causal nonseparability or entanglement or free operations involving inter-copy signaling that lead to more efficient distillation protocols than the ones studied here. 
All of these are fascinating venues for future research.

%%%%%%%%%%%%%%%%%%%%%%%%%%%%%%%%%%%%%%%%%%%%%%%%%%%%%%%%%%%%%%%%%%%%%%%%%%%%%%%%%%%%%%%%%%%%%%%
%%%%%%%%%%%%%%%%%%%%%%%%%%%%%%%%%%%%%%%%%%%%%%%%%%%%%%%%%%%%%%%%%%%%%%%%%%%%%%%%%%%%%%%%%%%%%%%
\section*{Acknowledgements}
We thank R. Chaves and J.  Pienaar  for discussions, and Mateus Ara\'{u}jo for helpful comments on the manuscript. We acknowledge financial support from the Brazilian agencies CNPq (PQ grant No. 311416/2015-2 and INCT-IQ), FAPERJ (PDR10 E-26/202.802/2016, JCN E-26/202.701/2018), CAPES (PROCAD2013), FAPESP, and the Serrapilheira Institute (grant number Serra-1709-17173).

%%%%%%%%%%%%%%%%%%%%%%%%%%%%%%%%%%%%%%%%%%%%%%%%%%%%%%%%%%%%%%%%%%%%%%%%%%%%%%%%%%%%%%%%%%%%%%%
%%%%%%%%%%%%%%%%%%%%%%%%%%%%%%%%%%%%%%%%%%%%%%%%%%%%%%%%%%%%%%%%%%%%%%%%%%%%%%%%%%%%%%%%%%%%%%%

\bibliography{Causal,Steering}

\begin{thebibliography}{58}%
\makeatletter
\providecommand \@ifxundefined [1]{%
 \@ifx{#1\undefined}
}%
\providecommand \@ifnum [1]{%
 \ifnum #1\expandafter \@firstoftwo
 \else \expandafter \@secondoftwo
 \fi
}%
\providecommand \@ifx [1]{%
 \ifx #1\expandafter \@firstoftwo
 \else \expandafter \@secondoftwo
 \fi
}%
\providecommand \natexlab [1]{#1}%
\providecommand \enquote  [1]{``#1''}%
\providecommand \bibnamefont  [1]{#1}%
\providecommand \bibfnamefont [1]{#1}%
\providecommand \citenamefont [1]{#1}%
\providecommand \href@noop [0]{\@secondoftwo}%
\providecommand \href [0]{\begingroup \@sanitize@url \@href}%
\providecommand \@href[1]{\@@startlink{#1}\@@href}%
\providecommand \@@href[1]{\endgroup#1\@@endlink}%
\providecommand \@sanitize@url [0]{\catcode `\\12\catcode `\$12\catcode
  `\&12\catcode `\#12\catcode `\^12\catcode `\_12\catcode `\%12\relax}%
\providecommand \@@startlink[1]{}%
\providecommand \@@endlink[0]{}%
\providecommand \url  [0]{\begingroup\@sanitize@url \@url }%
\providecommand \@url [1]{\endgroup\@href {#1}{\urlprefix }}%
\providecommand \urlprefix  [0]{URL }%
\providecommand \Eprint [0]{\href }%
\providecommand \doibase [0]{http://dx.doi.org/}%
\providecommand \selectlanguage [0]{\@gobble}%
\providecommand \bibinfo  [0]{\@secondoftwo}%
\providecommand \bibfield  [0]{\@secondoftwo}%
\providecommand \translation [1]{[#1]}%
\providecommand \BibitemOpen [0]{}%
\providecommand \bibitemStop [0]{}%
\providecommand \bibitemNoStop [0]{.\EOS\space}%
\providecommand \EOS [0]{\spacefactor3000\relax}%
\providecommand \BibitemShut  [1]{\csname bibitem#1\endcsname}%
\let\auto@bib@innerbib\@empty
%</preamble>
\bibitem [{\citenamefont {Hardy}(2005)}]{Hardy2005}%
  \BibitemOpen
  \bibfield  {author} {\bibinfo {author} {\bibfnamefont {Lucien}\ \bibnamefont
  {Hardy}},\ }\bibfield  {title} {\enquote {\bibinfo {title} {{Probability
  Theories with Dynamic Causal Structure: A New Framework for Quantum
  Gravity}},}\ }\href {http://arxiv.org/abs/gr-qc/0509120} {\ ,\ \bibinfo
  {pages} {1--68} (\bibinfo {year} {2005})},\ \Eprint
  {http://arxiv.org/abs/0509120} {arXiv:0509120 [gr-qc]} \BibitemShut {NoStop}%
\bibitem [{\citenamefont {Hardy}(2007)}]{Hardy2007}%
  \BibitemOpen
  \bibfield  {author} {\bibinfo {author} {\bibfnamefont {Lucien}\ \bibnamefont
  {Hardy}},\ }\bibfield  {title} {\enquote {\bibinfo {title} {{Towards quantum
  gravity: A framework for probabilistic theories with non-fixed causal
  structure}},}\ }\href {\doibase 10.1088/1751-8113/40/12/S12} {\bibfield
  {journal} {\bibinfo  {journal} {Journal of Physics A: Mathematical and
  Theoretical}\ }\textbf {\bibinfo {volume} {40}},\ \bibinfo {pages}
  {3081--3099} (\bibinfo {year} {2007})},\ \Eprint
  {http://arxiv.org/abs/0608043} {arXiv:0608043 [gr-qc]} \BibitemShut {NoStop}%
\bibitem [{\citenamefont {Hardy}(2009)}]{Hardy2009}%
  \BibitemOpen
  \bibfield  {author} {\bibinfo {author} {\bibfnamefont {Lucien}\ \bibnamefont
  {Hardy}},\ }\bibfield  {title} {\enquote {\bibinfo {title} {{Quantum Gravity
  Computers: On the Theory of Computation with Indefinite Causal Structure}},}\
  \ }(\bibinfo {year} {2009})\ pp.\ \bibinfo {pages} {379--401},\ \Eprint
  {http://arxiv.org/abs/0701019} {arXiv:0701019 [quant-ph]} \BibitemShut
  {NoStop}%
\bibitem [{\citenamefont {Zych}\ \emph {et~al.}(2017)\citenamefont {Zych},
  \citenamefont {Costa}, \citenamefont {Pikovski},\ and\ \citenamefont
  {Brukner}}]{Zych2017}%
  \BibitemOpen
  \bibfield  {author} {\bibinfo {author} {\bibfnamefont {Magdalena}\
  \bibnamefont {Zych}}, \bibinfo {author} {\bibfnamefont {Fabio}\ \bibnamefont
  {Costa}}, \bibinfo {author} {\bibfnamefont {Igor}\ \bibnamefont {Pikovski}},
  \ and\ \bibinfo {author} {\bibfnamefont {Caslav}\ \bibnamefont {Brukner}},\
  }\bibfield  {title} {\enquote {\bibinfo {title} {{Bell's Theorem for Temporal
  Order}},}\ }\href {http://arxiv.org/abs/1708.00248} {\  (\bibinfo {year}
  {2017})},\ \Eprint {http://arxiv.org/abs/1708.00248} {arXiv:1708.00248}
  \BibitemShut {NoStop}%
\bibitem [{\citenamefont {Chiribella}(2012)}]{Chiribella2012}%
  \BibitemOpen
  \bibfield  {author} {\bibinfo {author} {\bibfnamefont {Giulio}\ \bibnamefont
  {Chiribella}},\ }\bibfield  {title} {\enquote {\bibinfo {title} {{Perfect
  discrimination of no-signalling channels via quantum superposition of causal
  structures}},}\ }\href {\doibase 10.1103/PhysRevA.86.040301} {\bibfield
  {journal} {\bibinfo  {journal} {Physical Review A - Atomic, Molecular, and
  Optical Physics}\ }\textbf {\bibinfo {volume} {86}},\ \bibinfo {pages} {1--5}
  (\bibinfo {year} {2012})},\ \Eprint {http://arxiv.org/abs/1109.5154}
  {arXiv:1109.5154} \BibitemShut {NoStop}%
\bibitem [{\citenamefont {Chiribella}\ \emph {et~al.}(2013)\citenamefont
  {Chiribella}, \citenamefont {D'Ariano}, \citenamefont {Perinotti},\ and\
  \citenamefont {Valiron}}]{Chiribella2013}%
  \BibitemOpen
  \bibfield  {author} {\bibinfo {author} {\bibfnamefont {Giulio}\ \bibnamefont
  {Chiribella}}, \bibinfo {author} {\bibfnamefont {Giacomo~Mauro}\ \bibnamefont
  {D'Ariano}}, \bibinfo {author} {\bibfnamefont {Paolo}\ \bibnamefont
  {Perinotti}}, \ and\ \bibinfo {author} {\bibfnamefont {Benoit}\ \bibnamefont
  {Valiron}},\ }\bibfield  {title} {\enquote {\bibinfo {title} {{Quantum
  computations without definite causal structure}},}\ }\href {\doibase
  10.1103/PhysRevA.88.022318} {\bibfield  {journal} {\bibinfo  {journal}
  {Physical Review A - Atomic, Molecular, and Optical Physics}\ }\textbf
  {\bibinfo {volume} {88}},\ \bibinfo {pages} {1--15} (\bibinfo {year}
  {2013})},\ \Eprint {http://arxiv.org/abs/arXiv:0912.0195v4}
  {arXiv:arXiv:0912.0195v4} \BibitemShut {NoStop}%
\bibitem [{\citenamefont {Chiribella}\ \emph {et~al.}(2008)\citenamefont
  {Chiribella}, \citenamefont {D'Ariano},\ and\ \citenamefont
  {Perinotti}}]{Chiribella2008}%
  \BibitemOpen
  \bibfield  {author} {\bibinfo {author} {\bibfnamefont {G.}~\bibnamefont
  {Chiribella}}, \bibinfo {author} {\bibfnamefont {G.~M.}\ \bibnamefont
  {D'Ariano}}, \ and\ \bibinfo {author} {\bibfnamefont {P.}~\bibnamefont
  {Perinotti}},\ }\bibfield  {title} {\enquote {\bibinfo {title} {{Quantum
  circuit architecture}},}\ }\href {\doibase 10.1103/PhysRevLett.101.060401}
  {\bibfield  {journal} {\bibinfo  {journal} {Physical Review Letters}\
  }\textbf {\bibinfo {volume} {101}},\ \bibinfo {pages} {1--4} (\bibinfo {year}
  {2008})},\ \Eprint {http://arxiv.org/abs/0712.1325} {arXiv:0712.1325}
  \BibitemShut {NoStop}%
\bibitem [{\citenamefont {Chiribella}\ \emph {et~al.}(2009)\citenamefont
  {Chiribella}, \citenamefont {D'Ariano},\ and\ \citenamefont
  {Perinotti}}]{Chiribella2009}%
  \BibitemOpen
  \bibfield  {author} {\bibinfo {author} {\bibfnamefont {Giulio}\ \bibnamefont
  {Chiribella}}, \bibinfo {author} {\bibfnamefont {Giacomo~Mauro}\ \bibnamefont
  {D'Ariano}}, \ and\ \bibinfo {author} {\bibfnamefont {Paolo}\ \bibnamefont
  {Perinotti}},\ }\bibfield  {title} {\enquote {\bibinfo {title} {{Theoretical
  framework for quantum networks}},}\ }\href {\doibase
  10.1103/PhysRevA.80.022339} {\bibfield  {journal} {\bibinfo  {journal}
  {Physical Review A - Atomic, Molecular, and Optical Physics}\ }\textbf
  {\bibinfo {volume} {80}},\ \bibinfo {pages} {1--20} (\bibinfo {year}
  {2009})},\ \Eprint {http://arxiv.org/abs/0904.4483} {arXiv:0904.4483}
  \BibitemShut {NoStop}%
\bibitem [{\citenamefont {Oreshkov}\ \emph {et~al.}(2012)\citenamefont
  {Oreshkov}, \citenamefont {Costa},\ and\ \citenamefont
  {Brukner}}]{Oreshkov2012}%
  \BibitemOpen
  \bibfield  {author} {\bibinfo {author} {\bibfnamefont {Ognyan}\ \bibnamefont
  {Oreshkov}}, \bibinfo {author} {\bibfnamefont {Fabio}\ \bibnamefont {Costa}},
  \ and\ \bibinfo {author} {\bibfnamefont {{\v{C}}aslav}\ \bibnamefont
  {Brukner}},\ }\bibfield  {title} {\enquote {\bibinfo {title} {{Quantum
  correlations with no causal order}},}\ }\href {\doibase 10.1038/ncomms2076}
  {\bibfield  {journal} {\bibinfo  {journal} {Nature Communications}\ }\textbf
  {\bibinfo {volume} {3}},\ \bibinfo {pages} {1092} (\bibinfo {year} {2012})},\
  \Eprint {http://arxiv.org/abs/1105.4464} {arXiv:1105.4464} \BibitemShut
  {NoStop}%
\bibitem [{\citenamefont {Ara{\'{u}}jo}\ \emph {et~al.}(2015)\citenamefont
  {Ara{\'{u}}jo}, \citenamefont {Branciard}, \citenamefont {Costa},
  \citenamefont {Feix}, \citenamefont {Giarmatzi},\ and\ \citenamefont
  {Brukner}}]{Araujo2015}%
  \BibitemOpen
  \bibfield  {author} {\bibinfo {author} {\bibfnamefont {Mateus}\ \bibnamefont
  {Ara{\'{u}}jo}}, \bibinfo {author} {\bibfnamefont {Cyril}\ \bibnamefont
  {Branciard}}, \bibinfo {author} {\bibfnamefont {Fabio}\ \bibnamefont
  {Costa}}, \bibinfo {author} {\bibfnamefont {Adrien}\ \bibnamefont {Feix}},
  \bibinfo {author} {\bibfnamefont {Christina}\ \bibnamefont {Giarmatzi}}, \
  and\ \bibinfo {author} {\bibfnamefont {{\v{C}}aslav}\ \bibnamefont
  {Brukner}},\ }\bibfield  {title} {\enquote {\bibinfo {title} {{Witnessing
  causal nonseparability}},}\ }\href {\doibase 10.1088/1367-2630/17/10/102001}
  {\bibfield  {journal} {\bibinfo  {journal} {New J. Phys.}\ }\textbf {\bibinfo
  {volume} {17}},\ \bibinfo {pages} {102001} (\bibinfo {year} {2015})},\
  \Eprint {http://arxiv.org/abs/1506.03776} {arXiv:1506.03776} \BibitemShut
  {NoStop}%
\bibitem [{\citenamefont {Oreshkov}\ and\ \citenamefont
  {Giarmatzi}(2016)}]{Oreshkov2016}%
  \BibitemOpen
  \bibfield  {author} {\bibinfo {author} {\bibfnamefont {Ognyan}\ \bibnamefont
  {Oreshkov}}\ and\ \bibinfo {author} {\bibfnamefont {Christina}\ \bibnamefont
  {Giarmatzi}},\ }\bibfield  {title} {\enquote {\bibinfo {title} {{Causal and
  causally separable processes}},}\ }\href {\doibase
  10.1088/1367-2630/18/9/093020} {\bibfield  {journal} {\bibinfo  {journal}
  {New Journal of Physics}\ }\textbf {\bibinfo {volume} {18}},\ \bibinfo
  {pages} {093020} (\bibinfo {year} {2016})},\ \Eprint
  {http://arxiv.org/abs/1506.05449} {arXiv:1506.05449} \BibitemShut {NoStop}%
\bibitem [{\citenamefont {MacLean}\ \emph {et~al.}(2017)\citenamefont
  {MacLean}, \citenamefont {Ried}, \citenamefont {Spekkens},\ and\
  \citenamefont {Resch}}]{MacLean2017}%
  \BibitemOpen
  \bibfield  {author} {\bibinfo {author} {\bibfnamefont {Jean-Philippe~W.}\
  \bibnamefont {MacLean}}, \bibinfo {author} {\bibfnamefont {Katja}\
  \bibnamefont {Ried}}, \bibinfo {author} {\bibfnamefont {Robert~W.}\
  \bibnamefont {Spekkens}}, \ and\ \bibinfo {author} {\bibfnamefont {Kevin~J.}\
  \bibnamefont {Resch}},\ }\bibfield  {title} {\enquote {\bibinfo {title}
  {{Quantum-coherent mixtures of causal relations}},}\ }\href {\doibase
  10.1038/ncomms15149} {\bibfield  {journal} {\bibinfo  {journal} {Nature
  Communications}\ }\textbf {\bibinfo {volume} {8}},\ \bibinfo {pages} {15149}
  (\bibinfo {year} {2017})}\BibitemShut {NoStop}%
\bibitem [{\citenamefont {Branciard}\ \emph {et~al.}(2015)\citenamefont
  {Branciard}, \citenamefont {Ara{\'{u}}jo}, \citenamefont {Feix},
  \citenamefont {Costa},\ and\ \citenamefont {Brukner}}]{Branciard2016}%
  \BibitemOpen
  \bibfield  {author} {\bibinfo {author} {\bibfnamefont {Cyril}\ \bibnamefont
  {Branciard}}, \bibinfo {author} {\bibfnamefont {Mateus}\ \bibnamefont
  {Ara{\'{u}}jo}}, \bibinfo {author} {\bibfnamefont {Adrien}\ \bibnamefont
  {Feix}}, \bibinfo {author} {\bibfnamefont {Fabio}\ \bibnamefont {Costa}}, \
  and\ \bibinfo {author} {\bibfnamefont {{\v{C}}aslav}\ \bibnamefont
  {Brukner}},\ }\bibfield  {title} {\enquote {\bibinfo {title} {{The simplest
  causal inequalities and their violation}},}\ }\href {\doibase
  10.1088/1367-2630/18/1/013008} {\bibfield  {journal} {\bibinfo  {journal}
  {New Journal of Physics}\ }\textbf {\bibinfo {volume} {18}},\ \bibinfo
  {pages} {013008} (\bibinfo {year} {2015})},\ \Eprint
  {http://arxiv.org/abs/1508.01704} {arXiv:1508.01704} \BibitemShut {NoStop}%
\bibitem [{\citenamefont {Ara{\'{u}}jo}\ \emph {et~al.}(2014)\citenamefont
  {Ara{\'{u}}jo}, \citenamefont {Costa},\ and\ \citenamefont
  {Brukner}}]{Araujo2014}%
  \BibitemOpen
  \bibfield  {author} {\bibinfo {author} {\bibfnamefont {Mateus}\ \bibnamefont
  {Ara{\'{u}}jo}}, \bibinfo {author} {\bibfnamefont {Fabio}\ \bibnamefont
  {Costa}}, \ and\ \bibinfo {author} {\bibfnamefont {{\v{C}}aslav}\
  \bibnamefont {Brukner}},\ }\bibfield  {title} {\enquote {\bibinfo {title}
  {{Computational Advantage from Quantum-Controlled Ordering of Gates}},}\
  }\href {\doibase 10.1103/PhysRevLett.113.250402} {\bibfield  {journal}
  {\bibinfo  {journal} {Physical Review Letters}\ }\textbf {\bibinfo {volume}
  {113}},\ \bibinfo {pages} {250402} (\bibinfo {year} {2014})},\ \Eprint
  {http://arxiv.org/abs/1401.8127} {arXiv:1401.8127} \BibitemShut {NoStop}%
\bibitem [{\citenamefont {Procopio}\ \emph {et~al.}(2015)\citenamefont
  {Procopio}, \citenamefont {Moqanaki}, \citenamefont {Ara{\'{u}}jo},
  \citenamefont {Costa}, \citenamefont {{Alonso Calafell}}, \citenamefont
  {Dowd}, \citenamefont {Hamel}, \citenamefont {Rozema}, \citenamefont
  {Brukner},\ and\ \citenamefont {Walther}}]{Procopio2014}%
  \BibitemOpen
  \bibfield  {author} {\bibinfo {author} {\bibfnamefont {Lorenzo~M.}\
  \bibnamefont {Procopio}}, \bibinfo {author} {\bibfnamefont {Amir}\
  \bibnamefont {Moqanaki}}, \bibinfo {author} {\bibfnamefont {Mateus}\
  \bibnamefont {Ara{\'{u}}jo}}, \bibinfo {author} {\bibfnamefont {Fabio}\
  \bibnamefont {Costa}}, \bibinfo {author} {\bibfnamefont {Irati}\ \bibnamefont
  {{Alonso Calafell}}}, \bibinfo {author} {\bibfnamefont {Emma~G.}\
  \bibnamefont {Dowd}}, \bibinfo {author} {\bibfnamefont {Deny~R.}\
  \bibnamefont {Hamel}}, \bibinfo {author} {\bibfnamefont {Lee~A.}\
  \bibnamefont {Rozema}}, \bibinfo {author} {\bibfnamefont {{\v{C}}aslav}\
  \bibnamefont {Brukner}}, \ and\ \bibinfo {author} {\bibfnamefont {Philip}\
  \bibnamefont {Walther}},\ }\bibfield  {title} {\enquote {\bibinfo {title}
  {{Experimental superposition of orders of quantum gates}},}\ }\href {\doibase
  10.1038/ncomms8913} {\bibfield  {journal} {\bibinfo  {journal} {Nature
  Communications}\ }\textbf {\bibinfo {volume} {6}},\ \bibinfo {pages} {7913}
  (\bibinfo {year} {2015})},\ \Eprint {http://arxiv.org/abs/1412.4006}
  {arXiv:1412.4006} \BibitemShut {NoStop}%
\bibitem [{\citenamefont {Gu{\'{e}}rin}\ \emph {et~al.}(2016)\citenamefont
  {Gu{\'{e}}rin}, \citenamefont {Feix}, \citenamefont {Ara{\'{u}}jo},\ and\
  \citenamefont {Brukner}}]{Guerin2016}%
  \BibitemOpen
  \bibfield  {author} {\bibinfo {author} {\bibfnamefont {Philippe~Allard}\
  \bibnamefont {Gu{\'{e}}rin}}, \bibinfo {author} {\bibfnamefont {Adrien}\
  \bibnamefont {Feix}}, \bibinfo {author} {\bibfnamefont {Mateus}\ \bibnamefont
  {Ara{\'{u}}jo}}, \ and\ \bibinfo {author} {\bibfnamefont {{\v{C}}aslav}\
  \bibnamefont {Brukner}},\ }\bibfield  {title} {\enquote {\bibinfo {title}
  {{Exponential Communication Complexity Advantage from Quantum Superposition
  of the Direction of Communication}},}\ }\href {\doibase
  10.1103/PhysRevLett.117.100502} {\bibfield  {journal} {\bibinfo  {journal}
  {Phys. Rev. Lett.}\ }\textbf {\bibinfo {volume} {117}},\ \bibinfo {pages}
  {100502} (\bibinfo {year} {2016})},\ \Eprint
  {http://arxiv.org/abs/1605.07372} {arXiv:1605.07372} \BibitemShut {NoStop}%
\bibitem [{\citenamefont {Ara{\'{u}}jo}\ \emph {et~al.}(2017)\citenamefont
  {Ara{\'{u}}jo}, \citenamefont {Feix}, \citenamefont {Navascu{\'{e}}s},\ and\
  \citenamefont {Brukner}}]{Araujo2017}%
  \BibitemOpen
  \bibfield  {author} {\bibinfo {author} {\bibfnamefont {Mateus}\ \bibnamefont
  {Ara{\'{u}}jo}}, \bibinfo {author} {\bibfnamefont {Adrien}\ \bibnamefont
  {Feix}}, \bibinfo {author} {\bibfnamefont {Miguel}\ \bibnamefont
  {Navascu{\'{e}}s}}, \ and\ \bibinfo {author} {\bibfnamefont {{\v{C}}aslav}\
  \bibnamefont {Brukner}},\ }\bibfield  {title} {\enquote {\bibinfo {title} {{A
  purification postulate for quantum mechanics with indefinite causal
  order}},}\ }\href {\doibase 10.22331/q-2017-04-26-10} {\bibfield  {journal}
  {\bibinfo  {journal} {Quantum}\ }\textbf {\bibinfo {volume} {1}},\ \bibinfo
  {pages} {10} (\bibinfo {year} {2017})},\ \Eprint
  {http://arxiv.org/abs/1611.08535} {arXiv:1611.08535} \BibitemShut {NoStop}%
\bibitem [{\citenamefont {Ebler}\ \emph {et~al.}(2018)\citenamefont {Ebler},
  \citenamefont {Salek},\ and\ \citenamefont {Chiribella}}]{Ebler2018}%
  \BibitemOpen
  \bibfield  {author} {\bibinfo {author} {\bibfnamefont {Daniel}\ \bibnamefont
  {Ebler}}, \bibinfo {author} {\bibfnamefont {Sina}\ \bibnamefont {Salek}}, \
  and\ \bibinfo {author} {\bibfnamefont {Giulio}\ \bibnamefont {Chiribella}},\
  }\bibfield  {title} {\enquote {\bibinfo {title} {{Enhanced Communication with
  the Assistance of Indefinite Causal Order}},}\ }\href {\doibase
  10.1103/PhysRevLett.120.120502} {\bibfield  {journal} {\bibinfo  {journal}
  {Physical Review Letters}\ }\textbf {\bibinfo {volume} {120}},\ \bibinfo
  {pages} {120502} (\bibinfo {year} {2018})},\ \Eprint
  {http://arxiv.org/abs/1711.10165} {arXiv:1711.10165} \BibitemShut {NoStop}%
\bibitem [{\citenamefont {Rubino}\ \emph {et~al.}(2017)\citenamefont {Rubino},
  \citenamefont {Rozema}, \citenamefont {Feix}, \citenamefont {Ara{\'{u}}jo},
  \citenamefont {Zeuner}, \citenamefont {Procopio}, \citenamefont {Brukner},\
  and\ \citenamefont {Walther}}]{Rubino2016}%
  \BibitemOpen
  \bibfield  {author} {\bibinfo {author} {\bibfnamefont {Giulia}\ \bibnamefont
  {Rubino}}, \bibinfo {author} {\bibfnamefont {Lee~A.}\ \bibnamefont {Rozema}},
  \bibinfo {author} {\bibfnamefont {Adrien}\ \bibnamefont {Feix}}, \bibinfo
  {author} {\bibfnamefont {Mateus}\ \bibnamefont {Ara{\'{u}}jo}}, \bibinfo
  {author} {\bibfnamefont {Jonas~M.}\ \bibnamefont {Zeuner}}, \bibinfo {author}
  {\bibfnamefont {Lorenzo~M.}\ \bibnamefont {Procopio}}, \bibinfo {author}
  {\bibfnamefont {{\v{C}}aslav}\ \bibnamefont {Brukner}}, \ and\ \bibinfo
  {author} {\bibfnamefont {Philip}\ \bibnamefont {Walther}},\ }\bibfield
  {title} {\enquote {\bibinfo {title} {{Experimental verification of an
  indefinite causal order}},}\ }\href {\doibase 10.1126/sciadv.1602589}
  {\bibfield  {journal} {\bibinfo  {journal} {Science Advances}\ }\textbf
  {\bibinfo {volume} {3}},\ \bibinfo {pages} {e1602589} (\bibinfo {year}
  {2017})},\ \Eprint {http://arxiv.org/abs/1608.01683} {arXiv:1608.01683}
  \BibitemShut {NoStop}%
\bibitem [{\citenamefont {Goswami}\ \emph {et~al.}(2018)\citenamefont
  {Goswami}, \citenamefont {Giarmatzi}, \citenamefont {Kewming}, \citenamefont
  {Costa}, \citenamefont {Branciard}, \citenamefont {Romero},\ and\
  \citenamefont {White}}]{Goswami2018}%
  \BibitemOpen
  \bibfield  {author} {\bibinfo {author} {\bibfnamefont {K.}~\bibnamefont
  {Goswami}}, \bibinfo {author} {\bibfnamefont {C.}~\bibnamefont {Giarmatzi}},
  \bibinfo {author} {\bibfnamefont {M.}~\bibnamefont {Kewming}}, \bibinfo
  {author} {\bibfnamefont {F.}~\bibnamefont {Costa}}, \bibinfo {author}
  {\bibfnamefont {C.}~\bibnamefont {Branciard}}, \bibinfo {author}
  {\bibfnamefont {J.}~\bibnamefont {Romero}}, \ and\ \bibinfo {author}
  {\bibfnamefont {A.~G.}\ \bibnamefont {White}},\ }\bibfield  {title} {\enquote
  {\bibinfo {title} {{Indefinite Causal Order in a Quantum Switch}},}\ }\href
  {\doibase 10.1103/PhysRevLett.121.090503} {\bibfield  {journal} {\bibinfo
  {journal} {Physical Review Letters}\ }\textbf {\bibinfo {volume} {121}},\
  \bibinfo {pages} {090503} (\bibinfo {year} {2018})},\ \Eprint
  {http://arxiv.org/abs/1803.04302} {arXiv:1803.04302} \BibitemShut {NoStop}%
\bibitem [{\citenamefont {Wei}\ \emph {et~al.}(2018)\citenamefont {Wei},
  \citenamefont {Tischler}, \citenamefont {Zhao}, \citenamefont {Li},
  \citenamefont {Arrazola}, \citenamefont {Liu}, \citenamefont {Zhang},
  \citenamefont {Li}, \citenamefont {You}, \citenamefont {Wang}, \citenamefont
  {Chen}, \citenamefont {Sanders}, \citenamefont {Zhang}, \citenamefont
  {Pryde}, \citenamefont {Xu},\ and\ \citenamefont {Pan}}]{Wei2018}%
  \BibitemOpen
  \bibfield  {author} {\bibinfo {author} {\bibfnamefont {Kejin}\ \bibnamefont
  {Wei}}, \bibinfo {author} {\bibfnamefont {Nora}\ \bibnamefont {Tischler}},
  \bibinfo {author} {\bibfnamefont {Si-ran}\ \bibnamefont {Zhao}}, \bibinfo
  {author} {\bibfnamefont {Yu-huai}\ \bibnamefont {Li}}, \bibinfo {author}
  {\bibfnamefont {Juan~Miguel}\ \bibnamefont {Arrazola}}, \bibinfo {author}
  {\bibfnamefont {Yang}\ \bibnamefont {Liu}}, \bibinfo {author} {\bibfnamefont
  {Weijun}\ \bibnamefont {Zhang}}, \bibinfo {author} {\bibfnamefont {Hao}\
  \bibnamefont {Li}}, \bibinfo {author} {\bibfnamefont {Lixing}\ \bibnamefont
  {You}}, \bibinfo {author} {\bibfnamefont {Zhen}\ \bibnamefont {Wang}},
  \bibinfo {author} {\bibfnamefont {Yu-ao}\ \bibnamefont {Chen}}, \bibinfo
  {author} {\bibfnamefont {Barry~C}\ \bibnamefont {Sanders}}, \bibinfo {author}
  {\bibfnamefont {Qiang}\ \bibnamefont {Zhang}}, \bibinfo {author}
  {\bibfnamefont {Geoff~J}\ \bibnamefont {Pryde}}, \bibinfo {author}
  {\bibfnamefont {Feihu}\ \bibnamefont {Xu}}, \ and\ \bibinfo {author}
  {\bibfnamefont {Jian-Wei}\ \bibnamefont {Pan}},\ }\href
  {http://arxiv.org/abs/1810.10238} {\enquote {\bibinfo {title} {{Experimental
  Quantum Switching for Exponentially Superior Quantum Communication
  Complexity}},}\ } (\bibinfo {year} {2018}),\ \Eprint
  {http://arxiv.org/abs/1810.10238v1} {arXiv:1810.10238v1} \BibitemShut
  {NoStop}%
\bibitem [{\citenamefont {Brand{\~{a}}o}\ and\ \citenamefont
  {Gour}(2015)}]{Brandao2015}%
  \BibitemOpen
  \bibfield  {author} {\bibinfo {author} {\bibfnamefont {Fernando G. S.~L.}\
  \bibnamefont {Brand{\~{a}}o}}\ and\ \bibinfo {author} {\bibfnamefont {Gilad}\
  \bibnamefont {Gour}},\ }\bibfield  {title} {\enquote {\bibinfo {title}
  {{Reversible Framework for Quantum Resource Theories}},}\ }\href {\doibase
  10.1103/PhysRevLett.115.070503} {\bibfield  {journal} {\bibinfo  {journal}
  {Phys. Rev. Lett.}\ }\textbf {\bibinfo {volume} {070503}},\ \bibinfo {pages}
  {1--5} (\bibinfo {year} {2015})}\BibitemShut {NoStop}%
\bibitem [{\citenamefont {Coecke}\ \emph {et~al.}(2016)\citenamefont {Coecke},
  \citenamefont {Fritz},\ and\ \citenamefont {Spekkens}}]{Coecke2016}%
  \BibitemOpen
  \bibfield  {author} {\bibinfo {author} {\bibfnamefont {Bob}\ \bibnamefont
  {Coecke}}, \bibinfo {author} {\bibfnamefont {Tobias}\ \bibnamefont {Fritz}},
  \ and\ \bibinfo {author} {\bibfnamefont {Robert~W.}\ \bibnamefont
  {Spekkens}},\ }\bibfield  {title} {\enquote {\bibinfo {title} {{A
  mathematical theory of resources}},}\ }\href {\doibase
  10.1016/j.ic.2016.02.008} {\bibfield  {journal} {\bibinfo  {journal}
  {Information and Computation}\ }\textbf {\bibinfo {volume} {250}},\ \bibinfo
  {pages} {59--86} (\bibinfo {year} {2016})},\ \Eprint
  {http://arxiv.org/abs/1409.5531} {arXiv:1409.5531} \BibitemShut {NoStop}%
\bibitem [{\citenamefont {Gallego}\ \emph {et~al.}(2012)\citenamefont
  {Gallego}, \citenamefont {W{\"{u}}rflinger}, \citenamefont {Ac{\'{i}}n},\
  and\ \citenamefont {Navascu{\'{e}}s}}]{Gallego2012}%
  \BibitemOpen
  \bibfield  {author} {\bibinfo {author} {\bibfnamefont {Rodrigo}\ \bibnamefont
  {Gallego}}, \bibinfo {author} {\bibfnamefont {Lars~Erik}\ \bibnamefont
  {W{\"{u}}rflinger}}, \bibinfo {author} {\bibfnamefont {Antonio}\ \bibnamefont
  {Ac{\'{i}}n}}, \ and\ \bibinfo {author} {\bibfnamefont {Miguel}\ \bibnamefont
  {Navascu{\'{e}}s}},\ }\bibfield  {title} {\enquote {\bibinfo {title}
  {{Operational Framework for Nonlocality}},}\ }\href {\doibase
  10.1103/PhysRevLett.109.070401} {\bibfield  {journal} {\bibinfo  {journal}
  {Phys. Rev. Lett.}\ }\textbf {\bibinfo {volume} {109}},\ \bibinfo {pages}
  {70401} (\bibinfo {year} {2012})},\ \Eprint {http://arxiv.org/abs/1112.2647}
  {arXiv:1112.2647} \BibitemShut {NoStop}%
\bibitem [{\citenamefont {Gallego}\ and\ \citenamefont
  {Aolita}(2015)}]{Gallego2015}%
  \BibitemOpen
  \bibfield  {author} {\bibinfo {author} {\bibfnamefont {Rodrigo}\ \bibnamefont
  {Gallego}}\ and\ \bibinfo {author} {\bibfnamefont {Leandro}\ \bibnamefont
  {Aolita}},\ }\bibfield  {title} {\enquote {\bibinfo {title} {{Resource theory
  of steering}},}\ }\href {\doibase 10.1103/PhysRevX.5.041008} {\bibfield
  {journal} {\bibinfo  {journal} {Phys. Rev. X}\ }\textbf {\bibinfo {volume}
  {5}},\ \bibinfo {pages} {1--19} (\bibinfo {year} {2015})},\ \Eprint
  {http://arxiv.org/abs/arXiv:1409.5804} {arXiv:arXiv:1409.5804} \BibitemShut
  {NoStop}%
\bibitem [{\citenamefont {Gallego}\ and\ \citenamefont
  {Aolita}(2017)}]{Gallego2017}%
  \BibitemOpen
  \bibfield  {author} {\bibinfo {author} {\bibfnamefont {Rodrigo}\ \bibnamefont
  {Gallego}}\ and\ \bibinfo {author} {\bibfnamefont {Leandro}\ \bibnamefont
  {Aolita}},\ }\bibfield  {title} {\enquote {\bibinfo {title} {{Nonlocality
  free wirings and the distinguishability between Bell boxes}},}\ }\href
  {\doibase 10.1103/PhysRevA.95.032118} {\bibfield  {journal} {\bibinfo
  {journal} {Physical Review A}\ }\textbf {\bibinfo {volume} {95}},\ \bibinfo
  {pages} {1--14} (\bibinfo {year} {2017})},\ \Eprint
  {http://arxiv.org/abs/1611.06932} {arXiv:1611.06932} \BibitemShut {NoStop}%
\bibitem [{\citenamefont {Amaral}\ \emph {et~al.}(2018)\citenamefont {Amaral},
  \citenamefont {Cabello}, \citenamefont {Cunha},\ and\ \citenamefont
  {Aolita}}]{Amaral2018}%
  \BibitemOpen
  \bibfield  {author} {\bibinfo {author} {\bibfnamefont {Barbara}\ \bibnamefont
  {Amaral}}, \bibinfo {author} {\bibfnamefont {Ad{\'{a}}n}\ \bibnamefont
  {Cabello}}, \bibinfo {author} {\bibfnamefont {Marcelo~Terra}\ \bibnamefont
  {Cunha}}, \ and\ \bibinfo {author} {\bibfnamefont {Leandro}\ \bibnamefont
  {Aolita}},\ }\bibfield  {title} {\enquote {\bibinfo {title} {{Noncontextual
  Wirings}},}\ }\href {\doibase 10.1103/PhysRevLett.120.130403} {\bibfield
  {journal} {\bibinfo  {journal} {Physical Review Letters}\ }\textbf {\bibinfo
  {volume} {120}},\ \bibinfo {pages} {130403} (\bibinfo {year}
  {2018})}\BibitemShut {NoStop}%
\bibitem [{\citenamefont {Choi}(1975)}]{Choi1975}%
  \BibitemOpen
  \bibfield  {author} {\bibinfo {author} {\bibfnamefont {Man~Duen}\
  \bibnamefont {Choi}},\ }\bibfield  {title} {\enquote {\bibinfo {title}
  {{Completely positive linear maps on complex matrices}},}\ }\href {\doibase
  10.1016/0024-3795(75)90075-0} {\bibfield  {journal} {\bibinfo  {journal}
  {Linear Algebra and Its Applications}\ }\textbf {\bibinfo {volume} {10}},\
  \bibinfo {pages} {285--290} (\bibinfo {year} {1975})}\BibitemShut {NoStop}%
\bibitem [{\citenamefont {Bengtsson}\ and\ \citenamefont
  {Zyczkowski}(2006)}]{Bengtsson2006}%
  \BibitemOpen
  \bibfield  {author} {\bibinfo {author} {\bibfnamefont {Ingemar}\ \bibnamefont
  {Bengtsson}}\ and\ \bibinfo {author} {\bibfnamefont {Karol}\ \bibnamefont
  {Zyczkowski}},\ }\href {\doibase 10.1017/CBO9780511535048} {\emph {\bibinfo
  {title} {{Geometry of Quantum States}}}}\ (\bibinfo  {publisher} {Cambridge
  University Press},\ \bibinfo {address} {Cambridge},\ \bibinfo {year}
  {2006})\BibitemShut {NoStop}%
\bibitem [{\citenamefont {Costa}\ and\ \citenamefont
  {Shrapnel}(2016)}]{Costa2016a}%
  \BibitemOpen
  \bibfield  {author} {\bibinfo {author} {\bibfnamefont {Fabio}\ \bibnamefont
  {Costa}}\ and\ \bibinfo {author} {\bibfnamefont {Sally}\ \bibnamefont
  {Shrapnel}},\ }\bibfield  {title} {\enquote {\bibinfo {title} {{Quantum
  causal modelling}},}\ }\href {\doibase 10.1088/1367-2630/18/6/063032}
  {\bibfield  {journal} {\bibinfo  {journal} {New Journal of Physics}\ }\textbf
  {\bibinfo {volume} {18}},\ \bibinfo {pages} {063032} (\bibinfo {year}
  {2016})}\BibitemShut {NoStop}%
\bibitem [{\citenamefont {Allen}\ \emph {et~al.}(2017)\citenamefont {Allen},
  \citenamefont {Barrett}, \citenamefont {Horsman}, \citenamefont {Lee},\ and\
  \citenamefont {Spekkens}}]{Allen2017}%
  \BibitemOpen
  \bibfield  {author} {\bibinfo {author} {\bibfnamefont {John-Mark~A.}\
  \bibnamefont {Allen}}, \bibinfo {author} {\bibfnamefont {Jonathan}\
  \bibnamefont {Barrett}}, \bibinfo {author} {\bibfnamefont {Dominic~C.}\
  \bibnamefont {Horsman}}, \bibinfo {author} {\bibfnamefont {Ciar{\'{a}}n~M.}\
  \bibnamefont {Lee}}, \ and\ \bibinfo {author} {\bibfnamefont {Robert~W.}\
  \bibnamefont {Spekkens}},\ }\bibfield  {title} {\enquote {\bibinfo {title}
  {{Quantum Common Causes and Quantum Causal Models}},}\ }\href {\doibase
  10.1103/PhysRevX.7.031021} {\bibfield  {journal} {\bibinfo  {journal}
  {Physical Review X}\ }\textbf {\bibinfo {volume} {7}},\ \bibinfo {pages}
  {031021} (\bibinfo {year} {2017})}\BibitemShut {NoStop}%
\bibitem [{\citenamefont {Feix}\ \emph {et~al.}(2016)\citenamefont {Feix},
  \citenamefont {Ara\'ujo},\ and\ \citenamefont {Brukner}}]{FAB16}%
  \BibitemOpen
  \bibfield  {author} {\bibinfo {author} {\bibfnamefont {Adrien}\ \bibnamefont
  {Feix}}, \bibinfo {author} {\bibfnamefont {Mateus}\ \bibnamefont {Ara\'ujo}},
  \ and\ \bibinfo {author} {\bibfnamefont {{\v{C}}aslav}\ \bibnamefont
  {Brukner}},\ }\bibfield  {title} {\enquote {\bibinfo {title} {Causally
  nonseparable processes admitting a causal model},}\ }\href@noop {} {\bibfield
   {journal} {\bibinfo  {journal} {New Journal of Physics}\ }\textbf {\bibinfo
  {volume} {18}},\ \bibinfo {pages} {083040} (\bibinfo {year}
  {2016})}\BibitemShut {NoStop}%
\bibitem [{\citenamefont {Horodecki}\ \emph {et~al.}(2009)\citenamefont
  {Horodecki}, \citenamefont {Horodecki}, \citenamefont {Horodecki},\ and\
  \citenamefont {Horodecki}}]{Horodecki2009}%
  \BibitemOpen
  \bibfield  {author} {\bibinfo {author} {\bibfnamefont {Ryszard}\ \bibnamefont
  {Horodecki}}, \bibinfo {author} {\bibfnamefont {Pawe{\l}}\ \bibnamefont
  {Horodecki}}, \bibinfo {author} {\bibfnamefont {Micha{\l}}\ \bibnamefont
  {Horodecki}}, \ and\ \bibinfo {author} {\bibfnamefont {Karol}\ \bibnamefont
  {Horodecki}},\ }\bibfield  {title} {\enquote {\bibinfo {title} {{Quantum
  entanglement}},}\ }\href {\doibase 10.1103/RevModPhys.81.865} {\bibfield
  {journal} {\bibinfo  {journal} {Reviews of Modern Physics}\ }\textbf
  {\bibinfo {volume} {81}},\ \bibinfo {pages} {865--942} (\bibinfo {year}
  {2009})},\ \Eprint {http://arxiv.org/abs/0702225} {arXiv:0702225 [quant-ph]}
  \BibitemShut {NoStop}%
\bibitem [{\citenamefont {Jungnitsch}\ \emph {et~al.}(2011)\citenamefont
  {Jungnitsch}, \citenamefont {Moroder},\ and\ \citenamefont
  {G{\"{u}}hne}}]{Jungnitsch2011}%
  \BibitemOpen
  \bibfield  {author} {\bibinfo {author} {\bibfnamefont {Bastian}\ \bibnamefont
  {Jungnitsch}}, \bibinfo {author} {\bibfnamefont {Tobias}\ \bibnamefont
  {Moroder}}, \ and\ \bibinfo {author} {\bibfnamefont {Otfried}\ \bibnamefont
  {G{\"{u}}hne}},\ }\bibfield  {title} {\enquote {\bibinfo {title} {{Taming
  Multiparticle Entanglement}},}\ }\href {\doibase
  10.1103/PhysRevLett.106.190502} {\bibfield  {journal} {\bibinfo  {journal}
  {Physical Review Letters}\ }\textbf {\bibinfo {volume} {106}},\ \bibinfo
  {pages} {190502} (\bibinfo {year} {2011})},\ \Eprint
  {http://arxiv.org/abs/1010.6049} {arXiv:1010.6049} \BibitemShut {NoStop}%
\bibitem [{\citenamefont {Aolita}\ \emph {et~al.}(2015)\citenamefont {Aolita},
  \citenamefont {de~Melo},\ and\ \citenamefont {Davidovich}}]{Aolita2015}%
  \BibitemOpen
  \bibfield  {author} {\bibinfo {author} {\bibfnamefont {Leandro}\ \bibnamefont
  {Aolita}}, \bibinfo {author} {\bibfnamefont {Fernando}\ \bibnamefont
  {de~Melo}}, \ and\ \bibinfo {author} {\bibfnamefont {Luiz}\ \bibnamefont
  {Davidovich}},\ }\bibfield  {title} {\enquote {\bibinfo {title} {{Open-system
  dynamics of entanglement:a key issues review}},}\ }\href {\doibase
  10.1088/0034-4885/78/4/042001} {\bibfield  {journal} {\bibinfo  {journal}
  {Reports Prog. Phys.}\ }\textbf {\bibinfo {volume} {78}},\ \bibinfo {pages}
  {042001} (\bibinfo {year} {2015})},\ \Eprint {http://arxiv.org/abs/1402.3713}
  {arXiv:1402.3713} \BibitemShut {NoStop}%
\bibitem [{\citenamefont {Davies}\ and\ \citenamefont
  {Lewis}(1970)}]{Davies1970}%
  \BibitemOpen
  \bibfield  {author} {\bibinfo {author} {\bibfnamefont {E.~B.}\ \bibnamefont
  {Davies}}\ and\ \bibinfo {author} {\bibfnamefont {J~T}\ \bibnamefont
  {Lewis}},\ }\bibfield  {title} {\enquote {\bibinfo {title} {{An operational
  approach to quantum probability}},}\ }\href {\doibase 10.1007/BF01647093}
  {\bibfield  {journal} {\bibinfo  {journal} {Communications in Mathematical
  Physics}\ }\textbf {\bibinfo {volume} {17}},\ \bibinfo {pages} {239--260}
  (\bibinfo {year} {1970})}\BibitemShut {NoStop}%
\bibitem [{Note1()}]{Note1}%
  \BibitemOpen
  \bibinfo {note} {One has in principle three different isomorphic ancilla
  spaces on Alice's side, one before the application of $U_{A_I}^{(j)}$, one
  between the application of $U_{A_I}^{(j)}$ and $U_{A_O}^{(j)}$, and one after
  $U_{A_O}^{(j)}$. Unlike the other variables, however, this ancilla is not
  matched to any Hilbert space outside $V_{AB}^{(j)}$, so we can use a single
  variable $\protect \mathaccentV {tilde}07EA$ for all of them. Same for
  $\protect \mathaccentV {tilde}07EB$.}\BibitemShut {Stop}%
\bibitem [{\citenamefont {Stinespring}(1955)}]{Stinespring1955}%
  \BibitemOpen
  \bibfield  {author} {\bibinfo {author} {\bibfnamefont {W~Forrest}\
  \bibnamefont {Stinespring}},\ }\bibfield  {title} {\enquote {\bibinfo {title}
  {{Positive functions on C*-algebras}},}\ }\href {\doibase
  10.1090/S0002-9939-1955-0069403-4} {\bibfield  {journal} {\bibinfo  {journal}
  {Proceedings of the American Mathematical Society}\ }\textbf {\bibinfo
  {volume} {6}},\ \bibinfo {pages} {211--211} (\bibinfo {year} {1955})},\
  \Eprint {http://arxiv.org/abs/9707028v2} {arXiv:9707028v2 [quant-ph]}
  \BibitemShut {NoStop}%
\bibitem [{Note2()}]{Note2}%
  \BibitemOpen
  \bibinfo {note} {{Strictly speaking, it does not define a valid process
  matrix, as it does not satisfy the necessary normalization conditions.
  However, it does define a physical process that can be implemented
  (probabilistically, with a probability depending on the instruments on the
  target) by post-selecting $|\protect \mathbb {w}_{\protect \rm qs}\delimiter
  "526930B \protect \tmspace -\thinmuskip {.1667em}\delimiter "526930B $ on a
  local measurement on $C$.}}\BibitemShut {Stop}%
\bibitem [{\citenamefont {Nielsen}(1999)}]{Nielsen1999}%
  \BibitemOpen
  \bibfield  {author} {\bibinfo {author} {\bibfnamefont {M.~A.}\ \bibnamefont
  {Nielsen}},\ }\bibfield  {title} {\enquote {\bibinfo {title} {{Conditions for
  a Class of Entanglement Transformations}},}\ }\href {\doibase
  10.1103/PhysRevLett.83.436} {\bibfield  {journal} {\bibinfo  {journal}
  {Physical Review Letters}\ }\textbf {\bibinfo {volume} {83}},\ \bibinfo
  {pages} {436--439} (\bibinfo {year} {1999})},\ \Eprint
  {http://arxiv.org/abs/9811053} {arXiv:9811053 [quant-ph]} \BibitemShut
  {NoStop}%
\bibitem [{\citenamefont {Du}\ \emph {et~al.}(2015)\citenamefont {Du},
  \citenamefont {Bai},\ and\ \citenamefont {Guo}}]{Du2015}%
  \BibitemOpen
  \bibfield  {author} {\bibinfo {author} {\bibfnamefont {Shuanping}\
  \bibnamefont {Du}}, \bibinfo {author} {\bibfnamefont {Zhaofang}\ \bibnamefont
  {Bai}}, \ and\ \bibinfo {author} {\bibfnamefont {Yu}~\bibnamefont {Guo}},\
  }\bibfield  {title} {\enquote {\bibinfo {title} {{Conditions for coherence
  transformations under incoherent operations}},}\ }\href {\doibase
  10.1103/PhysRevA.91.052120} {\bibfield  {journal} {\bibinfo  {journal}
  {Physical Review A}\ }\textbf {\bibinfo {volume} {91}},\ \bibinfo {pages}
  {052120} (\bibinfo {year} {2015})},\ \Eprint
  {http://arxiv.org/abs/1503.09176} {arXiv:1503.09176} \BibitemShut {NoStop}%
\bibitem [{\citenamefont {Dur}\ \emph {et~al.}(2000)\citenamefont {Dur},
  \citenamefont {Vidal},\ and\ \citenamefont {Cirac}}]{Dur2000}%
  \BibitemOpen
  \bibfield  {author} {\bibinfo {author} {\bibfnamefont {W.}~\bibnamefont
  {Dur}}, \bibinfo {author} {\bibfnamefont {G.}~\bibnamefont {Vidal}}, \ and\
  \bibinfo {author} {\bibfnamefont {J.~I.}\ \bibnamefont {Cirac}},\ }\bibfield
  {title} {\enquote {\bibinfo {title} {{Three qubits can be entangled in two
  inequivalent ways}},}\ }\href {\doibase 10.1103/PhysRevA.62.062314}
  {\bibfield  {journal} {\bibinfo  {journal} {Physical Review A - Atomic,
  Molecular, and Optical Physics}\ }\textbf {\bibinfo {volume} {62}},\ \bibinfo
  {pages} {062314--062311} (\bibinfo {year} {2000})},\ \Eprint
  {http://arxiv.org/abs/0005115} {arXiv:0005115 [quant-ph]} \BibitemShut
  {NoStop}%
\bibitem [{\citenamefont {Vrana}\ and\ \citenamefont
  {Christandl}(2015)}]{Vrana2015}%
  \BibitemOpen
  \bibfield  {author} {\bibinfo {author} {\bibfnamefont {P{\'{e}}ter}\
  \bibnamefont {Vrana}}\ and\ \bibinfo {author} {\bibfnamefont {Matthias}\
  \bibnamefont {Christandl}},\ }\bibfield  {title} {\enquote {\bibinfo {title}
  {{Asymptotic entanglement transformation between W and GHZ states}},}\ }\href
  {\doibase 10.1063/1.4908106} {\bibfield  {journal} {\bibinfo  {journal}
  {Journal of Mathematical Physics}\ }\textbf {\bibinfo {volume} {56}},\
  \bibinfo {pages} {022204} (\bibinfo {year} {2015})},\ \Eprint
  {http://arxiv.org/abs/1310.3244} {arXiv:1310.3244} \BibitemShut {NoStop}%
\bibitem [{\citenamefont {Walter}\ \emph {et~al.}(2016)\citenamefont {Walter},
  \citenamefont {Gross},\ and\ \citenamefont {Eisert}}]{Walter2016}%
  \BibitemOpen
  \bibfield  {author} {\bibinfo {author} {\bibfnamefont {Michael}\ \bibnamefont
  {Walter}}, \bibinfo {author} {\bibfnamefont {David}\ \bibnamefont {Gross}}, \
  and\ \bibinfo {author} {\bibfnamefont {Jens}\ \bibnamefont {Eisert}},\
  }\bibfield  {title} {\enquote {\bibinfo {title} {{Multi-partite
  entanglement}},}\ }\href {http://arxiv.org/abs/1612.02437} {\  (\bibinfo
  {year} {2016})},\ \Eprint {http://arxiv.org/abs/1612.02437}
  {arXiv:1612.02437} \BibitemShut {NoStop}%
\bibitem [{\citenamefont {Bennett}\ \emph
  {et~al.}(1996{\natexlab{a}})\citenamefont {Bennett}, \citenamefont
  {Brassard}, \citenamefont {Popescu}, \citenamefont {Schumacher},
  \citenamefont {Smolin},\ and\ \citenamefont {Wootters}}]{Bennett1995}%
  \BibitemOpen
  \bibfield  {author} {\bibinfo {author} {\bibfnamefont {Charles~H.}\
  \bibnamefont {Bennett}}, \bibinfo {author} {\bibfnamefont {Gilles}\
  \bibnamefont {Brassard}}, \bibinfo {author} {\bibfnamefont {Sandu}\
  \bibnamefont {Popescu}}, \bibinfo {author} {\bibfnamefont {Benjamin}\
  \bibnamefont {Schumacher}}, \bibinfo {author} {\bibfnamefont {John~A.}\
  \bibnamefont {Smolin}}, \ and\ \bibinfo {author} {\bibfnamefont {William~K.}\
  \bibnamefont {Wootters}},\ }\bibfield  {title} {\enquote {\bibinfo {title}
  {{Purification of Noisy Entanglement and Faithful Teleportation via Noisy
  Channels}},}\ }\href {\doibase 10.1103/PhysRevLett.76.722} {\bibfield
  {journal} {\bibinfo  {journal} {Physical Review Letters}\ }\textbf {\bibinfo
  {volume} {76}},\ \bibinfo {pages} {722--725} (\bibinfo {year}
  {1996}{\natexlab{a}})},\ \Eprint {http://arxiv.org/abs/9511027}
  {arXiv:9511027 [quant-ph]} \BibitemShut {NoStop}%
\bibitem [{\citenamefont {Bennett}\ \emph
  {et~al.}(1996{\natexlab{b}})\citenamefont {Bennett}, \citenamefont
  {Bernstein}, \citenamefont {Popescu},\ and\ \citenamefont
  {Schumacher}}]{Bennett1996}%
  \BibitemOpen
  \bibfield  {author} {\bibinfo {author} {\bibfnamefont {Charles~H.}\
  \bibnamefont {Bennett}}, \bibinfo {author} {\bibfnamefont {Herbert~J.}\
  \bibnamefont {Bernstein}}, \bibinfo {author} {\bibfnamefont {Sandu}\
  \bibnamefont {Popescu}}, \ and\ \bibinfo {author} {\bibfnamefont {Benjamin}\
  \bibnamefont {Schumacher}},\ }\bibfield  {title} {\enquote {\bibinfo {title}
  {{Concentrating partial entanglement by local operations}},}\ }\href
  {\doibase 10.1103/PhysRevA.53.2046} {\bibfield  {journal} {\bibinfo
  {journal} {Physical Review A}\ }\textbf {\bibinfo {volume} {53}},\ \bibinfo
  {pages} {2046--2052} (\bibinfo {year} {1996}{\natexlab{b}})},\ \Eprint
  {http://arxiv.org/abs/9511030} {arXiv:9511030 [quant-ph]} \BibitemShut
  {NoStop}%
\bibitem [{\citenamefont {Bennett}\ \emph
  {et~al.}(1996{\natexlab{c}})\citenamefont {Bennett}, \citenamefont
  {DiVincenzo}, \citenamefont {Smolin},\ and\ \citenamefont
  {Wootters}}]{Bennett1996a}%
  \BibitemOpen
  \bibfield  {author} {\bibinfo {author} {\bibfnamefont {Charles~H.}\
  \bibnamefont {Bennett}}, \bibinfo {author} {\bibfnamefont {David~P.}\
  \bibnamefont {DiVincenzo}}, \bibinfo {author} {\bibfnamefont {John~A.}\
  \bibnamefont {Smolin}}, \ and\ \bibinfo {author} {\bibfnamefont {William~K.}\
  \bibnamefont {Wootters}},\ }\bibfield  {title} {\enquote {\bibinfo {title}
  {{Mixed-state entanglement and quantum error correction}},}\ }\href {\doibase
  10.1103/PhysRevA.54.3824} {\bibfield  {journal} {\bibinfo  {journal}
  {Physical Review A}\ }\textbf {\bibinfo {volume} {54}},\ \bibinfo {pages}
  {3824--3851} (\bibinfo {year} {1996}{\natexlab{c}})},\ \Eprint
  {http://arxiv.org/abs/9604024} {arXiv:9604024 [quant-ph]} \BibitemShut
  {NoStop}%
\bibitem [{\citenamefont {Jia}\ and\ \citenamefont
  {Sakharwade}(2018)}]{Jia2018}%
  \BibitemOpen
  \bibfield  {author} {\bibinfo {author} {\bibfnamefont {Ding}\ \bibnamefont
  {Jia}}\ and\ \bibinfo {author} {\bibfnamefont {Nitica}\ \bibnamefont
  {Sakharwade}},\ }\bibfield  {title} {\enquote {\bibinfo {title} {{Tensor
  products of process matrices with indefinite causal structure}},}\ }\href
  {\doibase 10.1103/PhysRevA.97.032110} {\bibfield  {journal} {\bibinfo
  {journal} {Physical Review A}\ }\textbf {\bibinfo {volume} {97}},\ \bibinfo
  {pages} {032110} (\bibinfo {year} {2018})},\ \Eprint
  {http://arxiv.org/abs/1706.05532} {arXiv:1706.05532} \BibitemShut {NoStop}%
\bibitem [{\citenamefont {Gu{\'{e}}rin}\ \emph {et~al.}(2018)\citenamefont
  {Gu{\'{e}}rin}, \citenamefont {Krumm}, \citenamefont {Budroni},\ and\
  \citenamefont {Brukner}}]{Guerin2018}%
  \BibitemOpen
  \bibfield  {author} {\bibinfo {author} {\bibfnamefont {Philippe~Allard}\
  \bibnamefont {Gu{\'{e}}rin}}, \bibinfo {author} {\bibfnamefont {Marius}\
  \bibnamefont {Krumm}}, \bibinfo {author} {\bibfnamefont {Costantino}\
  \bibnamefont {Budroni}}, \ and\ \bibinfo {author} {\bibfnamefont
  {{\v{C}}aslav}\ \bibnamefont {Brukner}},\ }\href
  {http://arxiv.org/abs/1806.10374} {\enquote {\bibinfo {title} {{Composition
  rules for quantum processes: a no-go theorem}},}\ } (\bibinfo {year}
  {2018}),\ \Eprint {http://arxiv.org/abs/1806.10374} {arXiv:1806.10374}
  \BibitemShut {NoStop}%
\bibitem [{\citenamefont {Feix}\ and\ \citenamefont
  {Brukner}(2017)}]{Feix2017}%
  \BibitemOpen
  \bibfield  {author} {\bibinfo {author} {\bibfnamefont {Adrien}\ \bibnamefont
  {Feix}}\ and\ \bibinfo {author} {\bibfnamefont {{\v{C}}aslav}\ \bibnamefont
  {Brukner}},\ }\bibfield  {title} {\enquote {\bibinfo {title} {{Quantum
  superpositions of `common-cause' and `direct-cause' causal structures}},}\
  }\href {\doibase 10.1088/1367-2630/aa9b1a} {\bibfield  {journal} {\bibinfo
  {journal} {New Journal of Physics}\ }\textbf {\bibinfo {volume} {19}},\
  \bibinfo {pages} {123028} (\bibinfo {year} {2017})},\ \Eprint
  {http://arxiv.org/abs/1606.09241} {arXiv:1606.09241} \BibitemShut {NoStop}%
\bibitem [{\citenamefont {Zych}(2017)}]{ZychPhD}%
  \BibitemOpen
  \bibfield  {author} {\bibinfo {author} {\bibfnamefont {Magdalena}\
  \bibnamefont {Zych}},\ }\emph {\bibinfo {title} {{Quantum Systems under
  Gravitational Time Dilation}}},\ \href@noop {} {Ph.D. thesis} (\bibinfo
  {year} {2017})\BibitemShut {NoStop}%
\bibitem [{\citenamefont {Vidal}(1999)}]{Vidal1999}%
  \BibitemOpen
  \bibfield  {author} {\bibinfo {author} {\bibfnamefont {Guifr{\'{e}}}\
  \bibnamefont {Vidal}},\ }\bibfield  {title} {\enquote {\bibinfo {title}
  {{Entanglement of Pure States for a Single Copy}},}\ }\href {\doibase
  10.1103/PhysRevLett.83.1046} {\bibfield  {journal} {\bibinfo  {journal}
  {Physical Review Letters}\ }\textbf {\bibinfo {volume} {83}},\ \bibinfo
  {pages} {1046--1049} (\bibinfo {year} {1999})},\ \Eprint
  {http://arxiv.org/abs/9902033} {arXiv:9902033 [quant-ph]} \BibitemShut
  {NoStop}%
\bibitem [{\citenamefont {Yuan}\ \emph {et~al.}(2015)\citenamefont {Yuan},
  \citenamefont {Zhou}, \citenamefont {Cao},\ and\ \citenamefont
  {Ma}}]{Yuan2015}%
  \BibitemOpen
  \bibfield  {author} {\bibinfo {author} {\bibfnamefont {Xiao}\ \bibnamefont
  {Yuan}}, \bibinfo {author} {\bibfnamefont {Hongyi}\ \bibnamefont {Zhou}},
  \bibinfo {author} {\bibfnamefont {Zhu}\ \bibnamefont {Cao}}, \ and\ \bibinfo
  {author} {\bibfnamefont {Xiongfeng}\ \bibnamefont {Ma}},\ }\bibfield  {title}
  {\enquote {\bibinfo {title} {{Intrinsic randomness as a measure of quantum
  coherence}},}\ }\href {\doibase 10.1103/PhysRevA.92.022124} {\bibfield
  {journal} {\bibinfo  {journal} {Physical Review A}\ }\textbf {\bibinfo
  {volume} {92}},\ \bibinfo {pages} {022124} (\bibinfo {year} {2015})},\
  \Eprint {http://arxiv.org/abs/1505.04032v1} {arXiv:1505.04032v1} \BibitemShut
  {NoStop}%
\bibitem [{\citenamefont {Winter}\ and\ \citenamefont
  {Yang}(2016)}]{Winter2016}%
  \BibitemOpen
  \bibfield  {author} {\bibinfo {author} {\bibfnamefont {Andreas}\ \bibnamefont
  {Winter}}\ and\ \bibinfo {author} {\bibfnamefont {Dong}\ \bibnamefont
  {Yang}},\ }\bibfield  {title} {\enquote {\bibinfo {title} {{Operational
  Resource Theory of Coherence}},}\ }\href {\doibase
  10.1103/PhysRevLett.116.120404} {\bibfield  {journal} {\bibinfo  {journal}
  {Physical Review Letters}\ }\textbf {\bibinfo {volume} {116}},\ \bibinfo
  {pages} {120404} (\bibinfo {year} {2016})},\ \Eprint
  {http://arxiv.org/abs/1506.07975} {arXiv:1506.07975} \BibitemShut {NoStop}%
\bibitem [{\citenamefont {Popescu}\ and\ \citenamefont
  {Rohrlich}(1994)}]{Popescu1994}%
  \BibitemOpen
  \bibfield  {author} {\bibinfo {author} {\bibfnamefont {Sandu}\ \bibnamefont
  {Popescu}}\ and\ \bibinfo {author} {\bibfnamefont {Daniel}\ \bibnamefont
  {Rohrlich}},\ }\bibfield  {title} {\enquote {\bibinfo {title} {{Quantum
  nonlocality as an axiom}},}\ }\href {\doibase 10.1007/BF02058098} {\bibfield
  {journal} {\bibinfo  {journal} {Found. Phys.}\ }\textbf {\bibinfo {volume}
  {24}},\ \bibinfo {pages} {379--385} (\bibinfo {year} {1994})}\BibitemShut
  {NoStop}%
\bibitem [{\citenamefont {Marshall}\ \emph {et~al.}(2011)\citenamefont
  {Marshall}, \citenamefont {Olkin},\ and\ \citenamefont
  {Arnold}}]{Marshall2011}%
  \BibitemOpen
  \bibfield  {author} {\bibinfo {author} {\bibfnamefont {Albert~W.}\
  \bibnamefont {Marshall}}, \bibinfo {author} {\bibfnamefont {Ingram}\
  \bibnamefont {Olkin}}, \ and\ \bibinfo {author} {\bibfnamefont {Barry~C.}\
  \bibnamefont {Arnold}},\ }\href {\doibase 10.1007/978-0-387-68276-1} {\emph
  {\bibinfo {title} {{Inequalities: Theory of Majorization and Its
  Applications}}}},\ Springer Series in Statistics\ (\bibinfo  {publisher}
  {Springer New York},\ \bibinfo {address} {New York, NY},\ \bibinfo {year}
  {2011})\BibitemShut {NoStop}%
\bibitem [{\citenamefont {Ash}(1965)}]{Ash1965}%
  \BibitemOpen
  \bibfield  {author} {\bibinfo {author} {\bibfnamefont {Robert~B}\
  \bibnamefont {Ash}},\ }\href@noop {} {\emph {\bibinfo {title} {{Information
  Theory}}}},\ \bibinfo {edition} {1st}\ ed.\ (\bibinfo  {publisher} {Dover},\
  \bibinfo {year} {1965})\BibitemShut {NoStop}%
\bibitem [{\citenamefont {Cover}\ and\ \citenamefont
  {Thomas}(1991)}]{Cover1991}%
  \BibitemOpen
  \bibfield  {author} {\bibinfo {author} {\bibfnamefont {Thomas~M}\
  \bibnamefont {Cover}}\ and\ \bibinfo {author} {\bibfnamefont {Joy~a}\
  \bibnamefont {Thomas}},\ }\href {\doibase 10.1002/0471200611} {\emph
  {\bibinfo {title} {{Elements of Information Theory}}}},\ Wiley Series in
  Telecommunications\ (\bibinfo  {publisher} {John Wiley {\&} Sons, Inc.},\
  \bibinfo {address} {New York, USA},\ \bibinfo {year} {1991})\BibitemShut
  {NoStop}%
\end{thebibliography}%
%\bibliography{../../../Mendeley/Causal,../../../Mendeley/Steering}

%%%%%%%%%%%%%%%%%%%%%%%%%%%%%%%%%%%%%%%%%%
%%%%%%%%%%%%%%%%%%%%%%%%%%%%%%%%%%%%%%%%%%
\appendix{

%%%%%%%%%%%%%%%%%%%%%%%%%%%%%%%%%%%%%%%%%%
%%%%%%%%%%%%%%%%%%%%%%%%%%%%%%%%%%%%%%%%%%
\section{Conditions for $W$ and causal orders}
\label{app:norm_cond_W}

In order to formally describe many statements in the Appendices, we need to define an operation denoted by a subindex preceding an operator \cite{Araujo2015}: 
\begin{equation}
_XW:=(1/d_X) \ \mathbb1_X\!\otimes\!\Tr_XW \ ,
\label{eq:presub}
\end{equation}
where $d_X$ is the dimension of the arbitrary subspace $\mathsf H_X$. This operation replaces the original action of $W$ on subspace $\mathsf H_X$ by a trivial (and fully decorrelated) term.

For $W$ to be a valid process ($W\in\mathsf P$), the composition (link product) of $W\in\mathsf B(\mathsf H_{P}\otimes\mathsf H_{A_O}\otimes\mathsf H_{B_O}\otimes\mathsf H_{F}\otimes\mathsf H_{A_I}\otimes\mathsf H_{B_I}\otimes\mathsf H_{C})$ with any possible instruments applied by Alice and Bob (even those exploiting entangled ancillas between $A$ and $B$), must yield a valid CJ state on $\mathsf B(\mathsf H_{P}\otimes\mathsf H_{F}\otimes\mathsf H_{C})$ describing a CP TP global map from $\mathsf B(\mathsf H_{P})$ to $\mathsf B(\mathsf H_{F}\otimes\mathsf H_{C})$ \cite{Oreshkov2012,Araujo2015,Oreshkov2016}. These conditions hold iff $W$ obeys \cite{Araujo2015,Araujo2017} 
\begin{subequations}
\begin{align}
W &\geq 0 \label{eq:condPMpositive}\\
\Tr W &= d_P d_{A_O} d_{B_O} , \label{eq:condPMdim}\\
{}_{B_IB_OCF} W &= {}_{A_O} {}_{B_IB_OCF} W, \label{eq:condPMAOsignal}\\
{}_{A_IA_OCF} W &= {}_{B_O} {}_{A_IA_OCF} W, \label{eq:condPMBOsignal}\\
{}_{CF} W &= {}_{A_OCF} W + {}_{B_OCF} W - {}_{A_OB_OCF} W \label{eq:condPMnorm}\\
{}_{\boxtimes} {}_{CF}W& = {}_{P\boxtimes CF}W \ , \label{eq:condPMsignalP}
\end{align}
where $\boxtimes:=A_IA_OB_IB_O$. \label{eq:condPM}\end{subequations}
We can interpret some of these relations in terms of no-signaling restrictions. Let us take Eq.~\eqref{eq:condPMAOsignal} as an example. By stating that $_{B_IB_OCF}W$ is unchanged by replacing $A_O$ with a trivial, decorrelated input, one concludes that $A_O$ may only be correlated with the variables $B_I$, $B_O$, $C$, and $F$. 
As such, it cannot signal to any other variables, such as $A_I$. Since $A_I$ is in the past of $A_O$ and such signaling would yield causal loops, it is only natural that Eq.~\eqref{eq:condPMAOsignal} is a necessary condition for the validity of $W$. 
In fact, the causal order between $A$ and $B$ is formally defined in terms of such no-signaling restrictions. In general, a  process $W_{A\to B}$ is compatible \cite{Araujo2015} with the causal order $A\to B$ if, and only if,
\begin{subequations}\begin{equation}
{{}_{F}W_{A\to B} = {}_{B_O}{}_{F}W_{A\to B}} \ ,
\label{eq:AtoB}
\end{equation} as obeyed by $\kett{\mathbb 1_0}$. This relation only allows signaling from $B_O$ to $F$, precluding any signaling from $B_O$ to a variable belonging to lab $A$. In other words, this relation precludes signaling from $B$ to $A$, as expected. Analogously, a process $W_{B\to A}$ is compatible with the causal order $B\to A$ if, and only if, ${{}_{F}W_{A\to B} = {}_{A_O}{}_{F}W_{A\to B}}$,
\end{subequations}
a relation obeyed by $\kett{\mathbb 1_1}$ that forbids signaling from $A_O$ to lab $B$ in its causal past. 

%%%%%%%%%%%%%%%%%%%%%%%%%%%%%%%%%%%%%%%%%%
%%%%%%%%%%%%%%%%%%%%%%%%%%%%%%%%%%%%%%%%%%

\section{Proof of Theorem \ref{th:free_ops}}
\label{app:cond_free_V}

After presenting some useful preliminary results, we will break down the proof of Theorem \ref{th:free_ops} in two parts, that of $\mathsf{LOAE}$ (which uses the broader class $\mathsf{NSO}$) and that of $\mathsf{PLS}$. After proving that both classes map $\mathsf{P}$, $\mathsf{CS}$ and $\mathsf{NQC}$ onto themselves (Lemmas \ref{th:LOAEisfree} and \ref{th:PLSisfree}, respectively), Theorem \ref{th:free_ops} follows straightforwardly.
\subsection{Useful relations}
\noindent We will need the following result (``hopping''):
\begin{equation}
\Tr_X\left(_XW \ \ Y\right) = \Tr_X\left(W \ _XY\right) \ ,
\label{eq:hopping}
\end{equation}
that is, with respect to the inner product given by trace over (sub)space $X$, the operation given by subindex $_X$ is self-dual \cite{Araujo2015}. This is proven by starting
 from $\Tr_XW\ \Tr_XY$ and ``factoring out'' the second trace operator:
\begin{equation}
\frac1{d_X}\!\Tr_X\! W \Tr_X\! Y = \Tr_X\! \left(\frac{\mathbb1_X}{d_X}(\Tr_XW) Y\right) = \Tr_X(_XW \ Y) \ .
\label{eq:usproof}
\end{equation}
But the first expression in \eqref{eq:usproof} is completely symmetric on $W,Y$, so the same reasoning can be done with the first trace operator, yielding $\Tr_X(W \ _XY)$.

Additionally, given that $\mathcal{V}$ represents a CP TP map from $\mathsf B(\mathsf H_{A_I}\otimes\mathsf H_{A'_O}\otimes\mathsf H_{B_I}\otimes\mathsf H_{B'_O}\otimes\mathsf H_{C})$ to $\mathsf B(\mathsf H_{A'_I}\otimes\mathsf H_{A_O}\otimes\mathsf H_{B'_I}\otimes\mathsf H_{B_O}\otimes\mathsf H_{C'})$,  any valid $V$ must obey 
\begin{align}
\Tr V=& d_{A_I}d_{A'_O}d_{B_I}d_{B_O'}d_C \label{eq:condVdim}\\
 {}_{A_I'A_OB_I'B_OC'}V=& {}_{A_IA_O'B_IB_O'C \ A_I'A_OB_I'B_OC'}V \ . \label{eq:condVinputsoutputs}
\end{align}
Transformations of the form \eqref{eq:Vdecomposed} also separately obey
\begin{align}
							{}_{C'}V  &=			{}_{CC'}V \ , \label{eq:condVCseparately} \\
  {}_{A_I'A_OB_I'B_O}V  &= {}_{A_IA_O'B_IB_O' \ A_I'A_OB_I'B_O}V \ .
\label{eq:condVABseparately}
\end{align}

\subsection{LOAE and NSO}
In order to prove that $\mathsf{LOAE}$ is a class of free operations, we appeal to a broader class of \emph{nonsignaling operations}, $\mathsf{NSO}$, which forbids signaling from any of the $A$ variables of $V$ ($A_I,A_I',A_O',A_O$) to any of its $B$ variables  ($B_I,B_I',B_O',B_O$) and vice versa.
\begin{subequations}\begin{dfn}[Nonsignaling operations]
A process transformation $\mathcal{V}$ belongs to the class $\mathsf{NSO}$ if, and only if,
\label{def:NSO}
\begin{alignat}{4}
		{}_{A_O}\! V=&&			 {}_{A_O'A_O}\! V	\label{eq:condNSO_outputA}&\\
		{}_{B_O}\! V=&&			 {}_{B_O'B_O}\! V \label{eq:condNSO_outputB}&		\\
{}_{A_I'A_O}\! V=&&		{}_{A_IA_I'A_O}\! V	\label{eq:condNSO_inputA} &\\
{}_{B_I'B_O}\! V=&&		{}_{B_IB_I'B_O}\! V \label{eq:condNSO_inputB}	& \ .	
\end{alignat}
\end{dfn}
\noindent We notice that \label{eq:condNSO_AB}\end{subequations} Eqs.(\ref{eq:condNSO_outputA},\ref{eq:condNSO_outputB}) also exclude signaling from $A_O'$ ($B_O'$) to $A_I'$ ($B_I'$), preventing causal loops.
$\mathsf{NSO}$ is, in fact, a more general class than $\mathsf{LOAE}$ in two ways. First, it need not be separable in the $C|AB$ partition, allowing for coherent operations between $C$ and $AB$. Secondly, it allows for post-quantum resources, such as using a Popescu-Röhrlich box \cite{Popescu1994} to correlate outputs. Although mathematically well-defined, not all elements of $\mathsf{NSO}$ have a clear quantum-mechanical realization. However, the class $\mathsf{LOAE}$, which has a clear physical interpretation and parametrization, is shown to be a subset of $\mathsf{NSO}$, so that all properties proven for $\mathsf{NSO}$ are valid for $\mathsf{LOAE}$.
\begin{lem}$\mathsf{LOAE}\subseteq\mathsf{NSO}$
\label{th:LOAEinNSO}\end{lem}
\begin{proof}
We will show that $V_{AB}^{(j)}$ parametrized as in Eq. \eqref{eq:def_LOAE} obeys Eqs.\eqref{eq:condNSO_AB}. By linearity, the same will hold for $V=\sum_jV_C^{(j)}\otimes V_{AB}^{(j)}$. We first notice that invariance under the application of a subindex operator is equivalent to a trivial dependence on the corresponding partition, or $W={}_XW \Leftrightarrow W \propto \mathbb1_X$, i.e. $W$ is a tensor product of $\mathbb1_X$ and operators on the space of the remaining variables.

Let us begin by Eq.~\eqref{eq:condNSO_outputA}. Taking the partial trace $\Tr_{A_O}V_{AB}^{(j)}$ on Eq. \eqref{eq:def_LOAE}, the entire output space of $U^{(j)}_{A_O}:\mathsf H_{A'_O}\otimes\mathsf H_{\tilde A}\to\mathsf H_{A_O}\otimes\mathsf H_{\tilde A}$ is traced out.
In this case, the basis independence of the trace allows us to replace $U_{A_O}^{(j)}$ for an identity. The action on $\mathsf H_{A_O'}\otimes\mathsf H_{A_O}$ then reduces to
$\kettbbra{\mathbb1}{\mathbb1}_{A_O'A_O}$. Since 
$\Tr_{A_O}\left[\kettbbra{\mathbb1}{\mathbb1}_{A_O'A_O}\right] = \mathbb1_{A_O'}$,
then $\Tr_{A_O}V_{AB}^{(j)}\propto\mathbb1_{A_O'}$, or ${}_{A_O}V_{AB}^{(j)}={}_{A_O'A_O}V_{AB}^{(j)}$. The demonstration of Eq.~\eqref{eq:condNSO_outputB} is analogous, interchanging $A$ and $B$.

Next, we derive Eq.~\eqref{eq:condNSO_inputA} for LOAEs. This time, we take the partial trace $\Tr_{A_I'A_O}V_{AB}^{(j)}$ on Eq. \eqref{eq:def_LOAE}, and, analogously,
$U_{A_O}^{(j)}\otimes U_{A_I}^{(j)}$ can be replaced by an identity map due to the basis independence of the trace. The action on $\mathsf H_{A_I}\otimes\mathsf H_{A_I'}$ reduces to
$\kettbbra{\mathbb1}{\mathbb1}_{A_IA_I'}$ and since 
$\Tr_{A_I'}\left[\kettbbra{\mathbb1}{\mathbb1}_{A_IA_I'}\right] = \mathbb1_{A_I}$, then
$\Tr_{A_I'A_O}V_{AB}^{(j)}\propto\mathbb1_{A_I}$, or ${}_{A_I'A_O}V_{AB}^{(j)}={}_{A_IA_I'A_O}V_{AB}^{(j)}$. The demonstration of Eq.~\eqref{eq:condNSO_inputB} is analogous, interchanging $A$ and $B$.
\end{proof}

\begin{lem}
$\mathsf{NSO}$ is a class of free operations of causal nonseparability and of quantum control of causal orders, i.e., it maps $\mathsf{P}$, $\mathsf{CS}$ and $\mathsf{NQC}$ onto themselves.
\label{th:NSOisfree}
\end{lem}

\begin{proof}
\begin{subequations}
Let us first show that Eqs. \eqref{eq:condNSO_AB} preserve the causal orders $A\to B$ and $B\to A$, i.e., show that
\label{eq:preservABBA}
\begin{align}
{}_{F}W = {}_{B_OF}W &\Rightarrow {}_{F}(V*W) = {}_{B'_OF}(V*W) \label{eq:preservAB}\\
{}_{F}W = {}_{A_OF}W &\Rightarrow {}_{F}(V*W) = {}_{A'_OF}(V*W) \label{eq:preservBA}\ .
\end{align}\end{subequations}
Given that $V*W = \Tr_{\boxtimes C}[W \, V^{T_\boxtimes T_C}]$, where $\boxtimes:=A_IA_OB_IB_O$, \begin{subequations} we prove Eq. \eqref{eq:preservAB} via the equations indicated in the parentheses and the ``hopping'' result, Eq.\eqref{eq:hopping}:
\begin{align}
		{}_{F} \left(V*W\right) \,= \,\, &\Tr_{\boxtimes C}[{}_{F} W \ V^{T_\boxtimes T_C}] \label{eq:LOAEpreserveNSa}\\
\,\substack{\eqref{eq:preservAB}\\=} \,\, &\Tr_{\boxtimes C}[{}_{B_OF} W\ V^{T_\boxtimes T_C}] \nonumber\\%\label{eq:LOAEpreserveNSb}\\
		{}_{F} \left(V*W\right) \,\substack{{\rm hop}\\=}\, &\Tr_{\boxtimes C}[{}_{F} W\,\,{}_{B_O}V^{T_\boxtimes T_C}] \label{eq:LOAEpreserveNSc}\\
\,\substack{\eqref{eq:condNSO_outputB}\\=}\, &\Tr_{\boxtimes C}[{}_{F} W \ {}_{B_O' B_O}V^{T_\boxtimes T_C}] \nonumber\\%\label{eq:LOAEpreserveNSd}\\
\,=\, \ &{}_{B_O'}\Tr_{\boxtimes C}[{}_{F} W\ {}_{B_O}V^{T_\boxtimes T_C}]  \nonumber\\ %\label{eq:LOAEpreserveNSe}\\
\,\substack{\eqref{eq:LOAEpreserveNSc}\\=}\ &{}_{B_O'}{}_F(V*W) . \label{eq:LOAEpreserveNSf}
\end{align}\label{eq:LOAEpreserveNS}\end{subequations}
Eq. \eqref{eq:preservBA} is proven analogously. As such, from Eq. \eqref{eq:caussep} we see that for $\mathsf{NSO}$ preserves $\mathsf{CS}$. Because of the separable structure of Eq.\eqref{eq:Vdecomposed}, $\mathsf{NSO}$ also preserves $\mathsf{S}$. By linearity, $\mathrm{Conv}(\mathsf{CS}\cup\mathsf{S})=\mathsf{NQC}$ is preserved as well. 

We are left with the lengthier task of showing that $\mathsf{NSO}$ preserves $\mathsf P$, i.e., showing that the  validity constraints of Eq. \eqref{eq:condPM} are preserved under $\mathsf{NSO}$. 
The positivity constraint \eqref{eq:condPMpositive} is straightforward, since the link product preserves positivity.

The dimensionality constraint \eqref{eq:condPMdim} will initially be shown to be preserved by the simpler case of $W$ in a causal order ${A\to B}$ ($_FW={}_{B_OF}W$). We calculate ${}_{PF\boxtimes'C'}(V*W)$, where $\boxtimes':=A_I'A_O'B_I'B_O'$, from which the trace can be taken:\medskip\\
\begin{subequations}\label{eq:trpreserv}
${}_{PF\boxtimes'C'}\Tr_{\boxtimes C}(W  V^{T_\boxtimes T_C}) =$\vspace{-.6\baselineskip}
\begin{alignat}{4}
&=&{}_{P}\Tr_{\boxtimes C}(&&{}_FW \ && {}_{\boxtimes'C'}V^{T_\boxtimes T_C})\label{eq:trpreserva}\\
&\substack{\eqref{eq:condVCseparately}\\=}&	{}_{P}\Tr_{\boxtimes C}	(&&{}_{B_OF}W \ && {}_{C\boxtimes'C'}V^{T_\boxtimes T_C})  \label{eq:trpreservb}\\
& \substack{\rm hop\\=} &					{}_{P}\Tr_{\boxtimes C}	(&&{}_{CF}W \ && {}_{B_O}{}_{\boxtimes'C'}V^{T_\boxtimes T_C}) \label{eq:trpreservc}\\ 
&	\substack{\eqref{eq:condNSO_inputB}\\=}&{}_{P}\Tr_{\boxtimes C}(&&{}_{CF}W \ && {}_{B_IB_O}{}_{\boxtimes'C'}V^{T_\boxtimes T_C})\label{eq:trpreservd} \\
&	\substack{\rm hop\\=} &					{}_{P}\Tr_{\boxtimes C}	(&&{}_{B_IB_OCF}W \ && {}_{\boxtimes'C'}V^{T_\boxtimes T_C})\label{eq:trpreserve}\\
&	\substack{\eqref{eq:condPMAOsignal}\\=}&{}_{P}\Tr_{\boxtimes C}	(&&{}_{A_OB_IB_OCF}W \ && {}_{\boxtimes'C'}V^{T_\boxtimes T_C})\label{eq:trpreservf}\\
&	\substack{\rm hop\\=}&{}_{P}\Tr_{\boxtimes C}	(&&{}_{B_IB_OCF}W \ &&{}_{A_O}{}_{\boxtimes'C'}V^{T_\boxtimes T_C})\label{eq:trpreservg}\\
&	\substack{\eqref{eq:condNSO_inputA}\\=}&{}_{P}\Tr_{\boxtimes C}(&&{}_{B_IB_OCF}W \ && {}_{A_IA_O}{}_{\boxtimes'C'}V^{T_\boxtimes T_C})\label{eq:trpreservh}\\
&	\substack{\rm hop\\=}&{}_{P}\Tr_{\boxtimes C}	(&&{}_{F}W \ &&{}_{\boxtimes C}{}_{\boxtimes'C'}V^{T_\boxtimes T_C})\label{eq:trpreservi}\\
&	=&{}_{PF\boxtimes'C'}&&\Tr_{\boxtimes C}(W \ &&{}_{\boxtimes C}V^{T_\boxtimes T_C}) \ . \label{eq:trpreservj2}
\end{alignat}
Taking the trace of this expression, using the definition \eqref{eq:presub} and the fact that the subindex operator is TP, we find
\begin{align}
\Tr&(V*W) =\nonumber\\
&= \Tr_{PF\boxtimes\boxtimes'CC'} \left(W \ \ \frac{\mathbb1_\boxtimes\mathbb1_C}{d_\boxtimes d_C} {}_\otimes \Tr_{\boxtimes C}V^{T_\boxtimes T_C}\right) \label{eq:trpreservk}\\
& = (\Tr_{PF\boxtimes C}W)\frac{1}{d_\boxtimes d_C} (\Tr_{\boxtimes\boxtimes'CC'}V^{T_\boxtimes T_C}) \label{eq:trpreservl}\\
&	= d_Pd_{A_O}d_{B_O} \frac{1}{d_\boxtimes d_C} d_{A_O'}d_{B_O'}d_{A_I}d_{B_I}d_C \ , \label{eq:trpreservm}
\end{align}
where in the last line Eqs. (\ref{eq:condPMdim},\ref{eq:condVdim}) were used. We then obtain $\Tr(V*W)=d_P d_{A_O'} d_{B_O'}$, as desired.
If $W$ is not in the causal order $A\to B$, the demonstration changes as follows: instead of Eq. \eqref{eq:trpreservb}, we obtain a term identical to Eq. \eqref{eq:trpreservc} but without the $B_O$ subscript.
Using
Eq. \eqref{eq:condPMnorm},\end{subequations}
\begin{align}
{}_{PF\boxtimes'C'}(V*W)&= 
       {}_P\Tr_{\boxtimes C}	({}_{A_OCF}W 		\ {}_{\boxtimes'C'}V^{T_\boxtimes T_C}) \nonumber\\
		& +{}_P\Tr_{\boxtimes C}	({}_{B_OCF}W 		\ {}_{\boxtimes'C'}V^{T_\boxtimes T_C}) \label{eq:trpreservGen} \\
		& -{}_P\Tr_{\boxtimes C}	({}_{A_OB_OCF}W 	\ {}_{\boxtimes'C'}V^{T_\boxtimes T_C}) \ . \nonumber
\end{align}
For the first term on the right-hand side, the demonstration can be carried out as above, switching the roles of $A$ and $B$. For the second and third, all calculations on Eqs.\eqref{eq:trpreserv} are valid. As such, all three terms are equal to $d_P d_{A_O'} d_{B_O'}$. Given their signs, we obtain $\Tr(V*W)=d_P d_{A_O'} d_{B_O'}$, as before. 

Let us now prove that Eq. \eqref{eq:condPMAOsignal} is preserved. Once again we begin by assuming ${}_FW={}_{B_OF}W$ and afterwards lift that assumption. Firstly, ${}_{B_I'B_O'C'F}(V * W)$ can be shown to equal ${}_{B_I'B_O'C'}\Tr_{\boxtimes C}({}_{A_OB_IB_OCF}W \ V^{T_\boxtimes T_C})$ using Eqs. (\ref{eq:condNSO_inputB},\ref{eq:condPMAOsignal}) as done in Eqs.(\ref{eq:trpreserva}-\ref{eq:trpreservf}). So we can write ${}_{B_I'B_O'C'F}(V * W)$ as
\begin{subequations}\label{eq:nolooppreserv}
\begin{alignat}{4}
&&{}_{B_I'B_O'C'}\!\Tr_{\boxtimes C}(&&{}_{A_OB_IB_OCF}W && V^{T_\boxtimes T_C})\label{eq:nolooppreserva}\\
& \substack{\rm hop\\=} &	{}_{B_I'B_O'C'}\!\Tr_{\boxtimes C}(&&{}_{B_IB_OCF}W && {}_{A_O}V^{T_\boxtimes T_C}) \label{eq:nolooppreservb}\\ 
&	\substack{\eqref{eq:condNSO_outputA}\\=}&{}_{B_I'B_O'C'}\!\Tr_{\boxtimes C}(&&{}_{B_IB_OCF}W && {}_{A_O'A_O}V^{T_\boxtimes T_C})\label{eq:nolooppreservc} \\
&	= &	{}_{A_O'B_I'B_O'C'}\!\Tr_{\boxtimes C}(&&{}_{B_IB_OCF}W && {}_{A_O}V^{T_\boxtimes T_C})\label{eq:nolooppreservd}
\end{alignat}\end{subequations}
and comparing Eq. \eqref{eq:nolooppreservd} with Eq. \eqref{eq:nolooppreservb}, we find that ${}_{B_I'B_O'C'F}(V * W)={}_{A_O'B_I'B_O'C'F}(V * W)$, as desired. 
If $W$ is not compatible with $A\to B$, we arrive via Eqs. \eqref{eq:condPMnorm}, \eqref{eq:condVCseparately} at the three-term expression
\begin{align}
{}_{C'F}(V*W)= {}_{C'}\Tr_{\boxtimes C}({}_{A_OCF}W\ V^{T_\boxtimes T_C}) &\label{eq:nolooppreservGen} \\
					 			+		{}_{C'}\Tr_{\boxtimes C}({}_{B_OCF}W\ V^{T_\boxtimes T_C}) &\nonumber \\
								-   {}_{C'}\Tr_{\boxtimes C}({}_{A_OB_OCF}W\ V^{T_\boxtimes T_C}) &\ . \nonumber
\end{align}
The steps in Eqs. (\ref{eq:trpreserva}-\ref{eq:trpreservf}) apply to the second term on the right-hand side. To the first and third terms on the right-hand side, we can directly apply the steps in Eqs.(\ref{eq:nolooppreserv}). All three terms, then, equal ${}_{A_O'B_I'B_O'F}(W * V)$ and, due to their signs, ${}_{B_I'B_O'F}(W * V)={}_{A_O'B_I'B_O'F}(W * V)$ in general. The demonstration that condition \eqref{eq:condPMBOsignal} is preserved follows analogously, switching $A$ and $B$ throughout.

The preservation of condition \eqref{eq:condPMnorm} is demonstrated as follows. First let us notice that ${}_{C'F}(V*W)=({}_{C'}V*{}_{F}W)=({}_{C'}V*{}_{CF}W)={}_{C'}(V*{}_{CF}W)$, then
\begin{subequations}\label{eq:normpreserv}
\begin{align}
{}_{A_O'C'F}(V*W) + {}_{B_O'C'F}(V*W)& - {}_{A_O'B_O'C'F}(V*W) =  \nonumber \\
={}_{C'}\Tr_{\boxtimes C}[({}_{A_OCF}W+{}_{B_OCF}&W-{}_{A_OB_OCF}W) \nonumber\\
 ({}_{A_O'}V + {}_{B_O'}V - {}_{A_O'B_O'}& V)^{T_\boxtimes T_C}] \ ,   \label{eq:normpreserva}
\end{align}
where Eq. \eqref{eq:condPMnorm} has been applied to $W$. On the nine resulting terms we apply Eqs. (\ref{eq:hopping}, \ref{eq:condNSO_outputA},\ref{eq:condNSO_outputB}) to eliminate $A_O'$, $B_O'$ whenever possible, and we see that six of these terms  cancel out, leading to
\begin{align}
&{}_{A_O'C'F}(V*W) \!+\! {}_{B_O'C'F}(V*W) \!-\! {}_{A_O'B_O'C'F}(V*W) \!= \nonumber\\
&={}_{C'}\Tr_{\boxtimes C}[{}_{CF}W \ ({}_{A_O}V \!+\! {}_{B_O}V \!-\! {}_{A_OB_O} V)^{T_\boxtimes T_C}]  \\
&\substack{{\rm hop}\\=}{}_{C'}\Tr_{\boxtimes C}[({}_{A_OCF}W \!+\! {}_{B_OCF}W \!-\! {}_{A_OB_OCF} W) \ V^{T_\boxtimes T_C}]  \nonumber\\
&\substack{\eqref{eq:condPMnorm}\\=}{}_{C'}(V*{}_{CF}W)= {}_{C'F}(V*W) \ . 
\end{align}
\end{subequations}

Finally, to prove that condition \eqref{eq:condPMsignalP} is preserved, we once again begin by assuming $W$ compatible with $A\to B$ (${}_FW={}_{B_OF}W$) and later lift the assumption:
\begin{subequations}\label{eq:Psignalpreserv}
\begin{alignat}{4}
&&{}_{\boxtimes'C'F}(V*W)=&&&&\nonumber\\
&&=\Tr_{\boxtimes C}(&&{}_{F}W \ && {}_{\boxtimes'C'}V^{T_\boxtimes T_C})\label{eq:Psignalpreserva}\\
&&\substack{\eqref{eq:condVCseparately}\\=}\Tr_{\boxtimes C}(&&{}_{B_OF}W \ && {}_{C\boxtimes'C'}V^{T_\boxtimes T_C})\label{eq:Psignalpreservb}\\
&& \substack{\rm hop\\=}	\Tr_{\boxtimes C}(&&{}_{CF}W \ && {}_{B_O\boxtimes'C'}V^{T_\boxtimes T_C}) \label{eq:Psignalpreservc}\\ 
&&	\substack{\eqref{eq:condNSO_inputB}\\=}\Tr_{\boxtimes C}(&&{}_{CF}W \ &&\!\! {}_{B_IB_O\boxtimes'C'}V^{T_\boxtimes T_C})\label{eq:Psignalpreservd} \\
&& \substack{\rm hop\\=}\Tr_{\boxtimes C}(&&{}_{B_IB_OCF}W \ && {}_{\boxtimes'C'}V^{T_\boxtimes T_C})\label{eq:Psignalpreserve} \\
&& \substack{\eqref{eq:condPMAOsignal}\\=}\Tr_{\boxtimes C}(&&{}_{A_OB_IB_OCF}W \ && {}_{\boxtimes'C'}V^{T_\boxtimes T_C})\label{eq:Psignalpreservf}\\
&& \substack{\rm hop\\=}\Tr_{\boxtimes C}(&&{}_{B_IB_OCF}W \ && {}_{A_O\boxtimes'C'}V^{T_\boxtimes T_C})\label{eq:Psignalpreservg} \\
&& \substack{\eqref{eq:condNSO_inputA}\\=}\Tr_{\boxtimes C}(&&{}_{B_IB_OCF}W \ &&\!\! {}_{A_IA_O\boxtimes'C'}V^{T_\boxtimes T_C})\label{eq:Psignalpreservh}\\
&& \substack{\rm hop\\=}\Tr_{\boxtimes C}(&&{}_{\boxtimes CF}W \ && {}_{\boxtimes'C'}V^{T_\boxtimes T_C})\label{eq:Psignalpreservi}\\
&& \substack{\eqref{eq:condPMsignalP}\\=}\Tr_{\boxtimes C}(&&{}_{P\boxtimes CF}W \ && {}_{\boxtimes'C'}V^{T_\boxtimes T_C})\label{eq:Psignalpreservj}\\ 
&& = {}_{P\boxtimes'C'F}&&(V*W) \ \ \ &&  . \label{eq:Psignalpreservk}
\end{alignat}
If $W$ is not compatible with $A\to B$, we have, instead of Eq. \eqref{eq:Psignalpreservb}, the three-term expression [due to Eq. \eqref{eq:condPMnorm}]
\begin{align}
{}_{C'F}(V*W)&= {}_{C'}\Tr_{\boxtimes C}( {}_{A_OCF}W \ {}_{\boxtimes'}V^{T_\boxtimes T_C}) \nonumber \\
&               + {}_{C'}\Tr_{\boxtimes C}(   {}_{B_OCF}W \ {}_{\boxtimes'}V^{T_\boxtimes T_C}) 
\label{eq:PsignalpreservGen}\\
&				       	- {}_{C'}\Tr_{\boxtimes C}({}_{A_OB_OCF}W \ {}_{\boxtimes'}V^{T_\boxtimes T_C})\ . \nonumber
\end{align}
The calculation above can be done directly on the second term on the right-hand side and is also valid on the third. To the first term on the right-hand side, we can apply the same steps as above, but switching $A$ and $B$ throughout. All three terms, then, equal ${}_{P\boxtimes'C'F}(V*W)$ and, due to their signs, ${}_{\boxtimes'C'F}(V*W)={}_{P\boxtimes'C'F}(V*W)$.
\end{subequations}

We have then showed that $\mathsf{NSO}$ preserves $\mathsf P$, along with $\mathsf{CS}$ and $\mathsf{NQC}$.
\end{proof}
\noindent As a straightforward consequence of Lemmas \ref{th:LOAEinNSO} and \ref{th:NSOisfree},
\begin{cor}
$\mathsf{LOAE}$ is a class of free operations of causal nonseparability and of quantum control of causal orders, i.e., it maps $\mathsf{P}$, $\mathsf{CS}$ and $\mathsf{NQC}$ onto themselves.
\label{th:LOAEisfree}\end{cor}

\subsection{Probabilistic Lab Swaps}

\begin{lem}
$\mathsf{PLS}$ is a class of free operations of causal nonseparability and of quantum control of causal orders, i.e., it maps $\mathsf{P}$, $\mathsf{CS}$ and $\mathsf{NQC}$ onto themselves.
\label{th:PLSisfree}\end{lem}
\begin{proof}
We begin by noticing that from Eqs.(\ref{eq:def_id},\ref{eq:condVdim},\ref{eq:condVinputsoutputs},\ref{eq:condNSO_AB}) that $\kettbbra{\mathbb 1_{AB}}{\mathbb 1_{AB}}\in\mathsf{NSO}$, and hence preserves $\mathsf P$, $\mathsf{CS}$. 
On the other hand $V_{AB}^{({\rm sw})}:=\kettbbra{\mathbb s_{AB}}{\mathbb s_{AB}}$ from Eq. \eqref{eq:def_swap} obeys the following signaling conditions [compare Eqs.\eqref{eq:condNSO_AB}]
\begin{subequations}\label{eq:condSwap_AB}
\begin{alignat}{4}
		{}_{A_O} V_{AB}^{({\rm sw})}=&&	   {}_{B_O'A_O} V_{AB}^{({\rm sw})}	\label{eq:condSwap_outputA}&\\
		{}_{B_O} V_{AB}^{({\rm sw})}=&&	   {}_{A_O'B_O} V_{AB}^{({\rm sw})}	\label{eq:condSwap_outputB}&\\
{}_{B_I'A_O} V_{AB}^{({\rm sw})}=&&	{}_{A_IB_I'A_O} V_{AB}^{({\rm sw})}	\label{eq:condSwap_inputA}&\\
{}_{A_I'B_O} V_{AB}^{({\rm sw})}=&&	{}_{B_IA_I'B_O} V_{AB}^{({\rm sw})}	\label{eq:condSwap_inputB}&\ .		
\end{alignat}\end{subequations} 
\begin{subequations}
As expected, $V_{AB}^{({\rm sw})}$ inverts the ordering $A\to B$ and $B\to A$, i.e., from Eqs. \eqref{eq:condSwap_AB} it follows that
\label{eq:swapABBA}
\begin{align}\begin{split}
{}_{F}W = {}_{B_OF}W & \Rightarrow\\ \Rightarrow {}_{F}(&V_{AB}^{({\rm sw})}*W) = {}_{A'_OF}(V_{AB}^{({\rm sw})}*W) \ ,\end{split}\label{eq:swapAB}\\
\begin{split}
{}_{F}W = {}_{A_OF}W & \Rightarrow\\ \Rightarrow {}_{F}(&V_{AB}^{({\rm sw})}*W) = {}_{B'_OF}(V_{AB}^{({\rm sw})}*W) \ , 	\end{split}\label{eq:swapBA}
\end{align}\end{subequations}
which can be shown following the steps in Eq. \eqref{eq:LOAEpreserveNS} with Eq. \eqref{eq:condSwap_AB} instead of Eq. \eqref{eq:condNSO_AB}. Most importantly, although the causal order is inverted, the existence of a well-defined causal order is preserved when $V_{AB}^{({\rm sw})}$ alone is applied, and so is $\mathsf{CS}$.
The preservation of $\mathsf{P}$ [Eqs.\eqref{eq:condPM}] by $V_{AB}^{({\rm sw})}$ is demonstrated analogously as shown above for $\mathsf{NSO}$, switching $A_I'\leftrightarrow B_I'$, $A_O'\leftrightarrow B_O'$.

Given that both $\kettbbra{\mathbb 1_{AB}}{\mathbb 1_{AB}}$ and $\kettbbra{\mathbb s_{AB}}{\mathbb s_{AB}}$ preserve $\mathsf{CS}$ and $\mathsf{P}$, so does a probabilistic lab swap, since convex mixtures preserve Eqs.(\ref{eq:caussep},\ref{eq:condPM}). The separable form of Eq.\eqref{eq:Vdecomposed} guarantees preservation of $\mathsf S$, hence $\mathsf{PLS}$ preserves $\mathsf P$, $\mathsf{CS}$, and $\mathsf{NQC}$.
\end{proof}

%%%%%%%%%%%%%%%%%%%%%%%%%%%%%%%%%%%%%%%%%%
%%%%%%%%%%%%%%%%%%%%%%%%%%%%%%%%%%%%%%%%%%
\section{Proof of Theorem \ref{th:conversion}}
\label{app:proof_conv}

We prove Theorem \ref{th:conversion} constructively, presenting a protocol that transforms  $\kett{\mathbb w}$ into $\kett{\mathbb w'}$ [both generalized quantum switches obeying Eqs.(\ref{eq:generalizedQS},\ref{eq:constraints})], which can be decomposed into four steps:
\begin{enumerate}
	\item \label{it:unitAB}Apply certain unitaries to the labs inputs and outputs.
	\item \label{it:unitC}Apply a unitary on the control qubit mapping $\{\ket{\Phi_0}_C,\ket{\Phi_1}_C\}$ into $\{\ket{\Phi'_0}_C,\ket{\Phi'_1}_C\}$.
	\item \label{it:nondem} Make a non-demolition measurement on the control qubit to skew $\boldsymbol p$ into $\boldsymbol p'_{\rm id}=\boldsymbol p'=\{p_0',p_1'\}$.
	\item \label{it:flipswap}This measurement may incorrectly turn $\boldsymbol p$ into $\boldsymbol p'_{\rm sw}:=\{p_1',p_0'\}$ instead.
However, this is heralded by the measurement outcome, conditioned on which a correction is applied: the control qubit is flipped and the labs are swapped.	
	\end{enumerate}

\begin{proof}
For step \ref{it:unitAB}, we apply
\begin{multline}
\kett{\mathbb u_{PA}'\mathbb u_{PA}^{\dagger}}_{A_IA_I'}\ \kett{\mathbb u_{PB}'\mathbb u_{PB}^{\dagger}}_{B_IB_I'}\otimes \\ \otimes \kett{\mathbb u_{AF}^{\dagger}\mathbb u_{AF}'}_{A_O'A_O}\ \kett{\mathbb u_{BF}^{\dagger}\mathbb u_{BF}'}_{B_O'B_O} \ ,
\label{eq:unitaries_dagger}
\end{multline}
transforming $\mathbb u_0$, $\mathbb u_1$ into $\mathbb u'_0$, $\mathbb u'_1$, where $\mathbb u_{PA}^{\prime\dagger}\mathbb u_{BA}^{\prime}\mathbb u_{BF}^{\prime\dagger}=\mathbb1=\mathbb u_{PB}^{\prime\dagger}\mathbb u_{AB}^{\prime}\mathbb u_{AF}^{\prime\dagger}$ has been used. These constraints reflect the fact that the freedom to pick four local unitaries is not sufficient to attain all six unitaries in Eq.\eqref{eq:generalizedQS}.
In fact, the reader can verify that, starting from a superposition of $\ket{\Phi_0}_C\kett{\mathbb 1_0}$ and $\ket{\Phi_1}_C\kett{\mathbb 1_1}$, Eq.\eqref{eq:constraints} is necessary and sufficient for the target process to be attainable by local unitaries.

Steps \ref{it:unitC}-\ref{it:flipswap} can be encapsulated in a single transformation.
The majorization relation $\boldsymbol p\preccurlyeq\boldsymbol p'$ implies \cite{Marshall2011,Winter2016}
\begin{equation}
\boldsymbol p = \sum_\pi \lambda_\pi \boldsymbol p'_{\pi}
\label{eq:convexbymajor}
\end{equation}
where $\boldsymbol p'_{\pi}$ is the $\pi$-th permutation of $\boldsymbol p'$, with $\pi={\rm id}$ representing the identity and $\pi={\rm sw}$ the swap, and $\boldsymbol\lambda\coloneqq\{\lambda_\pi\}_{\pi}$ is a probability distribution on $\pi$.
The transformation will be decomposed as in Eq.~\eqref{eq:Vdecomposed} (for $\pi$ playing the role of $j$) with $V_C^{(\pi)}\coloneqq\kett{\mathbb v_C^{(\pi)}}\bbra{\mathbb v_C^{(\pi)}}$, being
\begin{equation}
\kett{\mathbb v_C^{(\pi)}} = \sqrt{\lambda_{\pi}} \sum_i \sqrt{\frac{p'_{\pi(i)}}{p_i}} \ket{\Phi_i}_{C}\ket{\Phi'_{\pi(i)}}_{C'}  \ ,
\label{eq:VCjconversion}
\end{equation}
and
\begin{equation}
V_{AB}^{({\rm id})}\coloneqq\kettbbra{\mathbb 1_{AB}}{\mathbb 1_{AB}} \ , \ \ V_{AB}^{({\rm sw})}\coloneqq\kettbbra{\mathbb s_{AB}}{\mathbb s_{AB}} \ ,
\label{eq:VABjconversion}
\end{equation}
for $\kett{\mathbb 1_{AB}}$ and $\kett{\mathbb s_{AB}}$ given by Eq.~\eqref{eq:def_ILS}. In addition, the short-hand notation $\pi(i)\coloneqq i$, for $\pi={\rm id}$, and $\pi(i)\coloneqq i+1 \mod 2$, for $\pi={\rm sw}$, has been used. The unitary in step \ref{it:unitC} always exists because $\{\ket{\Phi_j}\}_{j=0,1}$ and $\{\ket{\Phi_j'}\}_{j=0,1}$ are orthonormal bases. The measurement in Step \ref{it:nondem} is a two-outcome POVM on $\mathsf H_C$ that skews the flatter $\boldsymbol p$ distribution into either $\boldsymbol p'$ or $\boldsymbol p'_{\rm sw}$ (less balanced) while preserving the quantum superposition [the general form of such POVM is shown in Eq.\eqref{eq:skewer}].
The conditional lab swap (Step \ref{it:flipswap}) belongs to $\mathsf{PLS}$ and occurs together with a simple bit-flip on the control qubit in the $\{\ket{\Phi_j^{\prime}}\}_{j=0,1}$ basis. 
 Normalization follows from the majorization relation \eqref{eq:convexbymajor} .
Applied to $\kett{\mathbb w}$, the outcome for each $\pi$ is pure:
\begin{align}
\pi={\rm id},& \ \ \sqrt{\lambda_{\rm id}} \label{eq:Kid}
\left(\sqrt{p'_0}\ket{\Phi'_0}\kett{\mathbb u'_0}+\sqrt{p'_1}\ket{\Phi'_1}\kett{\mathbb u'_1}\right) \\
\pi={\rm sw},& \ \ \sqrt{\lambda_{\rm sw}} \label{eq:Ksw}
\left(\sqrt{p'_1}\ket{\Phi'_1}\kett{\mathbb u'_1}+\sqrt{p'_0}\ket{\Phi'_0}\kett{\mathbb u'_0}\right) ,
\end{align}
where we have used $\kett{\mathbb w_{0,1}}$ from Eqs.(\ref{eq:def_id0},\ref{eq:def_id1}). Both results are proportional to $\kett{\mathbb w'}$. A sum over $\pi$, with $\lambda_{\rm id}+\lambda_{\rm sw}=1$, yields this process exactly.
Except for the lab swap which belongs to $\mathsf{PLS}$, all transformations belong to $\mathsf{LOAE}$.
\end{proof}

%%%%%%%%%%%%%%%%%%%%%%%%%%%%%%%%%%%%%%%%%%
%%%%%%%%%%%%%%%%%%%%%%%%%%%%%%%%%%%%%%%%%%
\section{Probabilistic conversion and distillation}
\label{app:distill_proof}

The distillation protocol hinges on Lemma \ref{th:prob_conversion}, proven here by construction. 

\begin{proof}
The first step is to convert $\kett{\mathbb w}$ into a pure process $\kett{\mathbb w_{\rm aux}}$ with $\boldsymbol p_{\rm aux}\preccurlyeq\boldsymbol p'$ (for best success probability, we also have $\boldsymbol p_{\rm aux}\succcurlyeq\boldsymbol p'$ --- whether $\boldsymbol p_{\rm aux}=\boldsymbol p'$ or $\boldsymbol p_{\rm aux}=\boldsymbol p'_{\rm sw}$ depends on the specific processes $\kett{\mathbb w}$ and $\kett{\mathbb w'}$). This is achieved with a simple local-filtering measurement on Charlie's qubit, which can be written formally as the POVM
\{$M_C^{(j)}(x,y)\}_{j=0,1}=\{\kett{\mathbb m_C^{(j)}}\bbra{\mathbb m_C^{(j)}}\}_{j=0,1}$, where
\begin{equation}\begin{split}
\kett{\mathbb m_C^{(0)}}= \ \ \ \ \ \ \ \sqrt{x}\ket{\Phi_0}_C\ket{\Phi_0}_{C'}+&\ \ \ \ \ \ \ \ \sqrt {y}\ket{\Phi_1}_C\ket{\Phi_1}_{C'} \\
\kett{\mathbb m_C^{(1)}}= \sqrt{1-x}\ket{\Phi_0}_C\ket{\Phi_0}_{C'}+& \sqrt{1-y} \ket{\Phi_1}_C\ket{\Phi_1}_{C'} \ ,
\label{eq:skewer}
\end{split}\end{equation}
 with 
\begin{align}
x= &\textstyle\min\left\{\frac{p_1}{p_0}\frac{\max\{p_0',p_1'\}}{\min\{p_0',p_1'\}},1\right\} \label{eq:probconversion1} \ , \\
y= &\textstyle\min\left\{\frac{p_0}{p_1}\frac{\max\{p_0',p_1'\}}{\min\{p_0',p_1'\}},1\right\} \label{eq:probconversion2} \ .
\end{align}  
Outcome $j=0$, which corresponds to a rank-2 POVM element and occurs with probability $p_{\rm success}=\frac{\min\{p_0,p_1\}}{\min\{p_0',p_1'\}}$, leads to $\kett{\mathbb w_{\rm aux}}$ as desired. Outcome $j=1$, which necessarily corresponds to a rank-1 POVM element, indicates a failed result of the local filtering.
The success probability of this step (and of the overall procedure) is $p_{\rm success}=\frac{\min\{p_0,p_1\}}{\min\{p_0',p_1'\}}$.
If successful, on the process $\kett{\mathbb w_{\rm aux}}$ we apply the deterministic protocol from Theorem \ref{th:conversion} to produce $\kett{\mathbb w'}$. 
\end{proof}

To prove Corollary \ref{th:distillation}, we need to apply this probabilistic protocol, with $\kett{\mathbb w_{\rm qs}}$ as target, onto $N$ copies of $\kett{\mathbb w}$ in parallel. In the asymptotic limit of $N\to\infty$ \cite{Ash1965,Cover1991}, the result, with probability tending to one, is to have $N\,p_{\rm success}$ copies of $\kett{\mathbb w_{\rm qs}}$. Since the success/failure is heralded, one can pick the successful copies, distilling causal nonseparability with rate $r=p_{\rm success}=2\min\{p_0,p_1\}$, where we have used that $\boldsymbol p_{\rm qs}\coloneqq\{\frac12,\frac12\}$.

\subsection{Distillation with multicopy instruments}
\label{sec:multicopy_dist}

We now present an alternative distillation protocol making use of multicopy operations, inspired on the coherence-distillation protocol developed in \cite{Yuan2015,Winter2016}. Because of the limitations to act jointly on different copies of Alice's and Bob's labs, joint operations are used only on Charlie's control qudit. This is capable of obtaining distillation, albeit at a rate $r=\min\{p_0,p_1\}$, lower than that in Corollary \ref{th:distillation}.
This protocol distills $\kett{\mathbb w}$ [Eq.\eqref{eq:generalizedQS}] into $\kett{\mathbb w_{\rm qs}}$, and is composed of four steps.
\begin{enumerate}
	\item \label{it:unitariesABC}Application of local unitaries to the inputs and outputs of the labs, and to Charlie's control qubit, each acting on a single copy of the process. 
	\item \label{it:typeclassmeasmt} Measurement of the total number of qubit flips on the control qubits. This is a joint measurement on all control qubits, and generates (with high probability for large $N$) equally balanced superpositions of different causal orders.
	\item \label{it:subprojector}Subnormalized projector measurements on the control qubits. Also a joint measurement on all control qubits, its random outcome determines which copies will be turned into $\kett{\mathbb w_{\rm qs}}$ or not.
	\item \label{it:disentangle} Operations to disentangle certain copies from the others. 
	\end{enumerate}

\begin{proof}
The unitaries in Step \ref{it:unitariesABC} are the same as in the proof of Theorem \ref{th:conversion}, with $\kett{\mathbb u'_0}=\kett{\mathbb 1_0}$, $\kett{\mathbb u'_1}=\kett{\mathbb 1_1}$ [Eq.\eqref{eq:unitaries_dagger}] and $\ket{\Phi'_j}_{C'}=\ket j_{C'}$.
Step \ref{it:typeclassmeasmt} is a projective measurement on collective subspaces with a well-defined total number of $\ket0$'s and $\ket1$'s ($N\!-\!j$ and $j$, respectively). They are referred to as type-class measurements in \cite{Winter2016} and correspond to measuring the Hamming norm on a string in the computational basis. The measurement entangles the different copies, and this is the step that creates a balanced superposition of terms. For outcome $j$, the resulting process is
\begin{equation}
\sum_{\pi\in R_j} \frac{\pi\left[\kett{0_C\mathbb1_0}^{\otimes N\!-\!j}\!\otimes\!\kett{1_C\mathbb1_1}^{\otimes j}\right] }{ \sqrt{\binom{N}{j}}} \ ,
\label{eq:posttypeclass} 
\end{equation}
where $R_j$ is the set of permutations of the factors in parentheses, with $\|R_j\|=\binom{N}{j}$. For $N\to\infty$, the typical result is $j=Np_1$ \cite{Ash1965}.

Step \ref{it:subprojector} is meant to reduce the number of terms in the superposition to $2^k$  in order to obtain an exact number of copies of $\kett{\mathbb w_{\rm qs}}^{\otimes k}$. This is once again done with measurements on the control qubits. To avoid measurement outcomes that lead to a failure, the second POVM $\{E_{j,\ell}^N\}_{\ell\in\{1,L\}}$ is
\begin{equation}
E^N_{j,\ell}=\frac1{\sqrt n}\sum_{\pi\in R_{j,\ell}} \pi \left[\left(\ket0\bra0_C\right)^{\otimes N\!-\!j}\!\otimes\!\left(\ket1\bra1_C\right)^{\otimes j}\right] \ ,
\label{eq:subproj}
\end{equation}
where $R_{j,\ell}\subseteq R_j$ are sets (composed of $\|R_{j,\ell}\|=2^k$ elements) with typically non-empty intersections such that every element of $R_j$ belongs to exactly $n$ of such sets. 
 This can be done with $n={\rm LCM}[2^k,\binom Nj]/\binom Nj$, and with $L={\rm LCM}[2^k,\binom Nj]/2^k$ sets $R_{j,\ell}$, where ${\rm LCM}$ denotes the least common multiple. 
 The resulting state after outcome $\ell$ is
\begin{equation}
 \frac1{\sqrt{2^k}}\sum_{\pi\in R_{j,\ell}}\pi\left[\kett{0_C\mathbb1_0}^{\otimes N\!-\!j}\!\otimes\!\kett{1_C\mathbb1_1}^{\otimes j}\right]  \ .
\label{eq:postsubproj}
\end{equation}
The value $k$ will define our final rate, since $2^k=2^{rN}$. Clearly $2^k\leq\binom Nj$, since $2^k$ superposed terms are left from projecting a superposition of $\binom Nj$. Moreover, a decomposition of $\kett{\mathbb w_{\rm qs}}^{\otimes k}$ contains  a term $\propto\ket0_C\kett{\mathbb1_0}^{\otimes k}$ as well as one $\propto\ket1_C\kett{\mathbb1_1}^{\otimes k}$. By comparing with the state \eqref{eq:postsubproj}, 
 we see that $k\leq\min\{j,N\!-\!j\}$. Since $2^{\min\{j,N\!-\!j\}}\leq\binom Nj$, in fact $k={\min\{j,N\!-\!j\}}$.

At this point we have $N$ entangled copies of the whole process, in a superposition of $2^k$ terms. In Step \ref{it:disentangle}, we disentangle $N-k$ copies of the process from the remaining $k$. Both the control qudits and the labs' inputs/outputs must be disentangled in this step. 
To disentangle the $i$-th control qubit from the remaining control qubits, one makes a measurement on the $\ket\pm_{C,i}$ basis. The $\ket-_{C,i}$ outcome heraldedly introduces an unwanted sign, which can be corrected through controlled phase gates on the remaining control qubits. To disentangle the $i$-th lab from the rest, we apply on that lab
\begin{align}\begin{split}
V_{\rm bypass}^{AB}& = \kett{\mathbb1}\bbra{\mathbb1}_{A_IA_O} \otimes \ket{\phi}\bra\phi_{A_I'} \otimes  \mathbb1_{A_O'} \otimes \label{eq:bypass}\\
 &\otimes \kett{\mathbb1}\bbra{\mathbb1}_{B_IB_O} \otimes \ket{\phi}\bra\phi_{B_I'} \otimes  \mathbb1_{B_O'} \ ,
\end{split}\end{align}
an operation which bypasses the actual labs, short-circuiting the $P$ signal to $F$, giving dummy inputs $\ket\phi$ to the labs and discarding the labs' outputs. The definition of which copies to disentangle (and discard) depends on outcomes $j,\ell$, i.e., depends on feed-forwarding.

The distillation rate is $r=k/N=\min\{j,N-j\}/N$, which for $N\to\infty$ tends to the typical result \cite{Ash1965} $r=\min\{p_0,p_1\}$.
\end{proof}

As an illustration, in the case of $N=4,j=2$ the projector in Step \ref{it:typeclassmeasmt} acting jointly on many control qudits is
\begin{align}
\Pi_2^4=&\ket{0011}\bra{0011}_{C}+\ket{0101}\bra{0101}_{C}+\ket{1100}\bra{1100}_{C}+ \nonumber \\
        &\ket{0110}\bra{0110}_{C}+\ket{1001}\bra{1001}_{C}+\ket{1010}\bra{1010}_{C} \ .
\label{eq:Pi24}
\end{align}
The state after this projection, Eq. \eqref{eq:posttypeclass}, accordingly reads (in a compact notation)
\begin{align}\begin{split}
\frac1{\sqrt6}\big[
&\ \ \ \ \  \kett{0_C\mathbb1_0} \kett{0_C\mathbb1_0} \kett{1_C\mathbb1_1} \kett{1_C\mathbb1_1}\\
&+\kett{0_C\mathbb1_0} \kett{1_C\mathbb1_1} \kett{0_C\mathbb1_0} \kett{1_C\mathbb1_1}\\
&+\kett{1_C\mathbb1_1} \kett{1_C\mathbb1_1} \kett{0_C\mathbb1_0} \kett{0_C\mathbb1_0}\\
&+\kett{0_C\mathbb1_0} \kett{1_C\mathbb1_1} \kett{1_C\mathbb1_1} \kett{0_C\mathbb1_0}\\
&+\kett{1_C\mathbb1_1} \kett{0_C\mathbb1_0} \kett{0_C\mathbb1_0} \kett{1_C\mathbb1_1}\\
&+\kett{1_C\mathbb1_1} \kett{0_C\mathbb1_0} \kett{1_C\mathbb1_1} \kett{0_C\mathbb1_0}\big] \ .
\end{split}\label{eq:postmeas24}
\end{align}
In Step \ref{it:subprojector}, $k=2$ copies can be obtained. The POVM Eq. \eqref{eq:subproj}, will be composed of $L=3$ elements, and each projector appears in $n=2$ different elements:
\begin{widetext}
\begin{subequations}\begin{align}
E^4_{2,1}=& \left( \ket{0011}\bra{0011}_C +\ket{0101}\bra{0101}_C +\ket{1100}\bra{1100}_C +\ket{0110}\bra{0110}_C\right)/\sqrt2 , 
\label{eq:POVM24a}\\
E^4_{2,2}=& \left( \ket{1001}\bra{1001}_C +\ket{1010}\bra{1010}_C +\ket{1100}\bra{1100}_C +\ket{0110}\bra{0110}_C\right)/\sqrt2 , 
\label{eq:POVM24b}\\ 
E^4_{2,3}=& \left( \ket{1001}\bra{1001}_C +\ket{1010}\bra{1010}_C +\ket{0011}\bra{0011}_C +\ket{0101}\bra{0101}_C\right)/\sqrt2 . 
\label{eq:POVM24c} 
\end{align}\label{eq:POVM24}\end{subequations}
\end{widetext}

If e.g.\ outcome $(j,\ell)=(2,2)$ were obtained, the state would become a balanced superposition of the last four terms of Eq. \eqref{eq:postmeas24}. We would then keep the second and third copies, since these appear in all 2-bit combinations $(00,01,10,11)$, and discard the remaining two. The disentangling operations would be applied to the copies $i=1$, $i=4$, yielding the process
\begin{equation}
 W_1\otimes \kettbbra{\mathbb w_{\rm qs}}{\mathbb w_{\rm qs}}^{\otimes2} \otimes W_4 \ ,
\label{eq:24distilled}
\end{equation} 
where $W_i$ are obtained from Eq. \eqref{eq:bypass}. 
 Other outcomes of the POVM \eqref{eq:POVM24} require discarding different systems, illustrating the need for a feed-forward [\eqref{eq:POVM24a} leads to discarding the first two, \eqref{eq:POVM24c}, the second and fourth].

}\end{document}